%% file: ms.tex
\documentclass{article}

\usepackage{microtype}
\usepackage{graphicx}
\usepackage{subfigure}
\usepackage{booktabs} 

\usepackage{hyperref}

\usepackage{aistats2019}

\usepackage{times}

\usepackage{graphicx}
\usepackage{tikz}

\usetikzlibrary{fit,positioning,arrows,automata,calc}
\tikzset{
	main/.style={circle, minimum size = 5mm, thick, draw =black!80, node distance = 10mm},
	connect/.style={-latex, thick},
	box/.style={rectangle, draw=black!100}
}

\usepackage{epigraph}

\usepackage{csquotes}
\usepackage{hyperref}       
\usepackage{url}    
\usepackage{graphicx}   
\usepackage{relsize}     
\usepackage{booktabs}       
\usepackage{amsfonts}       
\usepackage{nicefrac}       
\usepackage{microtype}      
\usepackage{color}
\usepackage{scalerel}
\usepackage[toc,page]{appendix}
\usepackage{blindtext}
\usepackage{amssymb}
\usepackage{mathtools}
\usepackage{dsfont}
\usepackage{amsmath}
\usepackage{amsfonts}
\usepackage{amsthm}
\usepackage{float}
\usepackage{graphicx}
\graphicspath{ {images/} }
\usepackage{letltxmacro}%
\usepackage{thmtools,thm-restate}
\usepackage{relsize}
\usepackage{algorithm}
\usepackage{algorithmic}
\usepackage{cuted}

\DeclarePairedDelimiterX\Basics[1](){ #1}

\DeclarePairedDelimiterX{\infdivx}[2]{(}{)}{#1\;\delimsize\|\;#2}

\DeclareMathOperator*{\argmin}{arg\,min}

\makeatletter
\newcommand{\distas}[1]{\mathbin{\overset{#1}{\kern\z@\sim}}}%
\newsavebox{\mybox}\newsavebox{\mysim}
\newcommand{\distras}[1]{%
	\savebox{\mybox}{\hbox{\kern3pt$\scriptstyle#1$\kern3pt}}%
	\savebox{\mysim}{\hbox{$\sim$}}%
	\mathbin{\overset{#1}{\kern\z@\resizebox{\wd\mybox}{\ht\mysim}{$\sim$}}}%
}
\makeatother

\usepackage{enumitem}
\newlist{inparaenum}{enumerate}{2}
\setlist[inparaenum]{nosep}
\setlist[inparaenum,1]{label=\bfseries\arabic*.}
\setlist[inparaenum,2]{label=\arabic{inparaenumi}\emph{\alph*})}

\usepackage{bm}
\usepackage{breqn}
\usepackage{physics}
\usepackage{enumitem}
\usepackage{bbold}
\usepackage[colorinlistoftodos]{todonotes}

\newtheorem{cor}{Corollary}[section]
\newtheorem{prop}{Proposition}[section]

\newtheorem{assumption}{Assumption}
\newtheorem{definition}{Definition}
\newtheorem{remark}{Remark}

\newtheorem{thm}{Theorem}

\newcommand{\Wc}{\mathcal{W}}

\newcommand{\x}{X^{(1)}}
\newcommand{\zx}{Z^{(1)}}
\newcommand{\y}{X^{(2)}}
\newcommand{\zy}{Z^{(2)}}
\newcommand{\n}{\eta^{(1)}}
\newcommand{\w}{\eta^{(2)}}

\newcommand{\bcF}{\bm{\mathcal{F}}}

\newcommand{\Nc}{\bm{\mathcal{N}}}

\newcommand{\Xb}{\textbf{X}}

\newcommand{\bl}{\Bigg |}
\newcommand{\Ex}{\mathbb{E}}
\newcommand{\Pb}{\mathbb{P}}
\newcommand{\Rb}{\mathbb{R}}
\newcommand{\Sc}{\mathcal{S}}
\newcommand{\Ec}{\mathcal{E}}
\newcommand{\xe}{X^{e}}
\newcommand{\xs}{X^{s}}
\newcommand{\xms}{X^{ms}}
\newcommand{\xmss}{X^{mss}}
\newcommand{\nee}{\eta^{e}}
\newcommand{\ns}{\eta^{s}}
\newcommand{\nms}{\eta^{ms}}
\newcommand{\tX}{\tilde{X}}
\newcommand{\tn}{\tilde{\eta}}

\usepackage{enumitem}
\setlist{nolistsep}

\newcommand{\la}{\langle}
\newcommand{\ra}{\rangle}
\newcommand{\ztp}{Z_T^{\perp}}
\usepackage[compact]{titlesec}         
\titlespacing{\section}{0pt}{0pt}{0pt}
\usepackage{etoolbox}
\AfterEndEnvironment{strip}{\leavevmode}
\begin{document}
	\twocolumn[
	
	\aistatstitle{Near optimal finite time identification of arbitrary linear dynamical systems}
	
\aistatsauthor{ Tuhin Sarkar \And Alexander Rakhlin }

\aistatsaddress{ MIT \And MIT } ]

	\begin{abstract}
We derive finite time error bounds for estimating general linear time-invariant (LTI) systems from a single observed trajectory using the method of least squares. We provide the first analysis of the general case when eigenvalues of the LTI system are arbitrarily distributed in three regimes:  stable, marginally stable, and explosive. Our analysis yields sharp upper bounds for each of these cases separately. We observe that although the underlying process behaves quite differently in each of these three regimes, the systematic analysis of a self--normalized martingale difference term helps bound identification error up to logarithmic factors of the lower bound. On the other hand, we demonstrate that the least squares solution may be statistically inconsistent under certain conditions even when the signal-to-noise ratio is high. 
\end{abstract}

	
	\smallskip

	\input{content/intro}
	\input{content/contributions}	
	\input{content/model}	
	\input{content/main_results}

	\input{content/inconsistent}

	\input{content/discussion}
	\bibliographystyle{alpha}
	\bibliography{bibliography}

	\input{content/appendix_matrix}
	\input{content/appendix_prob}

	\input{content/results_stable}
	\input{content/sharp_bnds}
	\input{content/results_explosive}
	\input{content/regularity}
	\input{content/composite}
	\input{content/misc}

\end{document}

%% file: content/intro.tex
\section{Introduction}
\label{intro}
Finite time system identification---the problem of estimating the parameters of an unknown dynamical system given a finite time series of its output---is an important problem in the context of time-series analysis, control theory, economics and reinforcement learning. In this work we will focus on obtaining sharp non--asymptotic bounds for \textit{linear} dynamical system identification using the ordinary least squares (OLS) method. Such a system is described by $X_{t+1} = AX_t + \eta_{t+1}$ where $X_t \in \Rb^d$ is the state of the system and $\eta_t$ is the unobserved process noise. The goal is to learn $A$ by observing only $X_t$'s. Our techniques can easily be extended to the more general case when there is a control input $U_t$, \textit{i.e.}, $X_{t+1} = AX_{t} + BU_t + \eta_{t+1}$. In this case $(A, B)$ are unknown, and we can choose $U_t$.

Linear systems are ubiquitous in control theory. For example, proportional-integral-derivative (PID) controller is a popular linear feedback control system found in a variety of devices, from planetary soft landing systems for rockets (see e.g.~\cite{accikmecse2013lossless}) to coffee machines. Further, linear approximations to many non--linear systems have been known to work well in practice. Linear systems also appear as auto--regressive (AR) models in time series analysis and econometrics. Despite its importance, sharp non--asymptotic characterization of identification error in such models was relatively unknown until recently.

In the statistics literature, correlated data is often dealt with using mixing--time arguments (see e.g. \cite{yu1994rates}). 
However, a fundamental limitation of the mixing-time method is that bounds deteriorate when the underlying process mixes slowly. For discrete linear systems, this happens when $\rho(A)$---the spectral radius of $A$---approaches $1$. As a result these methods cannot extend to the case when $\rho(A) \geq 1$. More recently there has been renewed effort in obtaining sharp non--asymptotic error bounds for linear system identification~\cite{faradonbeh2017finite,simchowitz2018learning}. Specifically,~\cite{faradonbeh2017finite} analyzed the case when the system is either stable ($\rho(A) < 1$) or purely explosive ($\rho(A) > 1$). For the case when $\rho(A) < 1$ the techniques in \cite{faradonbeh2017finite} are similar to the standard mixing time arguments and, as a result, suffer from the same limitations. When the system is purely explosive, the authors of \cite{faradonbeh2017finite} show that finite time identification is only possible if the system is regular, \textit{i.e.}, if the geometric multiplicity of eigenvalues greater than unity is one. However, as discussed in~\cite{simchowitz2018learning}, the bounds obtained in~\cite{faradonbeh2017finite} are suboptimal due to a decoupled analysis of the sample covariance, $\sum_{t=1}^T X_tX_t^{\prime}$, and the martingale difference term $\sum_{t=1}^T X_t \eta_{t+1}'$. A second approach, based on Mendelson's small--ball method, was studied in~\cite{simchowitz2018learning}. Such a technique eschewed the need for mixing-time arguments and sharper error bounds for $1 - C/T \leq \rho(A) \leq 1 + C/T$ could be obtained. The authors in~\cite{simchowitz2018learning} argue that a larger signal-to-noise ratio, measured by $\lambda_{\min}(\sum_{t=0}^{T-1} A^{t}A^{t \prime})$, makes it easier to estimate $A$.  Although this intuition is consistent for the case when $\rho(A) \leq 1$, it does not extend to the case when eigenvalues are far outside the unit circle. Since $X_T = \sum_{t=1}^T A^{T-t} \eta_{t}$, the behavior of $X_T$ is dominated by $\{\eta_1, \eta_2, \ldots \}$, \textit{i.e.}, the past, due to exponential scaling by $\{A^{T-1}, A^{T-2},\ldots\}$. As a result, $X_1$ depends strongly on $\{X_2, \ldots, X_T\}$ and standard techniques of creating ``independent'' blocks of covariates fail. 

The problem of system identification has received a lot of attention. Asymptotic results on identification of AR models can be found in~\cite{lai1983asymptotic}. Some of the earlier work on finite time identification in systems theory include~\cite{campi2002finite,vidyasagar2006learning}. A more general setting of the problem considered here is when $X_t$ is observed indirectly via its filtered version, \textit{i.e.}, $Y_t = CX_t$ where $C$ is unknown. The single input single output (SISO) version of this problem, \textit{i.e.}, when $Y_t, U_t$ are numbers, has been studied in~\cite{hardt2016gradient} under the assumption that system is stable. Provable guarantees for system identification in general linear systems was also studied in~\cite{oymak2018non}. However, the analysis there requires that $||A|| < 1$. Generalization bounds for time series forecasting of non--stationary and non--mixing processes have been developed in~\cite{forecasting_mohri}.  

%% file: content/contributions.tex
\section{Contributions}
\label{contributions}
In this paper we offer a new statistical analysis of the ordinary least squares estimator of the dynamics $X_{t+1} = A X_t + \eta_{t+1}$ with no inputs. Unlike previous work, we do not impose any restrictions on the spectral radius of $A$ and provide nearly optimal rates (up to logarithmic factors) for every regime of $\rho(A)$. The contributions of our paper can be summarized as follows
\begin{itemize}
	\item  At the center of our techniques is a systematic analysis of the sample covariance $\sum_{t=1}^T X_t X_t^{\prime}$ and a certain self normalized martingale difference term. Although such a coupled analysis is similar in flavor to~\cite{simchowitz2018learning}, it comes without the overhead of choosing a block size and applies to a general case when covariates grow exponentially in time. 
	 	
	\item Specifically, for the case when $\rho(A) \leq 1$, we recover the optimal finite time identification error rates previously derived in~\cite{simchowitz2018learning}. For the case when all eigenvalues are outside the unit circle, we argue that small ball methods cannot be used. Instead we use anti--concentration arguments discussed in~\cite{faradonbeh2017finite,lai1983asymptotic}. By leveraging subgaussian tail inequalities we sharpen previous error bounds by removing polynomial factors. We also show that this analysis is indeed tight by deriving a matching lower bound. 
	
	\item We provide the first analysis of the general case when eigenvalues of $A$ are arbitrarily distributed in three regimes: stable, marginally stable and explosive. This involves a careful analysis of the noise-covariate cross terms as the underlying process behaves differently in each of these regimes.

	\item We show that when $A$ does not satisfy certain regularity conditions, OLS identification is statistically inconsistent, even when signal-to-noise ratio is high. Our result indicates that consistency of OLS identification depends on the condition number of the sample covariance matrix, rather than the signal-to-noise ratio itself.

\end{itemize}

%% file: content/model.tex
\section{Notation and Definitions}

A linear time invariant system (LTI) is parametrized by a matrix, $A$, where the observed variable, $X_t$, indexed by $t$ evolves as
\begin{equation}
X_{t+1} = AX_t + \eta_{t+1}. \label{lti}
\end{equation}
Here $\eta_t$ is the noise process. 
Denote by $\rho_i(A)$ the absolute value of the $i^{th}$ eigenvalue of the $d \times d$ matrix $A$. Then 
\[
\rho_{\max}(A) = \rho_1(A) \geq \rho_2(A) \geq \hdots \geq \rho_{d}(A) = \rho_{\min}(A).
\]
Similarly the singular values of $A$ are denoted by $\sigma_i(A)$. For any matrix $M$, $||M||_{\text{op}} = ||M||_2$.
\begin{definition}
	\label{stable}
A stable LTI system is that where $\rho_{\max}(A) < 1$. An explosive LTI system is that where $\rho_{\min}(A) > 1$.
\end{definition}
For simplicity of exposition, we assume that $X_0 = 0$ with probability $1$. All the results can be obtained by assuming $X_0$ to be some bounded vector. 
\begin{definition}
    \label{isotropic}
A random vector $X \in \Rb^{d}$ is called isotropic if for all $x \in \Rb^d$ we have
\[
\Ex \langle X, x \rangle^2 = ||x||^2_2 
\]
\end{definition}
\begin{assumption}
	\label{subgaussian_noise}
	$\{\eta_t\}_{t=1}^{\infty}$ are i.i.d isotropic subgaussian and coordinates of $\eta_t$ are i.i.d. Further, let $f(x)$ be the pdf of each noise coordinate then the essential supremum of $f(\cdot)$ is bounded above by $C < \infty$.  
\end{assumption}

	We will deal with only regular systems, \textit{i.e.}, LTI systems where eigenvalues of $A$ with absolute value greater than unity have geometric multiplicity one. We will show that when $A$ is not regular, OLS is statistically inconsistent.

Define the data matrix $\Xb$ and the noise matrix $E$ as  
\[
\Xb =\begin{bmatrix}
 X_0^{\prime} \\ X^{\prime}_1 \\ \vdots \\X_{T}^{\prime}
\end{bmatrix},
~~~E =\begin{bmatrix}
\eta_1^{\prime} \\ \eta^{\prime}_2 \\ \vdots \\\eta_{T+1}^{\prime},
\end{bmatrix}
\]
where the superscript $a^{\prime}$ denotes the transpose.
Then $\Xb$, $E$ are $(T+1) \times d$ matrices. Consider the OLS solution
\begin{equation*}
\hat{A} = \argmin_{B} \sum_{t=0}^{T}||X_{t+1} - BX_{t}||^2_2.
\end{equation*}
One can show that 
\begin{equation}
\label{error_lse}
A - \hat{A} = ((\Xb' \Xb)^{+} \Xb^{\prime} E)^{\prime}
\end{equation}
where $M^{+}$ is the pseudo inverse of M. We define
\begin{equation*}
Y_T = \Xb^{\prime} \Xb = \sum_{t=0}^{T} X_t X_t^{\prime},~~~~ S_T = \Xb^{\prime} E = \sum_{t=0}^{T} X_t \eta_{t+1}^{\prime}.
\end{equation*}
To analyze the error in estimating $A$, we will aim to bound the norm of $(\Xb^{\prime} \Xb)^{+} \Xb^{\prime}$. 
\begin{table*}
		\begin{center}
		\begin{tabular}{|l|}
		\hline
		$T_{\eta}(\delta) = C\Big(\log{\frac{2}{\delta}} + d \log{5}\Big)$\\
		$T_{s}(\delta) = C\Big({ d \log{( \text{tr}(\Gamma_T(A))+1)} + 2d \log{\frac{5}{\delta}} }\Big)$ \\
		$c(A, \delta) = T_{s}(\frac{2\delta}{3T})$\\
		$\beta_0(\delta) = \inf{\Big\{\beta|\beta^2\sigma_{\min}(\Gamma_{\lfloor \frac{1}{\beta}\rfloor}(A)) \geq \Big(\frac{ 16ec(A, \delta)}{ T\sigma_{\min}(A A^{\prime})}\Big)\Big\}}$\\
		$T_{ms}(\delta) = \inf{\Big\{T \Big| T \geq \frac{Cc(A, \delta)}{ \sigma_{\min}(A A^{\prime})}\Big\}}$\\ 
		$T_{u}(\delta)={\Big\{T \Big| \Big(4T^2 \sigma_1^2(A^{-\lfloor \frac{T+1}{2} \rfloor}) \text{tr}(\Gamma_T(A^{-1})) + \frac{T\text{tr}(A^{-T-1}\Gamma_T(A^{-1})A^{-T-1 \prime})}{\delta}\Big) \leq \frac{\phi_{\min}(A)^2 \psi(A)^2 \delta^2}{2\sigma_{\max}(P)^2} \Big\}}$ \\
		$\gamma(A, \delta)=\frac{4 \phi_{\max}(A)^2 \sigma_{\max}^2(A)}{\phi_{\min}(A)^2 \sigma_{\min}^2(A) \psi(A)^2 \delta^2} (1+\frac{1}{c}\log{\frac{1}{\delta}})\text{tr}(P(\Gamma_T(A^{-1}))P^{\prime})I$ \\ 
		$\gamma_s(A, \delta) = \sqrt{8d \Big(\log{\Big(\frac{5}{\delta}\Big) + \frac{1}{2}\log{\Big(4\text{tr}(\Gamma_T(A)) + 1 \Big)}}\Big)}$\\
		$\gamma_{ms}(A, \delta) = \sqrt{16 d \log{(\text{tr}(\Gamma_T(A)) + 1)} + 32d \log{\Big(\frac{15T}{2\delta}\Big)}}$\\
		$\gamma_e(A, \delta) = \frac{\sqrt{d}\sigma_{\max}(P) }{\phi_{\min}(A) \psi(A)\delta}\sqrt{ \log{\frac{2}{\delta}} + 2 \log{5} + \log {(1 + \gamma(A, \delta))}}$\\
		\hline
		\end{tabular}
		\caption{Notation} \label{notation}
		\end{center}
\end{table*}

We will occasionally replace $X_t$ (or $X(t)$) with the lower-case counterparts $x_t$ (or $x(t)$) to denote state at time $t$, whenever this does not cause confusion. Further, we will use $C,c$ to indicate universal constants that can change from line to line.
Define the \emph{Gramian} as
\begin{equation}
\label{gramian}
\Gamma_t(A) = \sum_{k=0}^t A^k A^{k\prime}
\end{equation}
and a Jordan block matrix $J_d(\lambda)$ as
	\begin{equation}
		\label{jordan}
	J_d(\lambda) =\begin{bmatrix}
	\lambda & 1 &  0 & \hdots & 0 \\
	0 & \lambda & 1 & \hdots & 0 \\
	\vdots & \vdots & \ddots & \ddots & \vdots \\
	0 & \hdots & 0 & \lambda & 1 \\
	0 & 0 & \hdots & 0 & \lambda  
	\end{bmatrix}_{d \times d}
	\end{equation}
We present the three classes of matrices that will be of interest to us:
\begin{itemize}
	\item The perfectly stable matrix class, $\Sc_0$ 
	$$\rho_{i}(A) \leq 1 - \frac{C}{T}$$
		for $1 \leq i \leq d$.
	\item The marginally stable matrix, $\Sc_1$ 
	$$1 - \frac{C}{T} < \rho_i(A) \leq 1+\frac{C}{T}$$ 
	for $1 \leq i \leq d$.
	\item The regular and explosive matrix, $\Sc_2$ 
	$$\rho_i > 1 + \frac{C}{T}$$
	for $1 \leq i \leq d$.
\end{itemize}
Slightly abusing the notation, whenever we write $A \in \Sc_i \cup \Sc_j$ we mean that $A$ has eigenvalues in both $\Sc_i, \Sc_j$. 
Critical to obtaining refined error rates, will be a result from the theory of self--normalized martingales. We let $\bcF_t = \sigma(\eta_1, \eta_2, \ldots, \eta_t, X_1, \ldots, X_t)$ to denote the filtration generated by the noise and covariate process. 
\begin{prop}
	\label{selfnorm_bnd}
	Let $V$ be a deterministic matrix with $V \succ 0$. For any $0 < \delta < 1$ and $\{\eta_t, X_t\}_{t=1}^{T}$ defined as before, we have with probability $1 - \delta$
	\begin{align}
	&||(\bar{Y}_{T-1})^{-1/2} \sum_{t=0}^{T-1} X_t \eta_{t+1}^{\prime}||_2 \nonumber\\
	&\leq R\sqrt{8d \log {\Bigg(\dfrac{5 \text{det}(\bar{Y}_{T-1})^{1/2d} \text{det}(V)^{-1/2d}}{\delta^{1/d}}\Bigg)}}
	\end{align}
	where $\bar{Y}^{-1}_{\tau} = (Y_{\tau} + V)^{-1}$ and $R^2$ is the subGaussian parameter of $\eta_t$.
\end{prop}
The proof can be found in appendix as Proposition~\ref{selfnorm_bnd_proof}. It rests on Theorem 1 in \cite{abbasi2011improved} which is itself an application of the pseudo-maximization technique in \cite{pena2008self} (see Theorem 14.7).

Finally, we define several $A$-dependent quantities that will appear in time complexities in the next section. 
\begin{definition}[Outbox Set]
	\label{outbox}
For the space $\Rb^{d}$ define the $a$--outbox, $S_d(a)$, as the following set
\[
S_d(a) = \{v | \min_{1 \leq i \leq d} |v_i| \geq a\}
\]
$S_d(a)$ will be used to quantify the following norm--like quantities of a matrix:
\begin{align}
\phi_{\min}(A) &= \sqrt{\inf_{v \in S_d(1)} \sigma_{\min}\Big(\sum_{i=1}^T \Lambda^{-i+1} vv^{\prime} \Lambda^{-i+1 \prime}\Big)} \\
\phi_{\max}(A) &= \sqrt{\sup_{||v||_2 = 1} \sigma_{\max}\Big(\sum_{i=1}^T \Lambda^{-i+1} vv^{\prime} \Lambda^{-i+1 \prime}\Big)}\label{anticonc_norm}
\end{align}
where $A = P^{-1} \Lambda P$ is the Jordan normal form of $A$.
\end{definition}
$\psi(A)$ is defined in Proposition~\ref{anti_conc1} and is needed for error bounds for explosive matrices.
\begin{prop}[Proposition 2 in~\cite{faradonbeh2017finite}]
	\label{anti_conc1}
	Let $\rho_{\min}(A) > 1$ and $P^{-1} \Lambda P = A$ be the Jordan decomposition of $A$. Define $z_T = A^{-T}\sum_{i=1}^TA^{T-i}\eta_i$ and
	$$\psi(A, \delta) = \sup \Bigg\{y \in \Rb : \Pb\Bigg(\min_{1 \leq i \leq d}|P_i^{'}z_T| < y \Bigg) \leq \delta \Bigg\}$$
	where $P = [P_1, P_2, \ldots, P_d]^{'}$. Then 
	$$\psi(A, \delta) \geq \psi(A) \delta > 0$$
Here $\psi(A) = \frac{1}{2 d \sup_{1 \leq i \leq d}C_{|P_i^{'}z_T|}}$ where $C_{X}$ is the essential supremum of the pdf of $X$.
\end{prop}
We summarize some notation in Table~\ref{notation} for convenience in representing our results. 

%% file: content/main_results.tex
\section{Main Results}
\label{main_results}
We will first show non--asymptotic rates for the three separate regimes, followed by the case when $A$ has a general eigenvalue distribution.
\begin{thm}
	\label{main_result}
	The following non-asymptotic bounds hold, with probability at least $1-\delta$, for the least squares estimator:
	\begin{itemize}
		\item For $A \in \Sc_0 \cup \Sc_1$  
		$$||A - \hat{A}||_{2} \leq \sqrt{\frac{C}{T}}\underbrace{\gamma_s\Big(A, \frac{\delta}{4}\Big)}_{=O(\sqrt{\log{(\frac{1}{\delta}})})}$$
		whenever 
		$$T \geq \max{\Big(T_{\eta}\Big(\frac{\delta}{4}\Big), T_s\Big(\frac{\delta}{4}\Big)\Big)}$$
		\item For $A \in \Sc_1$ 
		$$||A - \hat{A}||_{2} \leq \frac{C \sigma_{\max}(A^{-1})}{\sqrt{T\sigma_{\min}(\Gamma_{\lfloor \frac{1}{\beta_0(\delta)}\rfloor}(A))}}\underbrace{\gamma_{ms}\Big(A, \frac{\delta}{2}\Big)^2}_{=O(\log{(\frac{T}{\delta})})}$$ 
		whenever
		$$T \geq \max{\Big(2T_{\eta}\Big(\frac{\delta}{3T}\Big), 2T_s\Big(\frac{\delta}{3T}\Big), T_{ms}\Big(\frac{\delta}{2}\Big)\Big)}$$
		Since $\sigma_{\min}(\Gamma_{\lfloor \frac{1}{\beta_0(\delta)}\rfloor}(A)) \geq \alpha(d)\frac{T}{\log{T}}$, we have that 
		$$||A - \hat{A}||_{2} \leq \sqrt{\frac{\log{T}}{\alpha(d)}}\frac{\gamma_{ms}\Big(A, \frac{\delta}{2}\Big)^2}{T}$$ 
		\item For $A \in \Sc_2$ 
		$$||A - \hat{A}||_{2} \leq C\sigma_{\max}(A^{-T}) \underbrace{\gamma_e\Big(A, \frac{\delta}{5}\Big)}_{=O(\frac{1}{\delta})}$$
		whenever
		$$T \in T_{u}\Big(\frac{\delta}{5}\Big)$$
		Since $\sigma_{\max}(A^{-T}) \leq \alpha(d) (\rho_{\min}(A))^{-T}$ for $A \in \Sc_2$, the identification error decays exponentially with $T$.
	\end{itemize}
	Here $C, c$ are absolute constants and $\alpha(d)$ is a function that depends only on $d$.
\end{thm}
\begin{remark}
$T_u(\delta)$ is a set where there exists a minimum $T_{*} < \infty$ such that $T \in T_u(\delta)$ whenever $T \geq T_{*}$. However, there might be $T < T_{*}$ for which the inequality of $T_{u}(\delta)$ holds. Whenever we write $T \in T_u(\delta)$ we mean $T \geq T_{*}$. 
\end{remark}
\begin{proof}
	We start by writing an upper bound 
	\begin{align}
	\label{err}
	||A - \hat{A}||_{\text{op}} &\leq ||Y_T^{+}S_T||_{\text{op}} \nonumber \\
	&\leq ||(Y_T^{+})^{1/2}||_{\text{op}}||(Y_T^{+})^{1/2}S_T||_{\text{op}}.
	\end{align}
	The rest of the proof can be broken into two parts: 
	\begin{itemize}
		\item Showing invertibility of $Y_T$ and lower bounds on the least singular value
		\item Bounding the self-normalized martingale term given by $(Y_T^{+})^{1/2}S_T$
	\end{itemize}
The invertibility of $Y_T$ is where most of the work lies. Once we have a tight characterization of $Y_T$, one can simply obtain the error bound by using Proposition~\ref{selfnorm_bnd}. Here we sketch the basis of our approach. First, we find deterministic $V_{up}, V_{dn}, T_0$ such that 
\begin{align}
\Ec_0 &= \{0 \prec V_{dn} \preceq Y_T \preceq V_{up}, T \geq T_0\} \\
\Pb(\Ec_0) &\geq 1 - \delta
\end{align}

The next step is to bound the self--normalized term. Under $\Ec_0$, it is clear that $Y_T$ is invertible and we have
\[
(Y_T^{+})^{1/2}S_T = Y_T^{-1/2} S_T.
\]
Define event $\Ec_1$ in the following way
\begin{align*}
&\Ec_1 = \\
&\Bigg\{||S_{T}||_{(Y_T + V_{dn})^{-1}} \leq \sqrt{8d \log {\Bigg(\dfrac{5 \text{det}(Y_TV_{dn}^{-1} + I)^{1/2d}}{\delta^{1/d}}\Bigg)}}\Bigg\}
\end{align*}

It follows from Proposition~\ref{selfnorm_bnd} that $\Pb(\Ec_1) \geq 1- \delta$. Then 
\[
\Ec_0 \implies Y_T + V_{dn} \preceq 2 Y_T \implies (Y_T + V_{dn})^{-1} \succeq \frac{1}{2}Y_T^{-1},
\]
and we have that under $\Ec_0$
\[
||S_T||_{Y_T^{-1}} \leq \sqrt{2} ||S_T||_{(Y_T + V_{dn})^{-1}}.
\]
Now considering the intersection $\Ec_0 \cap \Ec_1$, we get
\begin{align}
&\Ec_0 \cap \Ec_1 \implies \nonumber \\
&\Ec_0 \cap \Bigg\{||S_{T}||_{Y_T^{-1}} \leq \sqrt{16d \log {\Bigg(\dfrac{5 \text{det}(V_{up}V_{dn}^{-1} + I)^{1/2d}}{\delta^{1/d}}\Bigg)}}\Bigg\}
\end{align}
We replaced the LHS of $\Ec_1$ by the lower bound obtained above and in the RHS replaced $Y_T$ by its upper bound under $\Ec_0$, $V_{up}$. Further, observe that $\Pb(\Ec_0 \cap \Ec_1) \geq 1 - 2\delta$. Under $\Ec_0 \cap \Ec_1$ we get 
\begin{equation}
||A - \hat{A}||_{\text{op}} \leq \underbrace{\frac{1}{\sigma_{\min}(V_{dn})}}_{\alpha_T}\underbrace{\sqrt{16d \log {\Bigg(\dfrac{5 \text{det}(V_{up}V_{dn}^{-1} + I)^{1/2d}}{\delta^{1/d}}\Bigg)}}}_{\beta_T} \label{error_form}
\end{equation}
where $\alpha_T$ goes to zero with $T$ and $\beta_T$ is typically a constant. This shows that OLS learns $A$ with increasing accuracy as $T$ grows. The deterministic $V_{up}, V_{dn}, T_0$ differ for each regime of $\rho(A)$ and typically depend on the probability threshold $\delta$. We now sketch the approach for finding these for each regime. 
\subsection*{$Y_T$ behavior when $A \in \Sc_0 \cup \Sc_1$}
The key step here is to characterize $Y_T$ in terms of $Y_{T-1}$. 
	\begin{align}
Y_T &=  x_0 x_0^{'} + A Y_{T-1} A^{'} + \nonumber \\
&+ \sum_{t=0}^{T-1}(A x_t\eta_{t+1}^{'} + \eta_{t+1}x_t^{'}A^{'}) + \sum_{t=1}^{T}\eta_t \eta_t^{'} \nonumber \\
&\succeq A Y_{T-1} A^{'} + \nonumber \\
&+ \sum_{t=0}^{T-1}(A x_t\eta_{t+1}^{'} + \eta_{t+1}x_t^{'}A^{'}) + \sum_{t=1}^{T}\eta_t \eta_t^{'}. \label{energy_bnd}
\end{align}
Since $\{\eta_t\}_{t=1}^T$ are i.i.d. subgaussian we can show that $\sum_{t=1}^T \eta_t \eta_t^{\prime}$ concentrates near $TI_{d \times d}$ with high probability. Using Proposition~\ref{selfnorm_bnd} once again, we will show that with high probability
\begin{align*}
\sum_{t=0}^{T-1}(A x_t\eta_{t+1}^{'} + \eta_{t+1}x_t^{'}A^{'}) &\succeq -\epsilon ( A Y_{T-1} A^{'} + \sum_{t=1}^{T}\eta_t \eta_t^{\prime})
\end{align*}
where $\epsilon \leq 1/2$ whenever $\rho_i(A) \leq 1 + C/T$ and $T \geq T_0$ for some $T_0$ depending only on $A$. As a result with high probability we have 
\begin{align}
Y_T &\succeq (1-\epsilon)A Y_{T-1} A^{'} + (1 - \epsilon)\sum_{t=1}^T \eta_t \eta_t^{\prime} \nonumber\\
&\succeq (1 - \epsilon)\sum_{t=1}^T \eta_t \eta_t^{\prime}. \label{bnd1}
\end{align}
The details of this proof are provided in appendix as Section~\ref{short_proof}. When $1 - C/T \leq \rho_i(A) \leq 1 + C/T$ we note that the bound in Eq.~\eqref{bnd1} is not tight. The key to sharpening the lower bound is the following observation: for  $T >\max{\Big(2T_{\eta}\Big(\frac{\delta}{3T}\Big), 2T_s\Big(\frac{\delta}{3T}\Big), T_{ms}\Big(\frac{\delta}{2}\Big)\Big)}$ we can ensure with high probability 
\begin{align}
\sum_{\tau=1}^t \eta_{\tau} \eta_{\tau}^{\prime} &= tI \nonumber\\
Y_t &\succeq (1-\epsilon)A Y_{t-1} A^{'} + (1 - \epsilon)tI  \label{bnd2}
\end{align}
simultaneously for all $t \geq T/2$. Then we will show that $\epsilon = \beta_0(\delta)$ in Table~\ref{notation}. The sharpening of $\epsilon$ from $1/2$ to $\beta_0(\delta)$ is only possible because all the eigenvalues of $A$ are close to unity. In that case by successively expanding Eq.~\eqref{bnd2} we get 
\begin{equation}
Y_T \succeq (1 - \epsilon)^{1/\beta_{0}(\delta)}A Y_{T/2-1} A^{'} + \frac{T}{2}\sum_{t=1}^{1/\beta_{0}(\delta)}(1-\epsilon)^{t}A^{t}A^{t \prime} \label{bnd3}
\end{equation}
and then Eq.~\eqref{bnd3} can be reduced to 
\[
Y_T \succeq (1 - \epsilon)^{1/\beta_{0}(\delta)}A Y_{T/2-1} A^{'} + \frac{T (\Gamma_{1/\beta_0(\delta)}(A)-I)}{ 4e}.
\]
We show that 
$$1/\beta_0(\delta) \geq \frac{\alpha(d)TR^2\sigma_{\min}(A A^{\prime})}{8ec(A, \delta)}$$ 
and by Proposition~\ref{gramian_lb}, $Y_T \succeq \alpha(d)T^2$ for some function $\alpha(\cdot)$ that depends only on $d$. The details of the proof are provided in appendix as Section~\ref{sharp_bounds}. 

To get deterministic upper bounds for $Y_T$ with high probability, we note that 
\begin{align*}
Y_T &\preceq \text{tr}\left(\sum_{t=1}^T X_t X_t^{\prime}\right) I.
\end{align*} 
Then we can use Hanson--Wright inequality or Markov inequality to get an upper bound as shown in appendix as Proposition~\ref{energy_markov}.
\subsection*{$Y_T$ behavior when $A \in \Sc_2$}
The concentration arguments used to show the convergence for stable systems do not work for unstable systems. As discussed before $X_t = \sum_{\tau=1}^T A^{t-\tau} \eta_t$ and, consequently, $X_T$ depends strongly on $X_1, X_2, \ldots$. Due to this dependence we are unable to use typical techniques where $X_i$s are divided into roughly independent blocks of covariates. to obtain concentration results. Motivated by~\cite{lai1983asymptotic}, we instead work by transforming $x_t$ as 
\begin{align}
z_t &= A^{-t}x_t \nonumber\\
&= x_0 + \sum_{\tau=1}^{t} A^{-\tau} \eta_{\tau}. \label{zt_form}
\end{align}

The steps of the proof proceed as follows. Define 
	\begin{align}
	U_T &= A^{-T}\sum_{t=1}^T x_t x_t^{\prime} A^{-T \prime} = A^{-T} Y_T A^{-T \prime} \nonumber \\
	 &= \sum_{t=1}^{T} A^{-T+t}z_t z_t^{\prime} A^{-T+t \prime} \nonumber \\
	F_{T} &= \sum_{t=0}^{T-1} A^{-t} z_T z_T^{'} A^{-t \prime} \label{ut_ft}
	\end{align}
We show that 
\[
||F_T - U_T||_{\text{op}} \leq \epsilon.
\]
Here $\epsilon$ decays exponentially fast with $T$. Then the lower and upper bounds of $U_T$ can be shown by proving corresponding bounds for $F_T$. A necessary condition for invertibility of $F_T$ is that the matrix $A$ should be regular (in a later section we show that it is also sufficient). If $A$ is regular, the deterministic lower bound for $F_T$ is fairly straightforward and depends on $\phi_{\min}(A)$ defined in Definition~\ref{outbox}. The upper bound can be obtained by using Hanson--Wright inequality. The complete steps are given in appendix as Section~\ref{explosive}. 
\end{proof}
The analysis presented here is sharper than~\cite{faradonbeh2017finite} as we use subgaussian matrix inequalities such as Hanson--Wright Inequality (Theorem~\ref{hanson-wright}) to bound the error terms in contrast to uniformly bounding each noise variable and applying a less efficient Bernstein inequality. Another minor difference is that~\cite{lai1983asymptotic},\cite{faradonbeh2017finite} consider $||U_T-F_{\infty}||$ instead and as a result they require a martingale concentration argument to show the existence of $z_{\infty}$.

Lower bounds for identification error when $\rho(A) \leq 1$ have been derived in~\cite{simchowitz2018learning}. In Table~\ref{main_result} and Theorem~\ref{main_result}, the error in identification for explosive matrices depends on $\delta$ as $\frac{1}{\delta}$ unlike stable and marginally stable matrices where the dependence is $\log{\frac{1}{\delta}}$. Typical minimax analyses, such as the one in~\cite{simchowitz2018learning}, are unable to capture this relation between error and $\delta$. Here we show that such a dependence is unavoidable:
\begin{prop}
	\label{minimax}
	Let $A=a \geq 1.1$ be a 1--D matrix and $\hat{A} = \hat{a}$ be its OLS estimate. Then whenever $Ca^2T^2a^{-T} > \delta^2$, we have with probability at least $\delta$ that
	\[
	|a - \hat{a}| \geq \frac{C(1-a^{-2}) \delta}{ -a^2 (\log{\delta})^3}
	\]
	where $C$ is a universal constant. If $Ca^2T^2a^{-T} \leq \delta^2$ then with probability at least $\delta$ we have 
	\[
	|a - \hat{a}| \geq \Big(\frac{C(1-a^{-2})}{-\delta \log{\delta}}\Big)a^{-T}
	\]	
\end{prop} 
Our lower bounds indicate that $\frac{1}{\delta}$ is inevitable in Theorem~\ref{main_result}, \textit{i.e.}, when $Ca^2T^2a^{-T} \leq \delta^2$. Second, when $Ca^2T^2a^{-T} > \delta^2$, our bound sharpens Theorem B.2 in~\cite{simchowitz2018learning}. The proof and an explicit comparison is provided in Section~\ref{optimal_bnd}.

For the general case we use a well known fact for matrices, namely, that there exists a similarity transform $\tilde{P}$ such that 
\begin{align}
A = \tilde{P}^{-1} \begin{bmatrix}
A_{e} & 0  & 0 \\
0 & A_{ms} & 0 \\
0 & 0 & A_s 
\end{bmatrix}\tilde{P} \label{partition0}
\end{align}
Here $A_{e} \in \Sc_0, A_{ms} \in \Sc_1, A_s \in \Sc_2$. Although one might be tempted to use Theorem~\ref{main_result} to provide error bounds, mixing between different components due to the transformation $\tilde{P}$ requires a careful analysis of identification error. We show that error bounds are limited by the slowest component as we describe below. We do not provide the exact characterization due to a shortage of space. The details are given in appendix as Section~\ref{composite_result_proof}.
\begin{thm}
	\label{composite_result}	
	For any regular matrix $A$ we have with probability at least $1-\delta$,
		\begin{itemize}
			\item For $A \in \Sc_1 \cup \Sc_2$  
			$$||A - \hat{A}||_{2} \leq  \frac{\text{poly}(\log{T},  \log{\frac{1}{\delta}})}{T}$$
			whenever 
			$$T \geq \text{poly}\Big(\log{\frac{1}{\delta}}\Big)$$
			\item For $A \in \Sc_0 \cup \Sc_1 \cup \Sc_2$ 
			$$||A - \hat{A}||_{2} \leq \frac{\text{poly}(\log{T},  \log{\frac{1}{\delta}})}{\sqrt{T}}$$ 
			whenever
			$$T \geq \text{poly}\Big(\log{\frac{1}{\delta}}\Big)$$
		\end{itemize}
		Here $\text{poly}(\cdot)$ is a polynomial function.
\end{thm}
\begin{proof}
Define the partition of $A$ as Eq.~\eqref{partition0}. Since
\begin{align}
X_t &= \sum_{\tau=1}^t A^{\tau-1}\eta_{t -\tau+1} \nonumber \\
\tilde{X}_t = \tilde{P}^{-1}X_t &= \sum_{\tau=1}^t \tilde{A}^{\tau-1}\underbrace{\tilde{P}^{-1}\eta_{t -\tau+1}}_{\tilde{\eta}_{t-\tau+1}}
\end{align}
then the transformed dynamics are as follows:
\begin{align*}
\tilde{X}_{t+1} &= \tilde{A}\tilde{X}_t + \tilde{\eta}_{t+1}.
\end{align*}
Here $\{\tilde{\eta}_t\}_{t=1}^T$ are still independent. Correspondingly we also have a partition for $\tX_t, \tn_t$
\begin{align}
\tilde{X}_t = \begin{bmatrix}
\xe_t \\
\xms_t \\
\xs_t
\end{bmatrix}&, \tilde{\eta}_t = \begin{bmatrix}
\nee_t \\
\nms_t \\
\ns_t
\end{bmatrix}
\end{align}
Then we have
\begin{align}
\sum_{t=1}^T \tX_t \tX_t^{\prime} &= \sum_{t=1}^T\begin{bmatrix}
\xe_t (\xe_t)^{\prime} & \xe_t (\xms_t)^{\prime} & \xe_t (\xs_t)^{\prime}\\
\xms_t (\xe_t)^{\prime} & \xms_t (\xms_t)^{\prime} & \xms_t (\xs_t)^{\prime} \\
\xe_t (\xs_t)^{\prime} & \xs_t (\xms_t)^{\prime} & \xs_t (\xs_t)^{\prime}
\end{bmatrix} \label{mixed_matrix}
\end{align}
The next step is to show the invertibility of $\sum_{t=1}^T \tX_t \tX_t^{\prime}$. Although reminiscent of our previous set up, there are some critical differences. First, unlike before, coordinates of $\tilde{\eta}_t$, \textit{i.e.},  $\{\nee_t,\nms_t,\ns_t\}$ are not independent. A major implication is that it is no longer obvious that the cross terms between different submatrices, such as $\sum_{t=1}^T \xe_t (\xms_t)^{\prime}$,  go to zero. Our proof will have three major steps:
\begin{itemize}
	\item First we will show that the diagonal submatrices are invertible. This follows from Theorem~\ref{main_result} by arguing that the result can be extended to a noise process $\{P\eta_t\}_{t=1}^T$ where $\{\eta_t\}_{t=1}^T$ are independent subgaussian and elements of $\eta_t$ are also independent for all $t$. The only change will be the appearance of additional $\sigma_1^2(P)$ subgaussian parameter (See Corollary~\ref{dep-hanson-wright}). We will then show that 
	\begin{align*}
	X_{mss}= \sum_{t=1}^T\begin{bmatrix}
	\xms_t (\xms_t)^{\prime} & \xms_t (\xs_t)^{\prime} \\
	\xs_t (\xms_t)^{\prime} & \xs_t (\xs_t)^{\prime}
	\end{bmatrix}
	\end{align*}
	is invertible. This will follow from Theorem~\ref{main_result} (its dependent extension). Specifically, since $X_{mss}$ contains only stable and marginally stable components, it falls under $A \in \Sc_0 \cup \Sc_1$. It should be noted that since $\xms_t, \xs_t$ are not independent in general, the invertibility of $X_{mss}$ can be shown only through Theorem~\ref{main_result}. In a similar fashion, $\sum_{t=1}^T\xe_t (\xe_t)^{\prime}$ is also invertible as it corresponds to $A \in \Sc_2$. 	

	\item Since invertibility of block diagonal submatrices in $\sum_{t=1}^T \tX_t \tX_t^{\prime}$ does not imply the invertibility of the entire matrix we also need to show that the cross terms $||\xe_t (\xms_t)^{\prime}||_2,||\xe_t (\xs_t)^{\prime}||_2$ are sufficiently small relative to the appropriate diagonal blocks.

	\item Along the way we also obtain deterministic lower and upper bounds for the sample covariance matrix following which the steps for bounding the error are similar to Theorem~\ref{main_result}.
\end{itemize}
The details are in appendix as Section~\ref{composite_result_proof}.
\end{proof}

%% file: content/inconsistent.tex
\section{Inconsistency of OLS}
\label{inconsistent}
We will now show that when a matrix is irregular, then it cannot be learned despite a high signal-to-noise ratio. Consider the two cases 
\begin{align*}
A_r &= \begin{bmatrix}
1.1 & 1 \\
0 & 1.1
\end{bmatrix}, A_o = \begin{bmatrix}
1.1 & 0 \\
0 & 1.1
\end{bmatrix}
\end{align*}
Here $A_r$ is a regular matrix and $A_o$ is not. Now we run Eq.~\eqref{lti} for $A=A_r, A_o$ for $T=10^3$. Let the OLS estimate of $A_r, A_o$ be $\hat{A}_r, \hat{A}_o$ respectively. Define 
\begin{align*}
\beta_r &= [A_r]_{1,2}, \beta_o = [A_o]_{1,2} \\
\hat{\beta_r} &= [\hat{A}_r]_{1,2}, \hat{\beta}_o = [\hat{A}_o]_{1,2}
\end{align*}
Although $\beta_r \approx \hat{\beta}_r$, $\hat{\beta}_o$ does not equal zero. Instead Fig.~\ref{beta_dist} shows that $\hat{\beta}_o$ has a non--trivial distribution which is bimodal at $\{-0.55, 0.55\}$ and as a result OLS is inconsistent for $A_o$. This happens because the sample covariance matrix for $A_o$ is singular despite the fact that $\Gamma_T(A_o) = (1.1)^T I$, \textit{i.e.}, a high signal to noise ratio. In general, the relation between OLS identification of $A$ and its controllability Gramian, $\Gamma_T(A)$, is tenuous for unstable systems unlike what is suggested in~\cite{simchowitz2018learning}.
\begin{figure}
	\includegraphics[width=\linewidth]{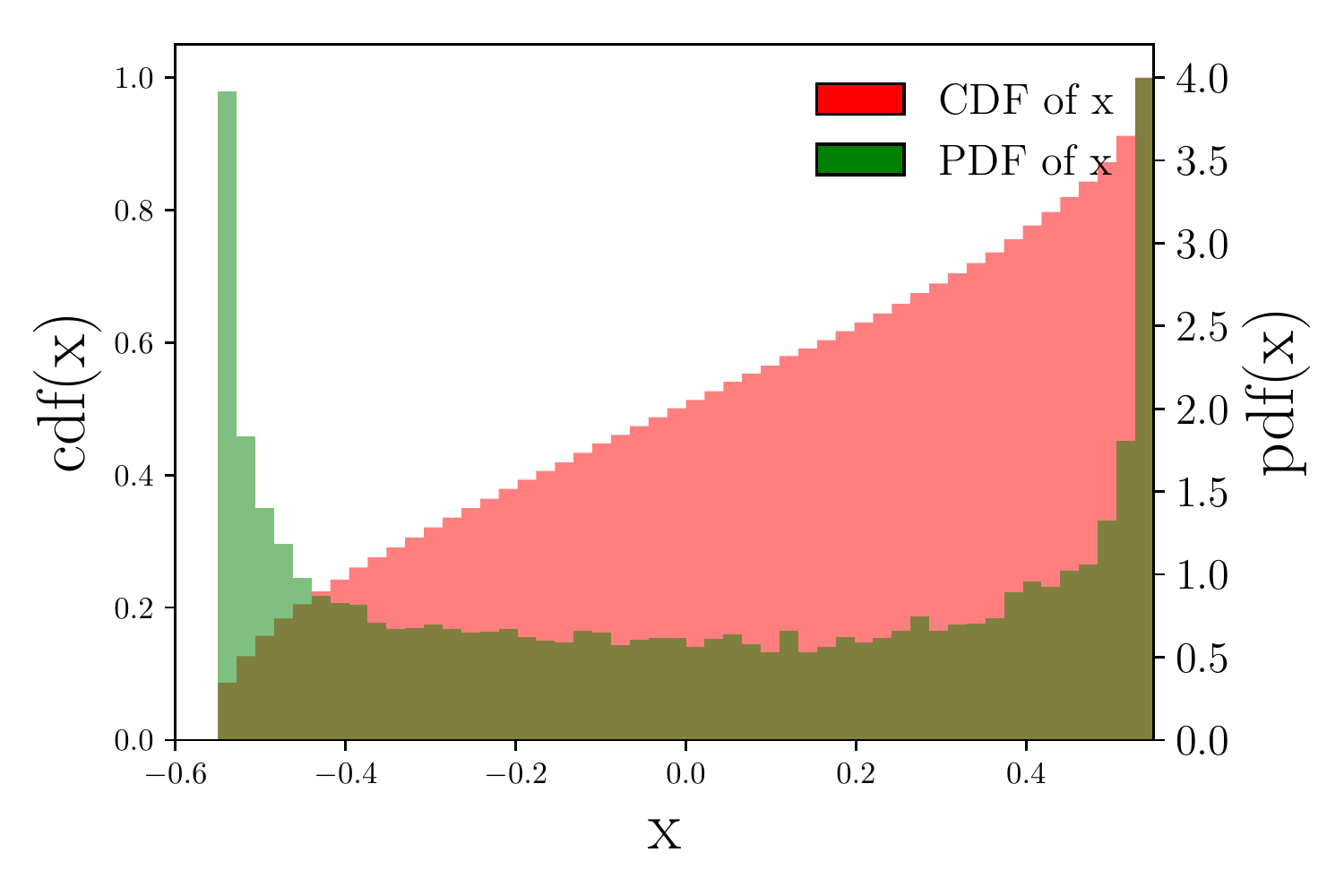}
	\caption{CDF and PDF of $\hat{\beta}_o$}
	\label{beta_dist}
\end{figure}
To see this singularity observe that 
\begin{align*}
X_{t+1} &= A_{o} \begin{bmatrix} X^{(1)}_t \\
X^{(2)}_t \end{bmatrix} + \begin{bmatrix} \eta_{t+1}^{(1)} \\
\eta_{t+1}^{(2)} \end{bmatrix}\\
Y_T &= \begin{bmatrix}
\sum_{t=1}^T (X^{(1)}_t)^2 & \sum_{t=1}^T (X^{(1)}_t)(X^{(2)}_t)\\
\sum_{t=1}^T (X^{(1)}_t)(X^{(2)}_t) & \sum_{t=1}^T (X^{(2)}_t)^2
\end{bmatrix}
\end{align*}
where $X^{(1)}_t, X^{(2)}_t$ are independent of each other. Define $a=1.1$.
\begin{prop}
	\label{singular}
Let $\{\eta_t\}_{t=1}^T$ be i.i.d standard Gaussian then whenever $T^2 \leq a^T$, we have that 
\[
||\hat{A}_o - A_{o}|| = \gamma_T
\]
where $\gamma_T$ is a random variable that admits a continuous pdf and does not decay to zero as $T \rightarrow \infty$. Further, the sample covariance matrix has the following singular values 
\begin{align*}
\sigma_1(\sum_{t=1}^T  X_t X_t^{\top}) &= \Theta(a^{2T}), \sigma_2(\sum_{t=1}^T  X_t X_t^{\top}) = O(\sqrt{T}a^{T})    
\end{align*}
\end{prop}
The proof is given in Section~\ref{inconsistent} and Proposition~\ref{condition_number}. Proposition~\ref{singular} suggests that the consistency of OLS estimate depends directly on the condition number of the sample covariance matrix. In fact, OLS is inconsistent when condition number grows exponentially fast in $T$ (as in the case of $A_o$). The proof requires a careful expansion of the (appropriately scaled) sample covariance matrix inverse using Woodbury's identity. Since the sample covariance matrix is highly ill--conditioned, it magnifies the noise-covariate cross terms so that the identification error no longer decays as time increases. Although for stable and marginally stable $A$ this invertibility can be characterized $\sigma_{\min}(\Gamma_T(A))$ such an intuition does not extend to explosive systems. This is because the behavior of $Y_T$ is dominated by ``past'' $\eta_t$s such as $\eta_1, \eta_2$ much more than the $\eta_{T-1}, \eta_{T}$ etc. When $A$ is explosive, all singular values of $||A^T||$ grow exponentially fast. Since $X_T = A^{T-1} \eta_1 + A^{T-2} \eta_2 + \ldots + A \eta_{T-1} + \eta_T$ the behavior of $X_T$ is dominated by $A^{T-1} \eta_1$. This causes a very strong dependence between $X_T$ and $X_{T+1}$ and some structural constraints (such as regularity) are necessary for OLS identification.

%% file: content/discussion.tex
\section{Discussion}
\label{discussion}
In this work we provided finite time guarantees for OLS identification for LTI systems. We show that whenever $A$ is regular, with an otherwise arbitrary distribution of eigenvalues, OLS can be used for identification. More specifically we give sharpest possible rates when $A$ belongs to one of $\{\Sc_0, \Sc_1, \Sc_2 \}$. When the assumption of regularity is violated, we show that OLS is statistically inconsistent. This suggests that statistical consistency relies on the conditioning of the sample covariance matrix and \textit{not} so much on the signal-to-noise ratio for explosive matrices. Despite substantial differences between the distributional properties of the covariates we find that time taken to reach a given error threshold scales the same (up to some constant that depends only on $A$) across all regimes in terms of the probability of error.  To see this, observe that Theorem~\ref{main_result} gives us with probability at least $1-\delta$ 
\begin{align}
A \in \Sc_0 &\implies ||A - \hat{A}|| \leq \sqrt{\frac{C_0(d)\log{\frac{1}{\delta}}}{T}} \nonumber \\
A \in \Sc_1 &\implies ||A - \hat{A}|| \leq \frac{C_1(d)}{T}{\log{\Big(\frac{T}{\delta}\Big)}} \nonumber \\
A \in \Sc_2 &\implies ||A - \hat{A}|| \leq \frac{C_2(d) \sigma_{\max}(A^{-T})}{\delta} \label{ub}
\end{align}
The lower bounds for $A \in \Sc_0$ and $A \in \Sc_1$ are given in~\cite{simchowitz2018learning} Appendix B, F.1 which are 
\begin{align}
A \in \Sc_0 &\implies ||A - \hat{A}|| \geq \sqrt{\frac{B_0(d)\log{\frac{1}{\delta}}}{T}} \nonumber\\
A \in \Sc_1 &\implies ||A - \hat{A}|| \geq \frac{B_1(d)}{T}{\log{\Big(\frac{1}{\delta}\Big)}} \label{lbb}
\end{align}
with probability at least $\delta$. For $A \in \Sc_2$ we provide a tighter lower bound in Proposition~\ref{minimax}, \textit{i.e.}, with probability at least $\delta$
\begin{equation}
A \in \Sc_2 \implies ||A - \hat{A}|| \geq \frac{B_2(d) \sigma_{\max}(A^{-T})}{-\delta \log{\delta}} \label{lbb2}
\end{equation}
Now fix an error threshold $\epsilon$, from Eq.~\eqref{ub} we get with probability $\geq 1 - \delta$
\begin{align*}
A \in \Sc_0 &\implies ||A - \hat{A}|| \leq \epsilon \text{ if } T \geq \frac{ \log{\frac{1}{\delta}}}{\epsilon^2 C_0(d)} \\
A \in \Sc_1 &\implies ||A - \hat{A}|| \leq \epsilon \text{ if } T \geq \frac{ \log{\frac{T}{\delta}}}{\epsilon C_1(d)} \\
A \in \Sc_2 &\implies ||A - \hat{A}|| \leq  \epsilon \text{ if } T \geq \frac{\log{\frac{1}{\delta {\epsilon}}} + \log{C_2(d)}}{\log{\rho_{\min}}}
\end{align*}
From Eq.~\eqref{lbb},\eqref{lbb2} we also know this is tight. In summary to reach a certain error threshold, $T$ must be at least as large as $\log{\frac{1}{\delta}}$ for every regime. 

Another key contribution of this work is providing finite time guarantees for a general distribution of eigenvalues. A major hurdle towards applying Theorem~\ref{main_result} to the general case is the mixing between separate components (corresponding to stable, marginally stable or explosive). Despite these difficulties we provide error bounds where each component, stable, marginally stable or explosive, has (almost) the same behavior as Theorem~\ref{main_result}. The techniques introduced here can be used to analyze extensions such as identification in the presence of a control input $U_t$ or heavy tailed distribution of noise (See Sections~\ref{extensions} and \ref{noise_ind}).

%% file: content/appendix_matrix.tex
\newpage
\onecolumn
\section{Appendix}
\label{appendix_matrix}
\begin{prop}
	\label{psd_result_2}
	Let $P, V$ be a psd and pd matrix respectively and define $\bar{P} = P + V$. Let there exist some matrix $Q$ for which we have the following relation
	\[
	||\bar{P}^{-1/2} Q|| \leq \gamma
	\]
	For any vector $v$ such that $v^{\prime} P v = \alpha, v^{\prime} V v =\beta$ it is true that
	\[
	||v^{\prime}Q|| \leq \sqrt{\beta+\alpha} \gamma 
	\]
\end{prop}
\begin{proof}
	Since 
	\[
	||\bar{P}^{-1/2} Q||_2^2 \leq \gamma^2 	
	\] 
	for any vector $v \in \Sc^{d-1}$ we will have 
	\[
	\frac{v^{\prime} \bar{P}^{1/2}\bar{P}^{-1/2} Q Q^{\prime}\bar{P}^{-1/2}\bar{P}^{1/2} v}{v^{\prime} \bar{P} v} \leq \gamma^2
	\]
	and substituting $v^{\prime} \bar{P} v = \alpha + \beta$ gives us
	\begin{align*}
	{v^{\prime}  Q Q^{\prime} v} &\leq \gamma^2{v^{\prime} \bar{P} v} \\
	&= (\alpha + \beta) \gamma^2
	\end{align*}
\end{proof}

\begin{prop}
	\label{inv_jordan}
Consider a Jordan block matrix $J_d(\lambda)$ given by \eqref{jordan}, then $J_d(\lambda)^{-k}$ is a matrix where each off--diagonal (and the diagonal) has the same entries, \textit{i.e.},
	\begin{equation}
	J_d(\lambda)^{-k} =\begin{bmatrix}
	a_1 & a_2 &  a_3 & \hdots & a_d \\
	0 & a_1 & a_2 & \hdots & a_{d-1} \\
	\vdots & \vdots & \ddots & \ddots & \vdots \\
	0 & \hdots & 0 & a_1 & a_2 \\
	0 & 0 & \hdots & 0 & a_1  
	\end{bmatrix}_{d \times d}
	\end{equation}
for some $\{a_i\}_{i=1}^d$.
\end{prop}
\begin{proof}
	$J_d(\lambda) = (\lambda I + N)$ where $N$ is the matrix with all ones on the $1^{st}$ (upper) off-diagonal. $N^k$ is just all ones on the $k^{th}$ (upper) off-diagonal and $N$ is a nilpotent matrix with $N^d = 0$. Then
	\begin{align*}
	(\lambda I + N)^{-1} &= (\sum_{l=0}^{d-1} (-1)^{l}\lambda^{-l-1}N^{l}) \\
	(-1)^{k-1}(k-1 )!(\lambda I + N)^{-k} &= \Big(\sum_{l=0}^{d-1} (-1)^{l}\frac{d^{k-1}\lambda^{-l-1}}{d \lambda^{k-1}}N^{l}\Big) \\
	&= \Big(\sum_{l=0}^{d-1} (-1)^{l}c_{l, k}N^{l}\Big) 
	\end{align*} 

and the proof follows in a straightforward fashion.
\end{proof}
\begin{prop}
	\label{reg_invertible}
Let $A$ be a regular matrix and $A = P^{-1} \Lambda P$ be its Jordan decomposition. Then
\[
\inf_{||a||_2 = 1}||\sum_{i=1}^d a_i \Lambda^{-i+1}||_2 > 0
\]
Further $\phi_{\min}(A) > 0$ where $\phi_{\min}(\cdot)$ is defined in Definition~\ref{outbox}.
\end{prop}
\begin{proof}
	When $A$ is regular, the geometric multiplicity of each eigenvalue is $1$. This implies that $A^{-1}$ is also regular. Regularity of a matrix $A$ is equivalent to the case when minimal polynomial of $A$ equals characteristic polynomial of $A$ (See Section~\ref{lemmab} in appendix), \textit{i.e.},
	\begin{align*}
	\inf_{||a||_2 = 1}||\sum_{i=1}^d a_i A^{-i+1}||_2  &> 0
	\end{align*}
	Since $A^{-j} = P^{-1} \Lambda^{-j} P$ we have
	\begin{align*}
	\inf_{||a||_2 = 1}||\sum_{i=1}^d a_i P^{-1}\Lambda^{-i+1}P||_2 &> 0 \\
	\inf_{||a||_2 = 1}||\sum_{i=1}^d a_i P^{-1}\Lambda^{-i+1}||_2 \sigma_{\min}(P) &> 0 \\
	\inf_{||a||_2 = 1}||\sum_{i=1}^d a_i \Lambda^{-i+1}||_2 \sigma_{\min}(P) \sigma_{\min}(P^{-1}) &> 0 \\
	\inf_{||a||_2 = 1}||\sum_{i=1}^d a_i \Lambda^{-i+1}||_2 &>0	
	\end{align*}
	Since $\Lambda$ is Jordan matrix of the Jordan decomposition, it is of the following form
	\begin{equation}
	\Lambda  =\begin{bmatrix}
	J_{k_1}(\lambda_1) & 0 & \hdots & 0 &0 \\
	0 & J_{k_2}(\lambda_2) & 0 & \hdots &0  \\
	\vdots & \vdots & \ddots & \ddots & \vdots \\
	0 & \hdots & 0 & J_{k_{l}}(\lambda_l) & 0 \\
	0 & 0 & \hdots & 0 & J_{k_{l+1}}(\lambda_{l+1})
	\end{bmatrix}
	\end{equation} 
	where $J_{k_i}(\lambda_i)$ is a $k_i \times k_i$ Jordan block corresponding to eigenvalue $\lambda_i$. Then
	\begin{equation}
	\Lambda^{-k}  =\begin{bmatrix}
	J^{-k}_{k_1}(\lambda_1) & 0 & \hdots & 0 &0 \\
	0 & J^{-k}_{k_2}(\lambda_2) & 0 & \hdots &0  \\
	\vdots & \vdots & \ddots & \ddots & \vdots \\
	0 & \hdots & 0 & J^{-k}_{k_{l}}(\lambda_l) & 0 \\
	0 & 0 & \hdots & 0 & J^{-k}_{k_{l+1}}(\lambda_{l+1})
	\end{bmatrix}
	\end{equation} 	
Since $||\sum_{i=1}^d a_i \Lambda^{-i+1}||_2 >0$, without loss of generality assume that there is a non--zero element in $k_1 \times k_1$ block. This implies 
\begin{align*}
||\underbrace{\sum_{i=1}^d a_i J_{k_1}^{-i+1}(\lambda_1)}_{=S}||_2  > 0
\end{align*}
By Proposition~\ref{inv_jordan} we know that each off--diagonal (including diagonal) of $S$ will have same element. Let $j_0 = \inf{\{j | S_{ij} \neq 0\}}$ and in column $j_0$ pick the element that is non--zero and highest row number, $i_0$. By design $S_{i_0, j_0} > 0$ and further
$$S_{k_1 -(j_0 - i_0), k_1} = S_{i_0, j_0}$$ 
because they are part of the same off--diagonal (or diagonal) of $S$. Thus the row $k_1 - (j_0 - i_0)$ has only one non--zero element because of the minimality of $j_0$.

We proved that for any $||a||=1$ there exists a row with only one non--zero element in the matrix $\sum_{i=1}^d a_i \Lambda^{-i+1}$. This implies that if $v$ is a vector with all non--zero elements, then $||\sum_{i=1}^d a_i \Lambda^{-i+1} v||_2 > 0$, \textit{i.e.},
\begin{align*}
\inf_{||a||_2 = 1}||\sum_{i=1}^d a_i \Lambda^{-i+1} v ||_2 &> 0 
\end{align*}
This implies
\begin{align*}
\inf_{||a||_2 = 1}||[v, \Lambda^{-1} v, \ldots, \Lambda^{-d+1}v] a||_2 &> 0\\
\sigma_{\min}([v, \Lambda^{-1} v, \ldots, \Lambda^{-d+1}v]) &> 0 \\
\end{align*}
By Definition~\ref{outbox} we have
\begin{align*}
\phi_{\min}(A) &> 0
\end{align*}
\end{proof}
\begin{prop}[Corollary 2.2 in~\cite{ipsen2011determinant}]
	\label{det_lb}
	For any positive definite matrix  $M$ with diagonal entries $m_{jj}$, $1 \leq j \leq d$ and $\rho$ is the spectral radius of the matrix $C$ with elements
	\begin{align*}
	c_{ij} &= 0 \hspace{3mm} \text{if } i=j \\
	 &=\frac{m_{ij}}{\sqrt{m_{ii}m_{jj}}}  \hspace{3mm} \text{if } i\neq j
	\end{align*}
	then 
	\begin{align*}
	0 < \frac{\prod_{j=1}^d m_{jj} - \text{det}(M)}{\prod_{j=1}^d m_{jj}} \leq 1 - e^{-\frac{d \rho^2}{1+\lambda_{\min}}}
	\end{align*}
	where $\lambda_{\min} = \min_{1 \leq j \leq d} \lambda_j(C)$.
\end{prop}
\begin{prop}
	\label{gramian_lb}
	Let $1 - C/T \leq \rho_i(A) \leq 1 + C/T$ and $A$ be a $d \times d$ matrix. Then there exists $\alpha(d)$ depending only on $d$ such that for every $8 d \leq t \leq T$
	\[
	\sigma_{\min}(\Gamma_t(A)) \geq t \alpha(d) 
	\]	
\end{prop}
\begin{proof}
	Since $A = P^{-1} \Lambda P$ where $\Lambda$ is the Jordan matrix. Since $\Lambda$ can be complex we will assume that adjoint instead of transpose. This gives
	\begin{align*}
	\Gamma_T(A) &= I + \sum_{t=1}^T A^t (A^{t})^{\prime} \\
		&= I + P^{-1}\sum_{t=1}^T \Lambda^tPP^{\prime} (\Lambda^t)^{*} P^{-1 \prime} \\
		&\succeq I + \sigma_{\min}(P)^2P^{-1}\sum_{t=1}^T \Lambda^t(\Lambda^t)^{*} P^{-1 \prime} 
	\end{align*}
Then this implies that 
\begin{align*}
\sigma_{\min}(	\Gamma_T(A)) &\geq 1 +\sigma_{\min}(P)^2 \sigma_{\min}(P^{-1}\sum_{t=1}^T \Lambda^t(\Lambda^t)^{\prime} P^{-1 \prime}) \\
&\geq 1 + \sigma_{\min}(P)^2 \sigma_{\min}(P^{-1})^2\sigma_{\min}(\sum_{t=1}^T \Lambda^t(\Lambda^t)^{\prime} ) \\
&\geq 1 + \frac{\sigma_{\min}(P)^2}{\sigma_{\max}(P)^2}\sigma_{\min}(\sum_{t=1}^T \Lambda^t(\Lambda^t)^{\prime} )
\end{align*}

Now 
\begin{align*}
\sum_{t=0}^T \Lambda^t(\Lambda^t)^{*} &= \begin{bmatrix}
\sum_{t=0}^T J^{t}_{k_1}(\lambda_1)(J^{t}_{k_1}(\lambda_1))^{*} & 0 & \hdots & 0  \\
0 & \sum_{t=1}^T J^{t}_{k_2}(\lambda_2)(J^{t}_{k_2}(\lambda_2))^{*} & 0 & \hdots   \\
\vdots & \vdots & \ddots & \ddots \\
0 & \hdots & 0 & \sum_{t=1}^T J^{t}_{k_{l}}(\lambda_l) (J^{t}_{k_l}(\lambda_l))^{*} 
\end{bmatrix}
\end{align*}

Since $\Lambda$	is block diagonal we only need to worry about the least singular value corresponding to some block. Let this block be the one corresponding to $J_{k_1}(\lambda_1)$, \textit{i.e.},
\begin{equation}
\sigma_{\min}(\sum_{t=0}^T \Lambda^t(\Lambda^t)^{*} ) =\sigma_{\min}(\sum_{t=0}^T J^{t}_{k_1}(\lambda_1)(J^{t}_{k_1}(\lambda_1))^{*}) \label{bnd_1}
\end{equation}

Define $B =  \sum_{t=0}^T J^{t}_{k_1}(\lambda_1)(J^{t}_{k_1}(\lambda_1))^{*}$. Note that $J_{k_1}(\lambda_1) = (\lambda_1 I + N)$ where $N$ is the nilpotent matrix that is all ones on the first off--diagonal and $N^{k_1} = 0$. Then 
\begin{align*}
(\lambda_1 I + N)^t &= \sum_{j=0}^t {t \choose j} \lambda_1^{t-j}N^{j} \\
(\lambda_1 I + N)^t((\lambda_1 I + N)^t)^{*} &= \Big(\sum_{j=0}^t {t \choose j} \lambda_1^{t-j}N^{j}\Big)\Big(\sum_{j=0}^t {t \choose j} (\lambda_1^{*})^{t-j}N^{j \prime}\Big) \\
&= \sum_{j=0}^t {t \choose j}^2 |\lambda_1|^{2(t-j)} \underbrace{N^j (N^j)^{\prime}}_{\text{Diagonal terms}} + \sum_{j \neq k}^{j=t, k=t} {t \choose k}{t \choose j} \lambda_1^j (\lambda_1^{*})^{k} N^j (N^k)^{\prime} \\
&= \sum_{j=0}^t {t \choose j}^2 |\lambda_1|^{2(t-j)} \underbrace{N^j (N^j)^{\prime}}_{\text{Diagonal terms}} + \sum_{j > k}^{j=t, k=t} {t \choose k}{t \choose j} \lambda_1^j (\lambda_1^{*})^{k} N^j (N^k)^{\prime} \\
&+ \sum_{j< k}^{j=t, k=t} {t \choose k}{t \choose j} \lambda_1^j (\lambda_1^{*})^{k} N^j (N^k)^{\prime} \\
&= \sum_{j=0}^t {t \choose j}^2 |\lambda_1|^{2(t-j)} \underbrace{N^j (N^j)^{\prime}}_{\text{Diagonal terms}} + \sum_{j > k}^{j=t, k=t} {t \choose k}{t \choose j} \underbrace{|\lambda_1|^{2k} \lambda_1^{j-k} N^{j-k} N^k(N^k)^{\prime}}_{\text{On $(j-k)$ upper off--diagonal}} \\
&+ \sum_{j< k}^{j=t, k=t} {t \choose k}{t \choose j} \underbrace{|\lambda_1|^{2j} (\lambda_1^{*})^{k-j} N^j(N^{j})^{\prime} (N^{j-k})^{\prime}}_{\text{On $(k-j)$ lower off--diagonal}}
\end{align*}
Let $\lambda_1 = r e^{i\theta}$, then similar to~\cite{erxiong1994691}, there is $D = \text{Diag}(1, e^{-i\theta}, e^{-2i\theta}, \ldots, e^{-i(k_1-1)\theta})$ such that $D (\lambda_1 I + N)^t((\lambda_1 I + N)^t)^{*} D^{*}$ is a real matrix. Observe that any term on $(j-k)$ upper off--diagonal of $(\lambda_1 I + N)^t((\lambda_1 I + N)^t)^{*}$ is of the form $r_0 e^{i(j-k)\theta}$. In the product $D (\lambda_1 I + N)^t((\lambda_1 I + N)^t)^{*} D^{*}$ any term on the $(j-k)$ upper off diagonal term now looks like $e^{-ij\theta + ik\theta} r_0 e^{i(j-k)\theta} = r_0$, which is real. Then we have 
\begin{align}
D (\lambda_1 I + N)^t((\lambda_1 I + N)^t)^{*} D^{*} &= \sum_{j=0}^t {t \choose j}^2 |\lambda_1|^{2(t-j)} \underbrace{N^j (N^j)^{\prime}}_{\text{Diagonal terms}} + \sum_{j > k}^{j=t, k=t} {t \choose k}{t \choose j} \underbrace{|\lambda_1|^{2k} |\lambda_1|^{j-k} N^{j-k} N^k(N^k)^{\prime}}_{\text{On $(j-k)$ upper off--diagonal}} \nonumber\\
&+ \sum_{j< k}^{j=t, k=t} {t \choose k}{t \choose j} \underbrace{|\lambda_1|^{2j} |\lambda_1|^{k-j} N^j(N^{j})^{\prime} (N^{k-j})^{\prime}}_{\text{On $(k-j)$ lower off--diagonal}} \label{real}
\end{align}
Since $D$ is unitary and $D (\lambda_1 I + N)^t((\lambda_1 I + N)^t)^{*} D^{*} =(|\lambda_1| I + N)^t((|\lambda_1| I + N)^t)^{\prime} $, we can simply work with the case when $\lambda_1 > 0$ and real, as the singular values remain invariant under unitary transformations. Now we show the growth of $ij^{th}$ term of the product $D (\lambda_1 I + N)^t((\lambda_1 I + N)^t)^{*} D^{*})$,
Define $B=\sum_{t=1}^T (|\lambda_1| I + N)^t((|\lambda_1| I + N)^t)^{\prime}$
\begin{align}
B_{ll} &=\sum_{t=1}^T [(\lambda_1 I + N)^t((\lambda_1 I + N)^t)^{*}]_{ll} \\
&= \sum_{t=1}^T \sum_{j=0}^{k_1-l} {t \choose j}^2 |\lambda_1|^{2(t-j)} \label{bll}
\end{align}
Since $1-C/T \leq |\lambda_1| \leq 1+C/T$, then for every $t \leq T$ we have 
$$e^{-C} \leq |\lambda_1|^t \leq e^{C}$$
Then 
\begin{align}
B_{ll} &= \sum_{t=1}^T \sum_{j=0}^{k_1-l} {t \choose j}^2 |\lambda_1|^{2(t-j)} \nonumber\\
&\geq e^{-2C} \sum_{t=1}^T \sum_{j=0}^{k_1-l} {t \choose j}^2 \nonumber\\
& \geq e^{-2C} \sum_{t=T/2}^T \sum_{j=0}^{k_1-l} {t \choose j}^2 \geq e^{-2C} \sum_{t=T/2}^T c_{k_1} \frac{t^{2k_1-2l+2} - 1}{t^2 - 1}  \geq C(k_1) T^{2k_1 - 2l+1} \label{lb}
\end{align}
An upper bound can be achieved in an equivalent fashion. 
\begin{align}
B_{ll} &= \sum_{t=1}^T \sum_{j=0}^{k_1-l} {t \choose j}^2 |\lambda_1|^{2(t-j)} \nonumber\\
& \leq e^{2C} T \sum_{j=0}^{k_1-l} T^{2j} \leq C(k_1) T^{2k_1 - 2l + 1} \label{ub1}
\end{align}
Similarly, for any $B_{k,k+l} $ we have
\begin{align}
B_{k, k+l} &=\sum_{t=1}^T \sum_{j=0}^{k_1-k - l} {t \choose j}{t \choose j+l} |\lambda_1|^{2j} |\lambda_1|^{l}  \\
&\geq \sum_{t=1}^T e^{-2C} \sum_{t=T/2}^T  \sum_{j=0}^{k_1-k - l} {t \choose j}{t \choose j+l}   \\
&\geq e^{-2C} \frac{T}{2}\sum_{j=0}^{k_1-k - l} {T/2 \choose j}{T/2 \choose j+l}  \\
&\geq C(k_1) T^{2k_1 - 2k -l +1}
\end{align}
and by a similar argument as before we get $B_{jk} = C(k_1)T^{2k_1-j-k +1}$. For brevity we use the same $C(k_1)$ to indicate different functions of $k_1$ as we are interested only in the growth with respect to $T$. To summarize
\begin{align}
B_{jk} &= C(k_1)T^{2k_1 - j - k +1} \label{jordan_value}
\end{align} 
whenever $T \geq 8d$. Recall Proposition~\ref{det_lb}, let the $M$ there be equal to $B$ then since 
\[
C_{ij} = C(k_1)\frac{B_{ij}}{\sqrt{B_{ii} B_{jj}}} = C(k_1)\frac{T^{2k_1 - j -k +1}}{\sqrt{T^{4k_1 - 2j - 2k + 2}}}
\]
it turns out that $C_{ij}$ is independent of $T$ and consequently $\lambda_{min}(C), \rho$ are independent of $T$ and depend only on $k_1$: the Jordan block size. Then $\prod_{j=1}^{k_1} B_{jj} \geq \text{det}(B) \geq \prod_{j=1}^{k_1} B_{jj} e^{-\frac{d\rho^2}{1 + \lambda_{\min}}} = C(k_1) \prod_{j=1}^{k_1} B_{jj}$. This means that $\text{det}(B) = C(k_1) \prod_{j=1}^{k_1} B_{jj}$ for some function $C(k_1)$ depending only on $k_1$. Further using the values for $B_{jj}$ we get
\begin{equation}
\label{det}
\text{det}(B) = C(k_1) \prod_{j=1}^{k_1} B_{jj} = \prod_{j=1}^{k_1} C(k_1) T^{2k_1 - 2l +1} = C(k_1) T^{k_1^2}
\end{equation}
Next we use Schur-Horn theorem, \textit{i.e.}, let $\sigma_i(B)$ be the ordered singular values of $B$ where $\sigma_i(B) \geq \sigma_{i+1}(B)$. Then $\sigma_i(B)$ majorizes the diagonal of $B$, \textit{i.e.}, for any $k \leq k_1$
\[
\sum_{i=1}^k \sigma_i(B) \geq \sum_{i=1}^{k} B_{ii}
\] 
Observe that $B_{ii} \leq B_{jj}$ when $i \leq j$. Then from Eq.~\eqref{jordan_value} it implies that
\begin{align*}
B_{k_1 k_1}=C_1(k_1)T &\geq \sigma_{k_1}(B) \\
\sum_{j=k_1-1}^{k_1} B_{jj} &= C_{2}(k_1)T^{3} + C_1(k_1)T \geq \sigma_{k_1 - 1}(A) + \sigma_{k_1}(A)
\end{align*}

Since $k_1 \geq 1$ it can be checked that for $T \geq T_1 =2k_1\sqrt{\frac{C_1(k_1)}{C_2(k_1)}}$ we have $\sigma_{k_1-1}(A) \leq {(1+(2k_1)^{-2})C_2(k_1)T^3} \leq {(1+k_1^{-1})C_2(k_1)T^3}$ as for every $T \geq T_1$ we have $C_2(k_1)T^3 \geq 4k_1^2C_1(k_1)T$. Again to upper bound $\sigma_{k_1-2}(A)$ we will use a similar argument
\begin{align*}
\sum_{j=k_1-2}^{k_1} B_{jj} &= C_3(k_1)T^{5} + C_2(k_1)T^{3} + C_1(k_1)T  \geq \sigma_{k_1-2}(A) +\sigma_{k_1-1}(A) + \sigma_{k_1}(A)
\end{align*}
and show that whenever 
\[
T \geq \max{\Big(T_1, 2k_1\sqrt{\frac{C_2(k_1)}{C_3(k_1)}}\Big)}
\]
we get $\sigma_{k_1-2}(A) \leq (1+(2k_1)^{-2} + (2k_1)^{-4}){C_3(k_1)T^5} \leq (1+k_1^{-1}){C_3(k_1)T^5}$ because $T \geq T_1$ ensures $C_2(k_1)T^3 \geq 4k_1^2C_1(k_1)T$ and $T \geq T_2 = 2k_1\sqrt{\frac{C_2(k_1)}{C_3(k_1)}}$ ensures $C_3(k_1)T^5 \geq 4k_1^2 C_2(k_1)T^3$. The $C_i(k_1)$ are not important, the goal is to show that for a sufficiently large $T$ we have an upper bound on each singular values (roughly) corresponding to the diagonal element. Similarly we can ensure for every $i$ we have $\sigma_i(A) \leq (1+k_1^{-1})C_{k_1 -i+1}(k_1)T^{2k_1 - 2i + 1}$, whenever 
\[
T > T_{i} = \max{\Big(T_{i-1}, 2k_1\sqrt{\frac{C_{i}(k_1)}{C_{i+1}(k_1)}}\Big)}
\]
Recall Eq.~\eqref{det} where $\text{det}(B) = C(k_1) T^{k_1^2}$. Assume that $\sigma_{k_1}(B) < \frac{C(k_1) T}{e \prod_{i=1}^d C_{i+1}(k_1)}$. Then whenever $T \geq \max{\Big(8d, \sup_{i}2k_1\sqrt{\frac{C_{i}(k_1)}{C_{i+1}(k_1)}}\Big)}$
\begin{align*}
\text{det}(B) &= C(k_1) T^{k_1^2} \\
\prod_{i=1}^{k_1}\sigma_i &= C(k_1) T^{k_1^2} \\
\sigma_{k_1}(B)(1+k_1^{-1})^{k_1-1} T^{k_1^2-1}\prod_{i=2}^{k_1}C_{i+1}  &\geq C(k_1) T^{k_1^2} \\
\sigma_{k_1}(B) &\geq \frac{C_{k_1}T}{(1+k_1^{-1})^{k_1-1}\prod_{i=2}^{k_1}C_{i+1}} \\
&\geq \frac{C(k_1) T}{e \prod_{i=1}^{k_1} C_{i+1}(k_1)}
\end{align*}

which is a contradiction. This means that $\sigma_{k_i}(B) \geq \frac{C(k_1) T}{e \prod_{i=1}^{k_1} C_{i+1}(k_1)}$. This implies 
\[
\sigma_{\min}(\Gamma_T(A)) \geq 1 + \frac{\sigma_{\min}(P)^2}{\sigma_{\max}(P)^2} C(k_1)T
\]
for some function $C(k_1)$ that depends only on $k_1$.
\end{proof}	
It is possible that $\alpha(d)$ might be exponentially small in $d$, however for many cases such as orthogonal matrices or diagonal matrices $\alpha(A)=1$ [As shown in~\cite{simchowitz2018learning}]. We are not interested in finding the best bound $\alpha(d)$ rather show that the bound of Proposition~\ref{gramian_lb} exists and assume that such a bound is known. 

\begin{prop}
	\label{gramian_ratio}
Let $t_1/t_2 = \beta > 1$ and $A$ be a $d \times d$ matrix. Then 
\[
\lambda_1(\Gamma_{t_1}(A)\Gamma_{t_2}^{-1}(A)) \leq C(d, \beta)
\]
where $C(d, \beta)$ is a polynomial in $\beta$ of degree at most $d^2$ whenever $t_i \geq 8d$.
\end{prop}
\begin{proof}
	Since $\lambda_1(\Gamma_{t_1}(A)\Gamma_{t_2}^{-1}(A))  \geq 0$
	\begin{align*}
	\lambda_1(\Gamma_{t_1}(A)\Gamma_{t_2}^{-1}(A)) &\leq \text{tr}(\Gamma_{t_1}(A)\Gamma_{t_2}^{-1}(A)) \\
	&\leq  \text{tr}(\Gamma_{t_2}^{-1/2}(A)\Gamma_{t_1}(A)\Gamma_{t_2}^{-1/2}(A)) \\
	&\leq d \sigma_1(\Gamma_{t_2}^{-1/2}(A)\Gamma_{t_1}(A)\Gamma_{t_2}^{-1/2}(A)) \\
	&\leq d\sup_{||x|| \neq  0}\frac{x^{\prime} \Gamma_{t_1}(A) x}{x^{\prime}\Gamma_{t_2}(A) x}
	\end{align*}

Now 
\begin{align*}
\Gamma_{t_i}(A) &= P^{-1}\sum_{t=0}^{t_i} \Lambda^{t}PP^{\prime}(\Lambda^{t})^{*}P^{-1 \prime} \\
&\preceq  \sigma_{\max}(P)^2 P^{-1}\sum_{t=0}^{t_i} \Lambda^{t}(\Lambda^{t})^{*}P^{-1 \prime} \\
\Gamma_{t_i}(A) &\succeq  \sigma_{\min}(P)^2 P^{-1}\sum_{t=0}^{t_i} \Lambda^{t}(\Lambda^{t})^{*}P^{-1 \prime}
\end{align*}
Then this implies 
\[
\sup_{||x|| \neq  0}\frac{x^{\prime} \Gamma_{t_1}(A) x}{x^{\prime}\Gamma_{t_2}(A) x} \leq \frac{\sigma_{\max}(P)^2}{\sigma_{\min}(P)^2} \sup_{||x|| \neq  0}\frac{x^{\prime} \sum_{t=0}^{t_1} \Lambda^{t}(\Lambda^{t})^{*} x}{x^{\prime}\sum_{t=0}^{t_2} \Lambda^{t}(\Lambda^{t})^{*} x}
\]
Then from Lemma 12 in~\cite{abbasi2011improved} we get that
\[
\sup_{||x|| \neq  0}\frac{x^{\prime} \sum_{t=0}^{t_1} \Lambda^{t}(\Lambda^{t})^{*} x}{x^{\prime}\sum_{t=0}^{t_2} \Lambda^{t}(\Lambda^{t})^{*} x} \leq \frac{\text{det}(\sum_{t=0}^{t_1} \Lambda^{t}(\Lambda^{t})^{*})}{\text{det}(\sum_{t=0}^{t_2} \Lambda^{t}(\Lambda^{t})^{*})}
\]
Then 
\begin{align*}
\frac{\text{det}(\sum_{t=0}^{t_2} \Lambda^{t}(\Lambda^{t})^{*})}{\text{det}(\sum_{t=0}^{t_1} \Lambda^{t}(\Lambda^{t})^{*})} &\leq  \frac{\text{det}(\prod_{i=1}^l (\sum_{t=0}^{t_2} J_{k_i}(\lambda_i)^{t}(J_{k_i}(\lambda_i)^{t})^{*}))}{\text{det}(\prod_{i=1}^l(\sum_{t=0}^{t_1} J_{k_i}(\lambda_i)^{t}(J_{k_i}(\lambda_i)^{t})^{*}))}
\end{align*}
Here $l$ are the number of Jordan blocks of $A$. Then our assertion follows from Eq.~\eqref{det} which implies that the determinant of $\sum_{t=0}^{t_2} J_{k_i}(\lambda_i)^{t}(J_{k_i}(\lambda_i)^{t})^{*} $ is equal to the product of the diagonal elements (times a factor that depends only on Jordan block size), \textit{i.e.}, $C(k_i)t_2^{k_i^2}$. As a result the ratio is given by
\[
\frac{\text{det}(\prod_{i=1}^l (\sum_{t=0}^{t_2} J_{k_i}(\lambda_i)^{t}(J_{k_i}(\lambda_i)^{t})^{*}))}{\text{det}(\prod_{i=1}^l(\sum_{t=0}^{t_1} J_{k_i}(\lambda_i)^{t}(J_{k_i}(\lambda_i)^{t})^{*}))} = \prod_{i=1}^l \beta^{k_i^2}
\] 
whenever $t_2, t_1 \geq 8d$. Summarizing we get 
\[
\sup_{||x|| \neq  0}\frac{x^{\prime} \Gamma_{t_1}(A) x}{x^{\prime}\Gamma_{t_2}(A) x} \leq \frac{\sigma_{\max}(P)^2}{\sigma_{\min}(P)^2} \prod_{i=1}^l \beta^{k_i^2}
\]
\end{proof}

%% file: content/appendix_prob.tex
\section{Probabilistic Inequailities}
\label{prob_ineq}
\begin{prop}[\cite{vershynin2010introduction}]
	\label{eps_net}
	Let $M$ be a random matrix. Then we have for any $\epsilon < 1$ and any $w \in \Sc^{d-1}$ that 
	\[
	\Pb(||M|| > z) \leq (1 + 2/\epsilon)^d \Pb(||Mw|| > (1-\epsilon)z)
	\]
\end{prop}
The proof of the Proposition can be found, for instance, in \cite{vershynin2010introduction}.

Proposition~\ref{eps_net} helps us in using the tools developed in de la Pena et. al. and~\cite{abbasi2011improved} for self--normalized martingales. We will define $\tilde{S}_t = \sum_{\tau=0}^{t-1} X_{\tau} \tilde{\eta}_{\tau+1}$ where $\tilde{\eta}_t=w^T \eta_t$ is standard normal when $w$ is a unit vector. Specifically, we use Lemma 9 of~\cite{abbasi2011improved} which we state here for convenience:

\begin{thm}[Theorem 1 in~\cite{abbasi2011improved}]
	\label{selfnorm_main}
	Let $\{\bcF_t\}_{t=0}^{\infty}$ be a filtration. Let $\{\eta_{t}\}_{t=1}^{\infty}$ be a real valued stochastic process such that $\eta_t$ is $\bcF_t$ measurable and $\eta_t$ is conditionally $R$-sub-Gaussian for some $R > 0$., \textit{i.e.}, 
	\[
	\forall \lambda \in \Rb \hspace{2mm} \Ex[e^{\lambda \eta_t} | \bcF_{t-1}] \leq e^{\frac{\lambda^2 R^2}{2}} 
	\]
	
	Let $\{X_t\}_{t=1}^{\infty}$ be an $\Rb^d$--valued stochastic process such that $X_t$ is $\bcF_{t}$ measurable. Assume that $V$ is a $d \times d$ positive definite matrix. For any $t \geq 0$ define 
	\[
	\bar{V}_t = V + \sum_{s=1}^t X_s X_s^{\prime} \hspace{2mm} S_t = \sum_{s=1}^t \eta_{s+1} X_s
	\]
	Then for any $\delta > 0$ with probability at least $1-\delta$ for all $t \geq 0$	
	\[
	||{S}_{t}||^2_{\bar{V}^{-1}_{t}} \leq 2 R^2 \log{\Bigg(\dfrac{\text{det}(\bar{V}_{t})^{1/2} \text{det}(V)^{-1/2}}{\delta}\Bigg)}
	\]
\end{thm} 
\begin{prop}
	\label{selfnorm_bnd_proof}
	Let $P$ have full row rank and
		\[
		X_{t+1} = AX_t + P \eta_{t+1}
		\]
	where $\{\eta_t\}_{t=1}^T$ is an i.i.d. subGaussian process with variance proxy $=1$ and each $\eta_t$ has independent elements. 
	For any $0 < \delta < 1$, we have with probability $1 - \delta$
	\begin{align}
	&||(\bar{Y}_{T-1})^{-1/2} \sum_{t=0}^{T-1} X_t \eta_{t+1}^{\prime}P^{\prime}||_2 \nonumber\\
	&\leq R\sqrt{8d \log {\Bigg(\dfrac{5 \text{det}(\bar{Y}_{T-1})^{1/2d} \text{det}(V)^{-1/2d}}{\delta^{1/d}}\Bigg)}}
	\end{align}
	where $\bar{Y}^{-1}_{\tau} = (\sum_{t=1}^{\tau} X_{t}X_t^{\prime}+ V)^{-1}$ and any deterministic $V$ with $V \succ 0$.
\end{prop}
\begin{proof}
	Note that $P\eta_t$ is a non--trivial subGaussian if $P$ has full rank. 
	
	Define $S_t = \sum_{s=1}^t X_s\eta_{s+1}^{\prime}P^{\prime}$. Using Proposition~\ref{eps_net} and setting $\epsilon=1/2$, we have that
	\begin{align}
	\Pb(||\bar{Y}_{T-1}^{-1/2}S_{T-1}||_2 \leq y) \leq 5^d \Pb(||\bar{Y}_{T-1}^{-1/2}S_{T-1}w||_2 \leq \frac{y}{2}) = \Pb(||\bar{Y}_{T-1}^{-1/2}S_{T-1}w||^2_2 \leq \frac{y^2}{4})\label{eq1}
	\end{align}
	Setting $S_{T-1} w = \sum_{s=1}^{T-1}X_s \eta_{s+1}^{\prime}P^{\prime}w$ we observe that $\eta_{s+1}^{\prime}P^{\prime}w$ satisfies the conditions of Theorem~\ref{selfnorm_main} with variance proxy $\sigma_{\max}(P)^2$. Then replace in Eq.~\eqref{eq1}
	$$y^2 = 8R^2 \log{\Bigg(\dfrac{\text{det}(\bar{Y}_{T-1})^{1/2} \text{det}(V)^{-1/2}}{5^{-d}\delta}\Bigg)}$$ 
	which gives us from Theorem~\ref{selfnorm_main}
	\[
	\Pb(||\bar{Y}_{T-1}^{-1/2}S_{T-1}||_2 \leq y) \leq \delta
	\]
\end{proof}

\begin{thm}[Hanson--Wright Inequality]
	\label{hanson-wright}
	Given a subGaussian vector $X=(X_1, X_2, \ldots, X_n) \in \Rb^n$ with $\sup_i ||X_i||_{\psi_2}\leq K$ and $X_i$ are independent. Then for any $B \in \Rb^{n \times n}$ and $t \geq 0$
	\begin{align}
	&\Pr(|X^{\prime} B X - \Ex[X^{\prime} B X]| \leq t) \nonumber\\
	&\leq 2 \exp\Bigg\{- c \min{\Big(\frac{t}{K^2 ||B||}, \frac{t^2}{K^4 ||B||^2_{HS}}\Big)}\Bigg\}
	\end{align}
\end{thm}
\begin{prop}[Theorem 5.39~\cite{vershynin2010introduction}]
	\label{noise_energy_bnd}
 Let $E$ be an $T \times d$ matrix whose rows $\eta_i^{\prime}$ are independent sub--Gaussian isotropic random vectors with variance proxy $1$ in $\Rb^{d}$. Then for every $t \geq 0$, with probability at least $1 - 2 e^{-ct^2}$ one has
\begin{equation}
    \sqrt{T} - C\sqrt{d} - t \leq \sigma_{\min}(E) \leq  \sqrt{T} + C\sqrt{d} + t
\end{equation}
\end{prop}
The implication of Proposition~\ref{noise_energy_bnd} is as follows: $E^{\prime}E \succeq (\sqrt{T} - C\sqrt{d} - t)^2 I$ with probability at least $1 - 2 e^{-ct^2}$. Let $t = \sqrt{\frac{1}{c}\log{\frac{2}{\delta}}}$, and ensure that 
\[
T \geq T_{\eta}(\delta) = C \Big( d + \log{\frac{2}{\delta}} \Big)
\]
for some large enough universal constant $C$. Then for $T > T_{\eta}(\delta)$ we have, with probability at least $1-\delta$, that
\begin{align}
\frac{3}{4}I &\preceq \dfrac{1}{T}\underbrace{\sum_{t=1}^T \eta_t \eta_t^{\prime}}_{E^{\prime}E} \preceq \frac{5}{4}I  \label{tight_noise_bound}
\end{align}
Further with the same probability 
\begin{align}
\frac{3\sigma_{\min}^{2}(P)}{4}I &\preceq \dfrac{1}{T}\sum_{t=1}^T P\eta_t \eta_t^{\prime}P^{\prime} \preceq \frac{5\sigma_{\max}^{2}(P)}{4}I \nonumber \\
T_{\eta}(\delta) &= C \Big( d + \log{\frac{2}{\delta}} \Big) \label{dep_tight_noise_bound}
\end{align}
\begin{cor}[Dependent Hanson--Wright Inequality]
	\label{dep-hanson-wright}
	Given independent subGaussian vectors $X_i \in \Rb^d$ such that $X_{ij}$ are independent and $\sup_{ij} ||X_{ij}||_{\psi_2}\leq K$. Let $P$ have full row rank. Define 
	$$X=\begin{bmatrix}PX_1 \\
	PX_2 \\
	\vdots \\
	PX_n \end{bmatrix} \in \Rb^{dn}$$ 
	Then for any $B \in \Rb^{dn \times dn}$ and $t \geq 0$
	\begin{align}
	&\Pr(|X^{\prime} B X - \Ex[X^{\prime} B X]| \leq t) \nonumber\\
	&\leq 2 \exp\Bigg\{- c \min{\Big(\frac{t}{K^2\sigma_1^2(P) ||B||}, \frac{t^2}{K^4 \sigma_1^4(P)||B||^2_{HS}}\Big)}\Bigg\}
	\end{align}
\end{cor}
\begin{proof}
	Define
	$$\tX = \begin{bmatrix}X_1 \\
	X_2 \\
	\vdots \\
	X_n \end{bmatrix}$$
	Now $\tilde{X}$ is such that $\tX_i$ are independent. Observe that $X = (I_{n \times n}\otimes P)\tX $. Then $X^{\prime} B X = \tilde{X} (I_{n \times n}\otimes P) B (I_{n \times n}\otimes P') \tilde{X}$. Since 
	\begin{align*}
	||(I_{n \times n}\otimes P) B (I_{n \times n}\otimes P')|| &\leq \sigma_1^2(P) ||B|| \\
	\text{tr}((I_{n \times n}\otimes P) B (I_{n \times n}\otimes P')(I_{n \times n}\otimes P) B (I_{n \times n}\otimes P')) &\leq \sigma_1^2(P) \text{tr}((I_{n \times n}\otimes P) B^2 (I_{n \times n}\otimes P')) \\
	&\leq \sigma_1^4(P) \text{tr}(B^2)
	\end{align*} and now we can use Hanson--Wright in Theorem~\ref{hanson-wright} and get the desired bound.
\end{proof}	
Let $X_t = \sum_{j=0}^{t-1} A^j \eta_{t-j}$. 

\begin{prop}
	\label{energy_markov}
	Let $P$ have full row rank and
	\[
	X_{t+1} = AX_t + P \eta_{t+1}
	\]
	where $\{\eta_t\}$ is an i.i.d. process and each $\eta_t$ has independent elements. Then with probability at least $1-\delta$, we have 
	\begin{align*}
	||\sum_{t=1}^T X_t X_t^{\prime}||_2 &\leq \sigma_1(P)^2\frac{ T\text{tr}(\Gamma_{T-1}(A))}{\delta} \\
	||\sum_{t=1}^T AX_t X_t^{\prime}A^{\prime}||_2 &\leq \sigma_1(P)^2 \frac{T\text{tr}(\Gamma_{T}(A) - I)}{\delta} 
	\end{align*}
	Let $\delta \in (0, e^{-1})$ then with probability at least $1-\delta$
	\[
	||\sum_{t=1}^T X_t X_t^{\prime}||_2 \leq \sigma_1(P)^2 \text{tr}(\sum_{t=0}^{T-1}\Gamma_t(A))\Big(1 + \frac{1}{c}\log{\Big(\frac{1}{\delta}\Big)}\Big)
	\]
	for some universal constant $c$.
\end{prop}
\begin{proof}
	Define $\tilde{\eta} = \begin{bmatrix} P\eta_1 \\ P\eta_2 \\ \vdots \\ P\eta_T\end{bmatrix}$. Then $\tilde{\eta}$ is a non--trivial subGaussian whenever $P$ has full row rank. 
	
	As in Corollary~\ref{dep-hanson-wright} by defining $\tilde{A}$ as
	\[
	\tilde{A} = \begin{bmatrix} I & 0 & 0 & \hdots & 0 \\
	A & I & 0 & \hdots & 0 \\
	\vdots & \vdots & \ddots & \vdots &\vdots\\
	\vdots & \vdots & \vdots & \ddots&\vdots\\
	A^{T-1} & A^{T-2} & A^{T-3} & \hdots & I
	\end{bmatrix} (I_{n \times n} \otimes P^{\prime})
	\]
	observe that 
	\[
	\tilde{A} \tilde{\eta} =  \begin{bmatrix} X_1 \\ X_2 \\ \vdots \\ X_T\end{bmatrix}.
	\]
	Since
	\[
	||X_t X_t^{\prime}|| = X_t^{\prime} X_t,
	\]
	we have that 
	$$||\sum_{t=1}^T X_t X_t^{\prime}|| \leq \sum_{t=1}^T  X_t^{\prime} X_t = \tilde{\eta}^{\prime} \tilde{A}^{\prime} \tilde{A}\tilde{\eta} = \text{tr}(\tilde{A}\tilde{\eta}\tilde{\eta}^{\prime} \tilde{A}^{\prime}).$$
	The assertion of proposition follows by applying Markov's Inequality to $\text{tr}(\tilde{A}\tilde{\eta} {\eta}^{\prime} \tilde{A}^{\prime})$. For the second part observe that each block matrix of $\tilde{A}$ is scaled by $A$, but the proof remains the same. 
	Then in the notation of Theorem~\ref{hanson-wright} $B=\tilde{A}^{\prime} \tilde{A}, X=\tilde{\eta}$
	\begin{align*}
	||B||_S &= \text{tr}(\tilde{A}^{\prime}\tilde{A}) \\
	&= \sum_{t=0}^{T-1}\text{tr}(\Gamma_t(A)) \\
	||B||_F^2 &\leq ||B||_S ||B||_2
	\end{align*}
	Define $c^{*} = \min{(c, 1)}$. Set $t = \frac{||B||_F^2}{c^{*}||B||}  {\log{(\frac{1}{\delta})}}$ and assume $\delta \in (0, e^{-1})$ then
	\begin{align*}
	\frac{t}{c^{*}||B||} \leq \frac{t^2}{c^{*}||B||_F^2}
	\end{align*}
	we get from Theorem~\ref{hanson-wright} that 
	\begin{align*}
	\tilde{\eta}^{\prime}\tilde{A}^{\prime}\tilde{A}\tilde{\eta} &\leq \text{tr}(\sum_{t=0}^{T-1}\Gamma_t(A)) + \frac{||B||_F^2}{c^{*}||B||} \log{\Big(\frac{1}{\delta}\Big)} \\
	&\leq \text{tr}(\sum_{t=0}^{T-1}\Gamma_t(A)) + \frac{||B||_s}{c^{*}} \log{\Big(\frac{1}{\delta}\Big)} \\
	&\leq  \text{tr}(\sum_{t=0}^{T-1}\Gamma_t(A))\Big(1 + \frac{1}{c^{*}}\log{\Big(\frac{1}{\delta}\Big)}\Big)
	\end{align*}
	with probability at least $1 - \exp{\Big(- \frac{c||B||_F^2}{c^{*}||B||_2^2}\log{\frac{1}{\delta}}\Big)}$. Since 
	\[
	\frac{c||B||_F^2}{c^{*}||B||_2^2} \geq 1
	\]
	it follows that 
	\[
	\exp{\Big(- \frac{c||B||_F^2}{c^{*}||B||_2^2}\log{\frac{1}{\delta}}\Big)} \leq \delta
	\]
	and we can conclude that with probability at least $1-\delta$
	\[
	\tilde{\eta}^{\prime}\tilde{A}^{\prime}\tilde{A}\tilde{\eta} \leq	 \text{tr}(\sum_{t=0}^{T-1}\Gamma_t(A))\Big(1 + \frac{1}{c^{*}}\log{\Big(\frac{1}{\delta}\Big)}\Big)
	\]
\end{proof}
\begin{cor}
	\label{sub_sum}
Whenever $\delta \in (0, e^{-1})$, we have with probability at least $1-\delta$
	\[
	||\sum_{t=k+1}^T X_t X_t^{\prime}||_2 \leq \sigma^2_1(P)\text{tr}(\sum_{t=k}^{T-1}\Gamma_t(A))\Big(1 + \frac{1}{c}\log{\Big(\frac{1}{\delta}\Big)}\Big)
	\]
	for some universal constant $c$.
\end{cor}
\begin{proof}
	The proof follows the same steps as Proposition~\ref{energy_markov}. Define
	\[
		\tilde{A} = \begin{bmatrix} I & 0 & 0 & \hdots & 0 \\
			A & I & 0 & \hdots & 0 \\
			\vdots & \vdots & \ddots & \vdots &\vdots\\
			\vdots & \vdots & \vdots & \ddots&\vdots\\
			A^{T-1} & A^{T-2} & A^{T-3} & \hdots & I
		\end{bmatrix}(I_{n \times n} \otimes P^{\prime})
	\]
	Define $\tilde{A}_k$ as the matrix formed by zeroing out all the rows of $\tilde{A}$ from $k+1$ row onwards. Then observe that 
	\begin{align*}
	||\sum_{t=k+1}^T X_t X_t^{\prime}|| &\leq \text{tr}(\sum_{t=k+1}^T X_t X_t^{\prime}) = \text{tr}(\sum_{t=1}^T X_t X_t^{\prime} - \sum_{t=1}^k X_t X_t^{\prime}) \\
	&= \tilde{\eta}^{\prime}(\tilde{A}^{\prime}\tilde{A} - \tilde{A}^{\prime}_k\tilde{A}_k)\tilde{\eta}
	\end{align*}
	Since $ \text{tr}(\sum_{t=1}^T X_t X_t^{\prime} - \sum_{t=1}^k X_t X_t^{\prime}) \geq 0$ for any $\tilde{\eta}$ it implies $B = (\tilde{A}^{\prime}\tilde{A} - \tilde{A}^{\prime}_k\tilde{A}_k) \succeq 0$. 
	\begin{align*}
	||B||_S &= \text{tr}(\tilde{A}^{\prime}\tilde{A}) = \sum_{t=k}^{T-1}\text{tr}(\Gamma_t(A)) \\
	||B||_F^2 &\leq ||B||_S ||B||_2
	\end{align*}
	Define $c^{*} = \min{(c, 1)}$. Set $t = \frac{||B||_F^2}{c^{*}||B||}  {\log{(\frac{1}{\delta})}}$ and assume $\delta \in (0, e^{-1})$ then
	\begin{align*}
	\frac{t}{c^{*}||B||} \leq \frac{t^2}{c^{*}||B||_F^2}
	\end{align*}
	we get from Theorem~\ref{hanson-wright} that 
	\begin{align*}
	\tilde{\eta}^{\prime}\tilde{A}^{\prime}\tilde{A}\tilde{\eta} &\leq ||B||_S + \frac{||B||_F^2}{c^{*}||B||} \log{\Big(\frac{1}{\delta}\Big)} \leq ||B||_S + \frac{||B||_S}{c^{*}} \log{\Big(\frac{1}{\delta}\Big)} \leq  ||B||_S\Big(1 + \frac{1}{c^{*}}\log{\Big(\frac{1}{\delta}\Big)}\Big)
	\end{align*}
	with probability at least $1 - \exp{\Big(- \frac{c||B||_F^2}{c^{*}||B||_2^2}\log{\frac{1}{\delta}}\Big)}$. Since 
	\[
	\frac{c||B||_F^2}{c^{*}||B||_2^2} \geq 1
	\]
	it follows that 
	\[
	\exp{\Big(- \frac{c||B||_F^2}{c^{*}||B||_2^2}\log{\frac{1}{\delta}}\Big)} \leq \delta
	\]
	and we can conclude that with probability at least $1-\delta$
	\[
	\tilde{\eta}^{\prime}\tilde{A}^{\prime}\tilde{A}\tilde{\eta} \leq	 \text{tr}(\sum_{t=k}^{T-1}\Gamma_t(A))\Big(1 + \frac{1}{c^{*}}\log{\Big(\frac{1}{\delta}\Big)}\Big)
	\]	
\end{proof}
\begin{prop}
	\label{cont_rand}
	Whenever the pdf of $X$, $f(\cdot)$, satisfies  $\text{ess sup}_x f(x) = C_X < \infty$ we have
	\[
	\Pb(|X| \leq \delta) \leq 2C_X \delta
	\]
\end{prop}
\begin{proof}
	Since the essential supremum of $f(\cdot)$ is bounded. Then 
	\[
	\Pb(|X| \leq \delta) = \int_{x=-\delta}^{\delta}f(x)dx \leq 2C_X \delta
	\]
\end{proof}
\begin{prop}[Proposition 2 in~\cite{faradonbeh2017finite}]
	\label{anti_conc}
	Let $P^{-1} \Lambda P = A$ be the Jordan decomposition of $A$ and define $z_T = A^{-T}\sum_{i=1}^TA^{T-i}\eta_i$. Further assume that $\eta_t$ is continuous, subGaussian with variance proxy $=1$ then
	$$\psi(A, \delta) = \sup \Bigg\{y \in \Rb : \Pb\Bigg(\min_{1 \leq i \leq d}|P_i^{'}z_T| < y \Bigg) \leq \delta \Bigg\}$$
	where $P = [P_1, P_2, \ldots, P_d]^{'}$. If $\rho_{\min}(A) > 1$, then 
	$$\psi(A, \delta) \geq \psi(A) \delta > 0$$
where $\psi(A)$ depend only on $A$.
\end{prop}
\begin{proof}
Define the event $\Ec = \{\min_{1 \leq i \leq d}|P_i^{'}z_T| < y\}, \Ec_i = \{|P_i^{'}z_T| < y\}$. Clearly $\Ec \implies \cup_{i=1}^d \Ec_i$, then
\begin{align*}
    \Pb(\Ec) &\leq \Pb( \cup_{i=1}^d \Ec_i) \leq \sum_{i=1}^d \Pb( \Ec_i) 
\end{align*}
From Proposition~\ref{cont_rand} and Assumption~\ref{subgaussian_noise}, we have $\Pb( \Ec_i) \leq 2 C_{|P_i^{'}z_T|} y$. Then we get 
\begin{align*}
\Pb(\Ec) &\leq (2\sum_{i=1}^d  C_{|P_i^{'}z_T|}) y   \leq 2 d \sup_{1 \leq i \leq d}C_{|P_i^{'}z_T|} y
\end{align*}
where $C_{|P_i^{'}z_T|}$ is the essential supremum of the pdf of $|P_i^{'}z_T|$. Then $\psi(A) = \frac{1}{2 d \sup_{1 \leq i \leq d}C_{|P_i^{'}z_T|}}$.
\end{proof}

%% file: content/results_stable.tex
\section{Lower Bound for $Y_T$ when $A \in \Sc_0 \cup \Sc_1$}
\label{short_proof}
Here we will prove our results when $\rho(A) \leq 1+ C/T$. Assume for this case that $\eta_{t} = L\bar{\eta}_t$ where $\{\bar{\eta}_t\}_{t=1}^T$ are i.i.d and all elements of $\bar{\eta}_t$ are independent. Further $L$ is full row rank. Define $\sigma_{\min}(LL^{\prime}) = R^2 > 0$. Let $\sigma_{\max}(LL^{\prime})=1$ (this does not affect our result: $R$ is just the inverse of the condition number).
Define
\begin{align*}
P &= A Y_{T-1} A^{\prime} \\
Q &= \sum_{\tau=0}^{T-1}{Ax_t \eta_{t+1}^{\prime} }\\
V &= TI \\
T_{\eta} &= C\Big(\log{\frac{2}{\delta}} + d \log{5}\Big) \\
\Ec_{1}(\delta) &= \Bigg\{||Q||^2_{(P+V)^{-1}} \leq 8 \log{\Bigg(\dfrac{5^d\text{det}(P+V)^{1/2} \text{det}(V)^{-1/2}}{\delta}\Bigg)}\Bigg\} \\
\Ec_{2}(\delta) &= \Bigg\{||\sum_{\tau=0}^{T-1} Ax_{\tau} x_{\tau}^{\prime}A^{\prime}|| \leq \frac{T \text{tr}(\Gamma_{T}(A) - I)}{\delta}\Bigg\}\\
\Ec_{\eta}(\delta) &= \{T > T_{\eta}(\delta), \frac{3R^2}{4}I \preceq \dfrac{1}{T}\sum_{t=1}^T \eta_t \eta_t^{'} \preceq \frac{5}{4}I\} \\
\Ec(\delta) &= \Ec_{\eta}(\delta) \cap \Ec_{1}(\delta) \cap \Ec_{2}(\delta) \\
\end{align*}
Recall that
\begin{equation}
\label{lb_step1}
{Y}_T \succeq A {Y}_{T-1} A^{\prime} + \sum_{t=0}^{T-1} {Ax_t \eta_{t+1}^{\prime} + \eta_{t+1} x_t^{\prime}A^{\prime}} + \sum_{t=1}^T \eta_t \eta_t^{\prime}
\end{equation}

Our goal here will be to control
\begin{equation}
\label{cross_terms}
||Q||_2
\end{equation}
Following Proposition~\ref{selfnorm_bnd}, Proposition~\ref{energy_markov}, it is true that $\Pb(\Ec_{1}(\delta) \cap \Ec_{2}(\delta)) \geq 1-2\delta$. We will show that 
$$\Ec(\delta) = \Ec_{\eta}(\delta) \cap \Ec_{1}(\delta) \cap \Ec_{2}(\delta) \implies \sigma_{\min}(\hat{Y}_T) \geq 1/4$$
Under $\Ec_{\eta}(\delta)$, we get 
\begin{align}
{Y}_T &\succeq A {Y}_{T-1} A^{\prime} + \sum_{t=0}^{T-1} {Ax_t \eta_{t+1}^{\prime} + \eta_{t+1} x_t^{\prime}A^{\prime}} + \sum_{t=1}^T \eta_t \eta_t^{\prime} \nonumber \\
{Y}_T &\succeq A {Y}_{T-1} A^{\prime} + \sum_{t=0}^{T-1} {Ax_t \eta_{t+1}^{\prime} + \eta_{t+1} x_t^{\prime}A^{\prime}} + \frac{3}{4}R^2TI  \nonumber \\
U^{\prime} {Y}_T U &\geq U^{\prime} A Y_{T-1} A^{\prime} U  + U^{\prime} \sum_{t=0}^{T-1} \Bigg({Ax_t \eta_{t+1}^{\prime} + \eta_{t+1} x_t^{\prime}A^{\prime}} \Bigg) U  + \frac{3}{4}TR^2 \hspace{3mm} \forall U\in \Sc^{d-1} \label{contra_eq}
\end{align} 
Intersecting Eq.~\eqref{contra_eq} with $\Ec_1(\delta) \cap \Ec_2(\delta)$, we find under $\Ec(\delta)$ 
\begin{align*}
&||Q||^2_{(P+V)^{-1}} \leq 8 \log{\Bigg(\dfrac{5^d\text{det}(P+V)^{1/2} \text{det}(V)^{-1/2}}{\delta}\Bigg)} \\
&\leq 8 \log{\Bigg(\dfrac{5^d \text{det}(\frac{T \text{tr}(\Gamma_{T}(A) - I)}{\delta} + TI)^{1/2}\text{det}(TI)^{-1/2}}{\delta}\Bigg)} \\
&\leq  8 \log{\Bigg(\dfrac{5^d \text{det}({ \text{tr}(\Gamma_{T}(A) - I)} + I)^{1/2}}{\delta^d}\Bigg)} 
\end{align*}
Using Proposition~\ref{psd_result_2} and letting $\kappa^2 = U^{\prime} P U$ then 
\begin{align*}
&||QU||_2 \\
&\leq \sqrt{\kappa^2 + T}\sqrt{8 \log{\Bigg(\dfrac{5^d \text{det}({ \text{tr}(\Gamma_{T}(A) - I)} + I)^{1/2}}{\delta^d}\Bigg)}}
\end{align*}
So Eq.~\eqref{contra_eq} implies 
\begin{align}
U^{\prime} {Y}_T U &\geq \kappa^2 \nonumber - \sqrt{(\kappa^2 + T)}{ \sqrt{ 16d \log{( \text{tr}(\Gamma_T - I)+1)} + 32d \log{\frac{5}{\delta}}}} + \frac{3}{4}TR^2 
\end{align}
which gives us
\begin{align}
U^{\prime} \frac{{Y}_T}{T} U &\geq \frac{\kappa^2}{T} - \sqrt{(\frac{\kappa^2}{T} + 1)}\underbrace{ \sqrt{ \frac{16d}{T} \log{( \text{tr}(\Gamma_T - I)+1)} + \frac{32d}{T}\log{\frac{5}{\delta}}}}_{=\beta} + \frac{3}{4}R^2 \label{contra_eq2}
\end{align}
If we can ensure
\begin{equation}
\label{t_req}
\frac{TR^4}{128} \geq  { \frac{d}{2} \log{( \text{tr}(\Gamma_T - I)+1)} + d \log{\frac{5}{\delta}} }
\end{equation}
then $\beta \leq R^2/2$, \textit{i.e.}, 
\[
 \sqrt{ \frac{16d}{T} \log{( \text{tr}(\Gamma_T - I)+1)} + \frac{32d}{T}\log{\frac{5}{\delta}}} \leq \frac{R^2}{2}
\]
Let $T$ be large enough that Eq.~\eqref{t_req} is satisfied then Eq.~\eqref{contra_eq2} implies
\begin{equation}
\label{final_eq}
U^{\prime} \frac{{Y}_T}{T} U \geq \frac{\kappa^2}{T} - \frac{\sqrt{(\frac{\kappa^2}{T}  + 1)}R^2}{2} + \frac{3R^2}{4} \geq \frac{R^2}{4} + \frac{\kappa^2}{2T}
\end{equation}

Since $U$ is arbitrarily chosen Eq.~\eqref{final_eq} implies
\begin{align}
Y_T \succeq \frac{TR^2}{4}I \label{lower_bnd}
\end{align}
with probability at least $1 - 3\delta$ whenever
\begin{align}
\rho_i(A) &\leq 1 + \frac{c}{T} \nonumber\\
T &\geq \max{\Big(C\Big(\log{\frac{2}{\delta}} + d \log{5}\Big), CR^2\Big({ \frac{d}{2} \log{( \text{tr}(\Gamma_T - I)+1)} + d \log{\frac{5}{\delta}} }\Big)\Big)} \label{t_req_comb}
\end{align}
\begin{remark}
Eq.~\eqref{t_req} is satisfied whenever $\text{tr}(\Gamma_T - I)$ grows at most polynomially in $T$. This is true whenever $\rho(A) \leq 1 +\frac{c}{T}$. 
\end{remark}

%% file: content/sharp_bnds.tex
\section{Sharpened bounds when $1 - \frac{c}{T}\leq \rho_i(A)  \leq 1 + \frac{c}{T}$}
\label{sharp_bounds}
Here we show that the bound for $Y_T$ in Eq.~\eqref{lower_bnd} can be sharpened to have quadratic growth in $T$. The key idea towards sharpening will be that we want Eq.~\eqref{lower_bnd} satisfied for every $t \geq \frac{T}{2}$ simultaneously, \textit{i.e.}, we need 
\begin{align}
Y_t \succeq \frac{tR^2}{4}I \label{lower_bnd_t}
\end{align}
simultaneously for $t \geq \frac{T}{2}$ with high probability. By similar arguments as before as long as we have
\begin{align}
\rho_i(A) &\leq 1 \nonumber\\
t &\geq \max{\Big(C\Big(\log{\frac{2}{\delta}} + d \log{5}\Big), CR^2\Big({ \frac{d}{2} \log{( \text{tr}(\Gamma_t - I)+1)} + d \log{\frac{5}{\delta}} }\Big)\Big)} \label{t_req_comb_t}
\end{align}
we can conclude with probability at least $1 - 2\delta$ that $Y_t \succeq \frac{tR^2}{4}I$. This means that with probability at least $1 - 3\delta \frac{T}{2}$ we have for $t \geq \frac{T}{2}$ simultaneously
\[
Y_t \succeq \frac{tR^2}{4}I
\]
when Eq.~\eqref{t_req_comb_t} is satisfied for each $t$. Since the LHS of Eq.~\eqref{t_req_comb_t} is least at $t = T/2$ and RHS is greatest at $t=T$, a sufficient condition for every $t \geq \frac{T}{2}$ satisfying Eq.~\eqref{t_req_comb_t} is the following 
\[
T \geq \max{\Big(C\Big(\log{\frac{2}{\delta}} + d \log{5}\Big), C\Big({ \frac{d}{2} \log{( \text{tr}(\Gamma_T - I)+1)} + d \log{\frac{5}{\delta}} }\Big)\Big)} 
\]
Then by substituting $\delta \rightarrow \frac{2\delta}{3T}$ we can conclude with probability at least $1-\delta$ that 
\[
Y_t \succeq \frac{tR^2}{4}I
\]
simultaneously for every $t \geq \frac{T}{2}$ whenever 
\begin{equation}
T \geq \max{\Big(C\Big(\log{\frac{3T}{2\delta}} + d \log{5}\Big), CR^2\Big({ \frac{d}{2} \log{( \text{tr}(\Gamma_T - I)+1)} + d \log{\frac{15T}{2\delta}} }\Big)\Big)} \label{T_sharp_cond}
\end{equation}
Define $\gamma_{t-1} = \sqrt{U^{\prime} A^{\prime} Y_{t-1} A U}$ and Eq.~\eqref{final_eq} becomes
\begin{align}
U^{\prime} Y_t U &\geq \gamma_{t-1}^2 - \sqrt{(\gamma_{t-1}^2 + t)}\underbrace{ \sqrt{ {16d} \log{( \text{tr}(\Gamma_t - I)+1)} + {32d}\log{\frac{15T}{2\delta}}}}_{\text{Under Eq.~\eqref{T_sharp_cond} is}\leq \frac{R^2\sqrt{t}}{2}} + \frac{3}{4}tR^2 \nonumber \\
&\geq \gamma_{t-1}^2  -(\gamma_{t-1} + \sqrt{t})\sqrt{ {16d} \log{( \text{tr}(\Gamma_t - I)+1)} + {32d}\log{\frac{15T}{2\delta}}} + \frac{3t}{4}R^2 \nonumber \\
&\geq \gamma_{t-1}^2  -\gamma_{t-1} \sqrt{ {16d} \log{( \text{tr}(\Gamma_t - I)+1)} + {32d}\log{\frac{15T}{2\delta}}} + \frac{3tR^2}{4} -  \sqrt{t} \underbrace{\sqrt{ {16d} \log{( \text{tr}(\Gamma_t - I)+1)} + {32d}\log{\frac{15T}{2\delta}}}}_{\leq R^2\frac{\sqrt{t}}{2}} \nonumber \\
&\geq \gamma_{t-1}^2 \Big(1 - \sqrt{\frac{{ {16d} \log{( \text{tr}(\Gamma_t - I)+1)} + {32d}\log{\frac{15T}{2\delta}}}}{\gamma_{t-1}^2}}\Big) + \frac{tR^2}{4} \nonumber \\
&\geq \gamma_{t-1}^2 \Big(1 - \underbrace{\sqrt{\frac{{ {16d} \log{( \text{tr}(\Gamma_T - I)+1)} + {32d}\log{\frac{15T}{2\delta}}}}{\gamma_{t-1}^2}}}_{=\sqrt{\frac{c(A, \delta)}{\gamma_{t-1}^2}}}\Big) + \frac{TR^2}{8} \label{contra_eq3}
\end{align}
Observe that
\begin{equation}
\gamma_{t-1} = \sqrt{U^{\prime} A^{\prime} Y_{t-1} A U}\geq \sigma_{\min}(A)\sqrt{\frac{TR^2}{8e}} \label{gamma_t}
\end{equation}
Eq.~\eqref{contra_eq3} will give us a non--trivial bound only when $\frac{c(A, \delta)}{\gamma_{t-1}^2} \leq 1/4$ which is true whenever
\begin{equation}
T \geq \frac{64ec(A, \delta)}{R^2\sigma_{\min}^2(A)} \label{bet_cond}
\end{equation}
The scaling $1 - \sqrt{\frac{c(A,\delta)}{\gamma_{t-1}^2 }}$ in Eq.~\eqref{contra_eq3} depends on $\gamma_{t-1}$ itself. We will show that 
\begin{align*}
&\gamma^2_{t-1} = T\Omega(1) \implies \gamma^2_{t-1} = T\Omega\Big(\sqrt{\frac{T}{c(A, \delta)}}\Big) \\ 
&\gamma^2_{t-1} = T\Omega\Big(\Big(\frac{T}{c(A, \delta)}\Big)^{1/2}\Big) \implies  \gamma^2_{t-1} = T\Omega\Big(\Big(\frac{T}{c(A, \delta)}\Big)^{3/4}\Big) \\
&\gamma^2_{t-1} = T\Omega\Big(\Big(\frac{T}{c(A, \delta)}\Big)^{\frac{2^k-1}{2^k}}\Big) \implies  \gamma^2_{t-1} = T\Omega\Big(\Big(\frac{T}{c(A, \delta)}\Big)^{\frac{2^{k+1} - 1}{2^{k+1}}}\Big)\\
&\implies \hdots \implies \gamma^2_{t-1} = T\Omega\Big(\frac{T}{c(A, \delta)}\Big) 
\end{align*}
From Eq.~\eqref{contra_eq3},\eqref{gamma_t} since 
\[
\sqrt{\frac{c(A,\delta)}{\gamma_{t-1}^2 }} \leq \sqrt{\frac{16ec(A,\delta)}{\sigma_{\min}(AA^{\prime}) T}} = \beta_1
\]
it follows that
\begin{align}
{Y_t} &{\succeq} \Bigg(1 - \underbrace{\sqrt{\frac{16ec(A,\delta)}{\sigma_{\min}(AA^{\prime}) TR^2}}}_{=\beta_1}\Bigg)A {Y_{t-1}} A^{\prime} + \frac{R^2TI}{8}\label{recur_eq}
\end{align}
The goal here is to refine the upper bound for $\sqrt{\frac{c(A,\delta)}{\gamma_{t-1}^2 }}$ such that 
\[
\sqrt{\frac{c(A,\delta)}{\gamma_{t-1}^2 }} \leq \frac{C}{T}
\]
Eq.~\eqref{recur_eq} implies that 
\begin{align*}
Y_t &\overset{(a)}{\succeq} \frac{TR^2}{8}\sum_{k=1}^{\min{(\lfloor \frac{1}{\beta_1}\rfloor, \frac{T}{4})}} (1 - \beta_1)^k A^k A^{k \prime} + \frac{R^2TI}{16} \\
&\overset{(b)}{\succeq} \frac{TR^2}{16e}\sum_{k=1}^{\min{(\lfloor \frac{1}{\beta_1}\rfloor, \frac{T}{4})}}  A^k A^{k \prime} + \frac{R^2TI}{16} \\
&{\succeq} \frac{R^2T}{16e} \Gamma_{\lfloor \frac{1}{\beta_1}\rfloor}(A) + \frac{R^2TI}{16}
\end{align*}
Here 
\begin{equation}
\label{beta}
\beta_1 = \sqrt{\frac{16ec(A,\delta)}{\sigma_{\min}(AA^{\prime}) R^2T}}
\end{equation}
Due to the choice of $T, d$ we will usually have $\lfloor \frac{1}{\beta_1}\rfloor^2 \leq \frac{T}{4}$. $(a)$ follows by successively expanding Eq.~\eqref{recur_eq}, $(b)$ follows because $(1 - \beta_1)^{\lfloor \frac{1}{\beta_1}\rfloor} \geq \frac{e^{-1}}{2}$ since $\beta_1 \leq 1/2$ by Eq.~\eqref{bet_cond}. Then we can conclude that 
\begin{align}
\gamma_{t-1}^2 &\geq {\sigma_{\min}(AY_tA^{\prime})} \nonumber \\
&\geq \frac{R^2T\sigma_{\min}(A A^{\prime})\sigma_{\min}(\Gamma_{\lfloor \frac{1}{\beta_1}\rfloor}(A)) }{16e} \label{beta_inter} 
\end{align}
which gives us
\begin{align}
\sqrt{\frac{c(A, \delta)}{\gamma_{t-1}^2}} &\leq \Big(\frac{ 16ec(A, \delta)}{ R^2T\sigma_{\min}(A A^{\prime})\sigma_{\min}(\Gamma_{\lfloor \frac{1}{\beta_1}\rfloor}(A))}\Big)^{1/2}=\beta_2  \label{recursion}
\end{align}
It is clear from Eq.~\eqref{recursion} that we get a recursion during the refinement process. Specifically at the $k^{th}$ repetition of Eq.~\eqref{recur_eq} up to Eq.~\eqref{recursion} we get,
\begin{equation}
\label{betak_rec}
\beta_k =   \Big(\frac{ 16ec(A, \delta)}{ R^2T\sigma_{\min}(A A^{\prime})\sigma_{\min}(\Gamma_{\lfloor \frac{1}{\beta_{k-1}}\rfloor}(A))}\Big)^{1/2}
\end{equation}
Now $\beta_k$ is a non-increasing sequence. We show this by induction. Since $\sigma_{\min}(\Gamma_t(A)) \geq 1$ and 
\[
\sqrt{\frac{16ec(A,\delta)}{\sigma_{\min}(AA^{\prime}) R^2T}} \leq 1
\]
it follows trivially that $\beta_2 \leq \beta_1$. Assume our hypothesis holds for all $k \leq m$. Then since $\Gamma_{t_1}(A) \succeq \Gamma_{t_2}(A)$ whenever $t_1 \geq t_2$ we have
\begin{align*}
\Big(\frac{ 16ec(A, \delta)}{ R^2T\sigma_{\min}(A A^{\prime})\sigma_{\min}(\Gamma_{\lfloor \frac{1}{\beta_{m}}\rfloor}(A))}\Big)^{1/2} &\leq \Big(\frac{ 16ec(A, \delta)}{R^2 T\sigma_{\min}(A A^{\prime})\sigma_{\min}(\Gamma_{\lfloor \frac{1}{\beta_{m-1}}\rfloor}(A))}\Big)^{1/2} \\
\beta_{m+1} &\leq \beta_{m}
\end{align*}
and we have proven our hypothesis. To now find the best upper bound for $\sqrt{\frac{c(A, \delta)}{\gamma_{t-1}^2}}$ we find the steady state solution for Eq.~\eqref{betak_rec}, \textit{i.e.}
\begin{equation}
\beta_0^2\sigma_{\min}(\Gamma_{\lfloor \frac{1}{\beta_0}\rfloor}(A))  =   \Big(\frac{ 16ec(A, \delta)}{ R^2T\sigma_{\min}(A A^{\prime})}\Big) \label{final_sol}
\end{equation}
Now a solution for $\beta_0 \in (\frac{2C}{\sigma_{\min}(A A^{\prime})TR^2}, 1)$. To see this set $\beta_0 = 1$, then LHS $>$ RHS. Next set $\beta_0 = \frac{2C}{\sigma_{\min}(A A^{\prime})TR^2}$ then since $\rho_{\min}(A^t) \geq \sigma_{\min}(A^{t})$ and $\rho_i \leq 1+C/T$ we see that 
\begin{align*}
\frac{4C^2\sigma_{\min}(\Gamma_{\lfloor \frac{1}{\beta_0}\rfloor}(A))}{\sigma_{\min}(A A^{\prime})^2 T^2} &\leq \frac{4\sum_{t=0}^{\sigma_{\min}(A)^2R^2T/2C} \rho_{\min}(A)^{2t}}{R^4\sigma_{\min}(A A^{\prime})^2T^2/C^2} \\
&\leq \frac{2eC}{\sigma_{\min}(A)^2T}\leq \Big(\frac{ 16ec(A, \delta)}{ R^2T\sigma_{\min}(A A^{\prime})}\Big)
\end{align*}
and LHS $<$ RHS because $C$ is a constant but $c(A, \delta)$ is growing logarithmically with $T$ (and we can pick $T$ accordingly). By ensuring that 
\begin{align*}
T \geq \frac{64ec(A, \delta)}{R^2\sigma_{\min}(A)^2}
\end{align*}
we also ensure that $\beta_1 < 1/2$ and as a result all subsequent $\beta_k < 1/2$. Now we can conclude that whenever $T \geq \frac{64ec(A, \delta)}{\sigma_{\min}(A)^2}$ we get  Eq.~\eqref{recur_eq} 
\begin{equation}
\label{recur_eq2}
{Y_t} {\succeq} (1 - \beta_0 )A {Y_{t-1}} A^{\prime} + \frac{TR^2I}{8}
\end{equation}
and following as before we get with probability at least $1-\delta$
\begin{align}
Y_T &\succeq \frac{TR^2}{16e} \Gamma_{\lfloor \frac{1}{\beta_0}\rfloor}(A) + \frac{TR^2I}{16} \label{stable_yt}
\end{align}
where $\beta_0$ is solution to 
\[
\beta_0^2\sigma_{\min}(\Gamma_{\lfloor \frac{1}{\beta_0}\rfloor}(A))  =   \Big(\frac{ 16ec(A, \delta)}{ TR^2\sigma_{\min}(A A^{\prime})}\Big)
\]
and 
\[
c(A, \delta) = { {16d} \log{( \text{tr}(\Gamma_T - I)+1)} + {32d}\log{\frac{15T}{2\delta}}}
\]

It should be noted that $\frac{1}{\beta_0}$ will equal $\frac{\sqrt{\alpha(d)}TR^2\sigma_{\min}(A A^{\prime})}{16ec(A, \delta)}$, \textit{i.e.}, grow linearly with $T$, as shown in Proposition~\ref{gramian_lb}. Then it can be seen from Eq.~\eqref{stable_yt} that 
\begin{align}
Y_T &\succeq \frac{TR^2}{16e} \Gamma_{\lfloor \frac{1}{\beta_0}\rfloor}(A) + \frac{TR^2I}{16}  \nonumber\\
Y_T &\succeq \frac{TR^2}{16e} \sigma_{\min}(\Gamma_{\lfloor \frac{1}{\beta_0}\rfloor}(A)) + \frac{TR^2I}{16}\nonumber \\
&\succeq \frac{TR^2}{16e} \frac{TR^2\sqrt{\alpha(d)}\sigma_{\min}(A A^{\prime})}{16ec(A, \delta)C(d)} I = \frac{\sqrt{\alpha(d)} T^2R^4\sigma_{\min}(A A^{\prime})}{256e^2c(A, \delta)} \label{quad}
\end{align}

%% file: content/results_explosive.tex
\section{Invertibility of $Y_T$ in explosive systems}	
\label{explosive}
Assume for this case that $\eta_{t} = L \bar{\eta}_t$ where $\{\bar{\eta}_t\}_{t=1}^T$ are i.i.d and all elements of $\bar{\eta}_t$ are independent. Further $L$ is full row rank. Define $\sigma_{\min}(LL^{\prime}) = R^2 > 0$. Let $\sigma_{\max}(LL^{\prime})=1$. Recall that 
\begin{align*}
z_t &= A^{-t}x_t\\
&= x_0 + \sum_{\tau=1}^{t} A^{-\tau} \eta_{\tau}
\end{align*}
Define
\begin{align*}
z(T, t) &= \Bigg(\sum_{s=0}^{t-1}A^{-s} \eta_{T+1-t+s}\Bigg)
\end{align*}
where $z(T, t) = 0$ for $t \leq 0, t \geq T+1$. An observation that will be useful is that $z(t)$ is statistically independent of $z(T) - z(t)$. Recall that $U_T = A^{-T}\sum_{t=1}^T x_t x_t^{\prime}A^{-T \prime}, F_T = \sum_{t=1}^T A^{-t+1} z_T z_T^{\prime}A^{-t+1 \prime}$
\subsection*{Bounding $||F_T - U_T||_{\text{op}}$}
	Observe that 
	\begin{align}
	z(T) - z(T-t) &= A^{-T+t-1}\Bigg(\sum_{s=0}^{t-1}A^{-s} \eta_{T+1-t+s}\Bigg) =A^{-T+t-1} z(T, t) \label{zt_diff}
	\end{align}
	Then
	\begin{align*}
	|| U_T - F_T||_{\text{op}} = &||\sum_{t=1}^{T} A^{-t}(z(T-t)z(T-t)^{'} - z(T)z(T)^{'})(A^{-t})^{'}||_2
	\end{align*}
	Let $u = z(T-t), v=z(T)$ and since $uu^{\prime} - vv^{\prime}= (u-v)u^{\prime} + u(u-v)^{\prime} - (u-v)(u-v)^{\prime}$ we have
	\begin{align}
	|| U_T - F_T||_{\text{op}} &\leq ||\sum_{t=1}^{T} A^{-t}(z(T-t) - z(T))(z(T-t) - z(T))^{'}A^{-t'}||_2 \nonumber\\
	&+||\sum_{t=1}^{T} A^{-t} ((z(T-t) - z(T))z(T-t)^{'} + z(T-t)(z(T-t)^{'} - z(T)^{'}) A^{-t '}||_2 \label{ft_ut}
	\end{align}
	The reason we decompose it in such a way is so that we can represent the cross terms $(z(T-t) - z(T))z(T-t)^{\prime}$ as the product of independent terms. This will be useful in using Hanson--Wright bounds as we show later.
	
	First we bound 
	$$ ||\sum_{t=1}^{T} A^{-t}(z(T-t) - z(T))(z(T-t) - z(T))^{'}A^{-t'}||_2$$ 
	
	From Eq.~\eqref{zt_diff} we see that $A^{-t}(z(T-t) - z(T)) = -A^{-T-1}z(T, t)$, then 
	\begin{align*}
	A^{-T-1}z(T, t) &= A^{-T-1}[0, 0, \ldots, \underbrace{I}_{T-t+1 \text{ term}}, A^{-1}, A^{-2} , \ldots, A^{-t+1}] \begin{bmatrix}
	\eta_1 \\ 
	\eta_2 \\
	\vdots \\
	\eta_T
	\end{bmatrix}
	\end{align*}
	Since $\sum_{t=1}^{T} (z(T-t) - z(T))(z(T-t) - z(T))^{'} \preceq \sum_{t=1}^{T} \text{trace}((z(T-t) - z(T))(z(T-t) - z(T))^{'}) I$. Based on these observations we have
	\begin{align*}
	&||\sum_{t=1}^{T} A^{-t}(z(T-t) - z(T))(z(T-t) - z(T))^{'}A^{-t'}||_2	= ||\sum_{t=1}^{T} A^{-T-1} z(T, t)z(T, t)^{'} A^{-T-1 '}||_2 \\
	&\leq  \text{trace}(A^{-T-1}\sum_{t=1}^{T} z(T, t)z(T, t)^{'}  A^{-T-1 '}) = \sum_{t=1}^{T} z(T, t)^{'}  A^{-T-1 '} A^{-T-1} z(T, t) = \tilde{\eta}^{'} \tilde{A}^{'} \tilde{A} \tilde{\eta} \\
	\end{align*}
	where $\tilde{\eta} = \begin{bmatrix}
	\eta_1 \\ 
	\eta_2 \\
	\vdots \\
	\eta_T
	\end{bmatrix}$ and 
	\[
	\tilde{A} = \begin{bmatrix}
	0 & 0 &\ldots & 0 & A^{-T-1} \\
	0 & 0 & \ldots & A^{-T-1} & A^{-T-2} \\
	\vdots & \vdots & \vdots & \vdots & \vdots \\
	A^{-T-1} & A^{-T-2} & \ldots & A^{-2T+1} & A^{-2T}
	\end{bmatrix}
	\]
	Since $\text{tr}(\tilde{A}\tilde{A}^{\prime}) =T\text{tr}(A^{-T-1}\Gamma_T(A^{-1})A^{-T-1 \prime})$. Applying Markov's Inequality (See Proposition~\ref{energy_markov}), we have with probability at least $1-\delta$ that 
	\begin{align}
	\tilde{\eta}^{'} \tilde{A}^{'} \tilde{A} \tilde{\eta} &\leq \frac{\text{tr}(\Ex[\tilde{A} \tilde{\eta}\tilde{\eta}^{'} \tilde{A}^{'}])}{\delta}  \leq \frac{\sigma_1(L)^2T\text{tr}(A^{-T-1}\Gamma_T(A^{-1})A^{-T-1 \prime})}{\delta} \label{final_diff}
	\end{align}
	Although this bound can be tightened by dependent Hanson--Wright (See Corollary~\ref{dep-hanson-wright}), there is no reason to do so as $\delta$ depends only logarithmically on $T$. In fact we get with probability at least $1 -\delta$ that 
	\begin{equation}
	    \tilde{\eta}^{'} \tilde{A}^{'} \tilde{A} \tilde{\eta}  \leq \Big(1 + \frac{1}{c} \log{\frac{1}{\delta}}\Big)(\sigma_1(L)^2T\text{tr}(A^{-T-1}\Gamma_T(A^{-1})A^{-T-1 \prime})) \label{tight_final_diff}
	\end{equation}
	
	Next we analyze the second term
	\begin{align*}
	&||\sum_{t=1}^{T} A^{-t} ((z(T-t) - z(T))z(T-t)^{'} + z(T-t)(z(T-t)^{'} - z(T)^{'}) A^{-t '}||_2
	\end{align*}
	Consider the summand $\sum_{t=1}^{T} A^{-t} ((z(T-t) - z(T))z(T-t)^{'} A^{-t \prime}$, then
	\begin{align}
	\label{cross_summand}
	\sum_{t=1}^{T} A^{-t} ((z(T-t) - z(T))z(T-t)^{'}A^{-t \prime} &= A^{-T-1}\sum_{t=1}^{T} z(T, t) z(T-t)^{'}A^{-t \prime}
	\end{align}
	We define scaled version of $z(T, t), z(T-t)$.
	\begin{align*}
	\tilde{z}(T, t)&= A^{-T-1}z(T, t) = A^{-T-1}\underbrace{[0, 0, \ldots, \underbrace{I}_{T-t+1 \text{ term}}, A^{-1}, A^{-2} , \ldots, A^{-t+1}]}_{A(T, t)} \begin{bmatrix}
	\eta_1 \\ 
	\eta_2 \\
	\vdots \\
	\eta_T
	\end{bmatrix} \\
 \tilde{z}(T-t)^{\prime} &= z(T-t)^{\prime}A^{-t \prime} = \underbrace{[\eta_1^{\prime}, \eta_2^{\prime}, \ldots, \eta_T^{\prime}]}_{\tilde{\eta}^{\prime}}\underbrace{\begin{bmatrix}
		A^{-t-1 \prime} \\ 
		A^{-t-2 \prime} \\
		\vdots \\
		A^{-T \prime} \\
		0 \\
		\vdots \\
		0
		\end{bmatrix}}_{A(T-t)^{\prime}} + x_0
	\end{align*}
	Then the probability of the second term can be written as 
	\begin{align}
	&\Pb(||\sum_{t=1}^{T} (\tilde{z}(T, t)\tilde{z}(T-t)^{'} + \tilde{z}(T-t)\tilde{z}(T, t)^{'})||_2 \geq z) \underbrace{\leq}_{\frac{1}{2}-\text{net}} 2 \times 5^{2d} \times\Pb(\bl \sum_{t=1}^{T} 2u^{'} \tilde{z}(T, t)\tilde{z}(T-t)^{'}v \bl ) \geq z/4) \nonumber \\
	&\leq 2 \times 5^{2d} \times \Pb\Bigg(\bl \tilde{\eta}^{\prime}{\Big(\sum_{t=1}^{T} A(T, t)^{\prime}A^{-T-1 \prime} uv^{\prime} A(T-t) + A(T-t)^{\prime} vu^{\prime}A^{-T-1} A(T, t)\Big)} \tilde{\eta} \bl \leq z/4 \Bigg) \label{cross}
	\end{align}
	To Eq.~\eqref{cross} apply Hanson-Wright inequality. For any $u, v$, due to the statistical independence of $z(T-t), z(T, t)$ we have
	$$\Ex[\sum_{t=1}^{T} 2u^{'} \tilde{z}(T, t)\tilde{z}(T-t)^{'}v] = 0$$
	
	We now need an upper bound on $||S||_2, ||S||_F$. Since $CD^{\prime} + DC^{\prime} \preceq CC^{\prime} + DD^{\prime}$
	
	\begin{align*}
	S &= \sum_{t=1}^{T} A(T, t)^{\prime}A^{-T-1 \prime} uv^{\prime} A(T-t) + A(T-t)^{\prime} vu^{\prime}A^{-T-1} A(T, t) \\	
	&=\sum_{t=1}^{T} \underbrace{A(T, t)^{\prime}A^{-(T+1)\epsilon \prime}}_{=C} \underbrace{A^{-(T+1)(1-\epsilon) \prime} uv^{\prime} A(T-t)}_{=D^{\prime}} + A(T-t)^{\prime} vu^{\prime}A^{-(T+1)(1-\epsilon)}A^{-(T+1)\epsilon} A(T, t)\\
	&\preceq \sum_{t=1}^{T} \underbrace{A(T, t)^{\prime}A^{-(T+1)\epsilon \prime}A^{-(T+1)\epsilon}A(T, t) }_{=CC^{\prime}}
	+\sum_{t=1}^{T} \underbrace{A(T-t)^{\prime} v u^{\prime} A^{-(T+1)(1-\epsilon)} A^{-(T+1)(1-\epsilon) \prime} uv^{\prime}  A(T-t)}_{=DD^{\prime}}   \\
	&\preceq \sigma_1^2(A^{-(T+1)\epsilon}) \sum_{t=1}^{T} A(T, t)^{\prime}A(T, t) + u^{\prime} A^{-(T+1)(1-\epsilon)} A^{-(T+1)(1-\epsilon) \prime} u \sum_{t=1}^{T} A(T-t)^{\prime} v v^{\prime}  A(T-t) \\
	&\preceq \sigma_1^2(A^{-(T+1)\epsilon}) \text{tr}\Big(\sum_{t=1}^{T} A(T, t)^{\prime}A(T, t)\Big)I + \sigma_1^{2}(A^{-(T+1)(1-\epsilon)}) \text{tr}\Big(\sum_{t=1}^{T} A(T-t)^{\prime} v v^{\prime}  A(T-t)\Big)I \\
	&\overset{(a)}{\preceq} 2T\sigma_1^2(A^{-(T+1)\epsilon}) \text{tr}(\Gamma_T(A^{-1})) I 
	\end{align*}
	Here $(a)$ follows because 
	$$A(T, t) A(T, t)^{\prime} = \Gamma_{t-1}(A), A(T-t)A(T-t)^{\prime} = \Gamma_{T-t}(A)$$

	Then whenever 
	\begin{equation}
	\label{T_req}
	T \geq T_0=\frac{2}{c}\Big(\log{\frac{1}{\delta} + \log{2} + 2d \log{5}}\Big)
	\end{equation}	
	Eq.~\eqref{cross} becomes with probability at least $1-\delta$ that
	\begin{equation}
	||\sum_{t=1}^{T} ((z(T-t) - z(T))z(T-t)^{'} + z(T-t)(z(T-t)^{'} - z(T)^{'})||_2 \leq 4T^2 \sigma_1^2(A^{-(T+1)\epsilon}) \text{tr}(\Gamma_T(A^{-1})) \label{cross_bound}
	\end{equation} 
	Then combining Eq.~\eqref{final_diff},\eqref{cross_bound} we get for $T \geq T_0$ given in Eq.~\eqref{T_req},  
	\begin{equation}
	||U_T - F_T||_{2} \leq \Bigg(4T^2 \sigma_1^2(A^{-(T+1)\epsilon}) \text{tr}(\Gamma_T(A^{-1}))+ \frac{T\text{tr}(A^{-T-1}\Gamma_T(A^{-1})A^{-T-1 \prime})}{\delta}\Bigg) \label{error_cum}
	\end{equation}
	with probability at least $1-2\delta$. We pick $\epsilon$ such that $(T+1)\epsilon = \lfloor \frac{T+1}{2} \rfloor$. In fact using Eq.~\eqref{tight_final_diff} instead of Eq.~\eqref{final_diff} we get 
	\begin{equation}
	||U_T - F_T||_{2} \leq \Bigg(4T^2 \sigma_1^2(A^{-(T+1)\epsilon}) \text{tr}(\Gamma_T(A^{-1}))+ \Big(1 + \frac{1}{c}\log{\frac{1}{\delta}}\Big)T\text{tr}(A^{-T-1}\Gamma_T(A^{-1})A^{-T-1 \prime})\Bigg) \label{tight_error_cum}
	\end{equation}
\subsection*{Bounding $U_T$}
	To give lower and upper bounds on $U_T$, we need to bound $F_T$. The steps involve
	\begin{align*}
	||U_T - F_T||_2 &\leq \Delta \\
	F_T &\succeq  V_{dn} \succ 0\\
	\implies U_T &\geq V_{dn} - \Delta I \\
	F_T &\preceq V_{up} \\
	\implies U_T &\preceq V_{up} + \Delta I \\
 	\end{align*}

	From Proposition~\ref{ft_inv} we get, with probability at least $1-2\delta$, 
	\begin{align*}
		F_T &\succeq \phi_{\min}(A)^2 \psi(A)^2 \delta^2 \sigma_{\min}(P^{-1})^2I \\
		F_T &\preceq \frac{\phi_{\max}(A)^2}{\sigma_{\min}(P)^2}(1+\frac{1}{c}\log{\frac{1}{\delta}})\text{tr}(P(\Gamma_T(A^{-1})-I)P^{\prime}) I 
	\end{align*}
	
	Define 
	\begin{align*}
	\Delta &= \frac{1}{2}\min{ \Bigg(\frac{\phi_{\max}(A)^2}{\sigma_{\min}(P)^2}(1+\frac{1}{c}\log{\frac{1}{\delta}})\text{tr}(P(\Gamma_T(A^{-1})-I)P^{\prime}), \phi_{\min}(A)^2 \psi(A)^2 \delta^2 \sigma_{\min}(P^{-1})^2\Bigg)} \\
	&= \frac{\phi_{\min}(A)^2 \psi(A)^2 \delta^2 \sigma_{\min}(P^{-1})^2}{2}
	\end{align*}
	Then in Eq.~\eqref{error_cum} by ensuring that
	\begin{align}
	&\Bigg(4T^2 \sigma_1^2(A^{-(T+1)\epsilon}) \text{tr}(\Gamma_T(A^{-1})) + \frac{T\text{tr}(A^{-T-1}\Gamma_T(A^{-1})A^{-T-1 \prime})}{\delta}\Bigg) \leq \frac{\phi_{\min}(A)^2 \psi(A)^2 \delta^2}{2\sigma_{\max}(P)^2}  \nonumber
	\end{align}
	we get with probability at least $1-4\delta$ (since this is the intersection of events governed by Eq.~\eqref{error_cum},\eqref{lb_ft},\eqref{ub_ft})
	\begin{align}
	U_T &\succeq \phi_{\min}(A)^2 \psi(A)^2 \delta^2 \sigma_{\min}(P^{-1})^2I - \frac{\phi_{\min}(A)^2 \psi(A)^2 \delta^2}{2\sigma_{\max}(P)^2}  I \succeq \frac{\phi_{\min}(A)^2 \psi(A)^2 \delta^2}{2\sigma_{\max}(P)^2}  I \label{ut-ft}	
	\end{align}
	Similarly, for the upper bound
	\begin{equation}
	U_T \preceq \frac{3\phi_{\max}(A)^2}{2\sigma_{\min}(P)^2}(1+\frac{1}{c}\log{\frac{1}{\delta}})\text{tr}(P(\Gamma_T(A^{-1})-I)P^{\prime}) I \label{ub-ft}
	\end{equation}
	
Thus with probability at least $1-4\delta$ we have 
	\begin{align}
Y_T&\succeq \frac{\phi_{\min}(A)^2 \psi(A)^2 \delta^2}{2\sigma_{\max}(P)^2}A^T A^{T \prime}\nonumber\\  
Y_T&\preceq\frac{3\phi_{\max}(A)^2}{2\sigma_{\min}(P)^2}(1+\frac{1}{c}\log{\frac{1}{\delta}})\text{tr}(P(\Gamma_T(A^{-1})-I)P^{\prime}) A^T A^{T \prime} \label{exp_bnds}
	\end{align}
whenever 
\begin{align}
&\Bigg(4T^2 \sigma_1^2(A^{-(T+1)\epsilon}) \text{tr}(\Gamma_T(A^{-1})) + \frac{T\text{tr}(A^{-T-1}\Gamma_T(A^{-1})A^{-T-1 \prime})}{\delta}\Bigg) \leq \frac{\phi_{\min}(A)^2 \psi(A)^2 \delta^2}{2\sigma_{\max}(P)^2}\label{t_exp_req}
\end{align}

%% file: content/regularity.tex
\section{Regularity and Invertibility}
\label{regularity_inv}
Through a counterexample in~\cite{nielsen2008singular}, Remark 4 in~\cite{phillips2013inconsistent} it is shown that unless a matrix is regular, the estimation of the parameters maybe asymptotically inconsistent. 

Recall $F_T$ from Eq.~\eqref{ut_ft}. Assume again that $\eta_{t} = L \bar{\eta}_t$ where $\{\bar{\eta}_t\}_{t=1}^T$ are i.i.d isotropic subGaussian and all elements of $\bar{\eta}_t$ are independent. Further $L$ is full row rank. Define $\sigma_{\min}(LL^{\prime}) = R^2 > 0$. Let $\sigma_{\max}(LL^{\prime})=1$ (this does not affect the main result as it appears only as a scaling). For the invertibility of $Y_T$ in explosive systems, it will be important that $F_T$ is invertible with high probability. It will turn out that invertibility of $F_T$ can be ensured by assuming regularity of $A$. This is Proposition 1 in~\cite{faradonbeh2017finite} and has been presented here for completeness. It will be useful to recall the definitions of $\phi_{\min}(A), \phi_{\max}(A)$ from Definition~\ref{outbox}.

\vspace{2mm}
We will show $F_{T}$ indeed has rank $d$ with probability $1$. Formally,
\begin{prop}
	\label{ft_inv}
	Let $A$ be regular, then we have with probability at least $1 - 2\delta$  
	\begin{align*}
	\sigma_{\min}(F_T) &\geq \frac{\phi_{\min}(A)^2}{\sigma_{\max}(P)^2} \psi(A)^2 \delta^2 \\
	\sigma_{\max}(F_T) &\leq \frac{\phi_{\max}(A)^2}{\sigma_{\min}(P)^2}(1+\frac{1}{c}\log{\frac{1}{\delta}})\text{tr}(P(\Gamma_T(A^{-1})-I)P^{\prime})
	\end{align*}
	where $A = P^{-1}\Lambda P$ is the Jordan decomposition of $A$.
\end{prop}
\begin{proof}
	Let $S_k = [z_T, A^{-1}z_T, \ldots, A^{-k}z_T]$ where $z_T = A^{-T}x_T = A^{-T}(\sum_{k=0}^{T-1} A^{k} L \bar{\eta}_{T-k})$. Note that $L \bar{\eta}_t$ is continuous whenever $L$ is full row rank. Then $F_T = S_T S_T^{\prime}$. Observe that 
	\[
	A^{-t}z_{T} = P^{-1} \Lambda^{-t} P z_{T}
	\]
	Define the event 
	$$\Ec_{+}(\delta) = \{\min_{1 \leq i \leq d} |P_i^{\prime}z_T| > \psi(A) \delta\}$$ 
	where $\psi(A) $ is the lower bound shown in Proposition~\ref{anti_conc1} (which we can use due to the continuity of $L\bar{\eta}_t$) and $v = Pz_T$. Under $\Ec_{+}(\delta)$, $|v_i| > 0$. Now we need a lower bound for $\sigma_{\min}(F_T)$ under $\Ec_{+}(\delta)$
	\begin{align}
	F_T  &= P^{-1} \sum_{i=1}^T \Lambda^{-i+1} Pz_T z_T^{\prime}P^{\prime} \Lambda^{-i+1 \prime} P^{-1 \prime} = P^{-1} \sum_{i=1}^T \Lambda^{-i+1} vv^{\prime} \Lambda^{-i+1 \prime} P^{-1 \prime} \label{ft_eq} \\
	&\succeq \phi_{\min}(A)^2 \psi(A)^2 \delta^2  P^{-1} P^{-1 \prime} {\succeq} \frac{\phi_{\min}(A)^2}{\sigma_{\max}(P)^2} \psi(A)^2 \delta^2  I\label{lb_ft}
	\end{align}
	Further, since $A$ is regular we have that $\phi_{\min}(A) > 0$ from Proposition~\ref{reg_invertible}. Then with probability at least $1-\delta$ we have 
	\[
		\sigma_{\min}(F_T) \geq \frac{\phi_{\min}(A)^2}{\sigma_{\max}(P)^2} \psi(A)^2 \delta^2 > 0
	\]

For the upper bound, observe that $Pz_T$ is a sub-Gaussian random variable. Since 
$$||Pz_T z_T^{\prime}P^{\prime}|| \leq z_T^{\prime} P^{\prime} P z_T$$
and recalling that 
\[
z_T = \underbrace{[A^{-1}, A^{-2}, \ldots, A^{-T}]}_{\tilde{A}} \begin{bmatrix}
\eta_1 \\\eta_2 \\
\vdots \\
\eta_T
\end{bmatrix}
\]
we can use dependent Hanson Wright inequality (Corollary~\ref{dep-hanson-wright}) to bound $z_T^{\prime} P^{\prime} P z_T$. In Theorem~\ref{hanson-wright}, 
\begin{align*}
B &= \tilde{A}^{\prime} P^{\prime} P \tilde{A} \\
\Ex[z_T^{\prime} P^{\prime} P z_T] &= \text{tr}(P(\Gamma_T(A^{-1})-I)P^{\prime}) \sigma_1(L)^2 =   \text{tr}(P(\Gamma_T(A^{-1})-I)P^{\prime})\\
||B||_2, ||B||_F \leq \text{tr}(\tilde{A}^{\prime} P^{\prime} P \tilde{A}) &=  \text{tr}(P(\Gamma_T(A^{-1})-I)P^{\prime})
\end{align*}
Then with probability at least $1 -\delta$ we have 
\[
z_T^{\prime} P^{\prime} P z_T \leq (1 + \frac{1}{c}\log{\frac{1}{\delta}})\text{tr}(P(\Gamma_T(A^{-1})-I)P^{\prime})
\]
and we get from Eq.~\eqref{ft_eq}
\begin{align}
F_T  &\preceq P^{-1} \sum_{i=1}^T \Lambda^{-i+1} Pz_T z_T^{\prime}P^{\prime} \Lambda^{-i+1 \prime} P^{-1 \prime} \nonumber \\
&\preceq (z_T^{\prime}P^{\prime} Pz_T) \sup_{||v||_2=1}\sigma_{\max}\Big(P^{-1} \sum_{i=1}^T \Lambda^{-i+1}  vv^{\prime}  \Lambda^{-i+1 \prime} P^{-1 \prime}\Big)I \nonumber \\
 &\preceq\frac{\phi_{\max}(A)^2}{\sigma_{\min}(P)^2}(1+\frac{1}{c}\log{\frac{1}{\delta}})\text{tr}(P(\Gamma_T(A^{-1})-I)P^{\prime}) I \label{ub_ft}
\end{align}
Then we have with probability at least $1-2\delta$ 
\begin{align}
F_T &\succeq \frac{\phi_{\min}(A)^2}{\sigma_{\max}(P)^2} \psi(A)^2 \delta^2 I \\
F_T &\preceq \frac{\phi_{\max}(A)^2}{\sigma_{\min}(P)^2}(1+\frac{1}{c}\log{\frac{1}{\delta}})\text{tr}(P(\Gamma_T(A^{-1})-I)P^{\prime}) I \label{bnds_fT}
\end{align}
\end{proof}

%% file: content/composite.tex
\section{Composite Result}
\label{composite_result_proof}
In this section we discuss error rates for regular matrices which may have eigenvalues anywhere in the complex plane. The key step is to recall that for every matrix $A$ it is possible to find $\tilde{P}$ such that 
\begin{align}
A = \tilde{P}^{-1} \underbrace{\begin{bmatrix}
A_{e} & 0  & 0 \\
0 & A_{ms} & 0 \\
0 & 0 & A_s 
\end{bmatrix}}_{=\tilde{A}}\tilde{P} \label{partition}
\end{align}
Here $A_{e}, A_{ms}, A_s$ are the purely explosive, marginally stable and stable portions of $A$. This follows because any matrix $A$ has a Jordan normal form $A = P^{-1} \Lambda P$, where $\Lambda$ is a block diagonal matrix and each block corresponds to an eigenvalue. We can always find $Q$ (a rearrangement matrix) such that $\Lambda$ is partitioned into two diagonal parts: explosive, marginally stable and stable, \textit{i.e.},
\begin{align}
A = P^{-1}Q^T \begin{bmatrix}
\Lambda_{e} & 0  & 0 \\
0 & \Lambda_{ms} & 0 \\
0 & 0 & \Lambda_s 
\end{bmatrix}QP
\end{align}
Clearly, $\tilde{P} = QP$. 
Since 
\begin{align}
X_t &= \sum_{\tau=1}^t A^{\tau-1}\eta_{t -\tau+1} \nonumber \\
\tilde{X}_t = \tilde{P}X_t &= \sum_{\tau=1}^t \tilde{A}^{\tau-1}\underbrace{\tilde{P}\eta_{t -\tau+1}}_{\tilde{\eta}_{t-\tau+1}}
\end{align}
Now, the transformed dynamics are as follows:
\begin{align*}
\tilde{X}_{t+1} &= \tilde{A}\tilde{X}_t + \tilde{\eta}_{t+1}
\end{align*}
where $\tilde{A}$ has been partitioned into explosive and stable components as Eq.~\eqref{partition}. Corresponding to $\tilde{A}$ partition $\tX_t, \tn_t$
\begin{align}
\tilde{X}_t = \begin{bmatrix}
\xe_t \\
\xms_t \\
\xs_t
\end{bmatrix}&, \tilde{\eta}_t = \begin{bmatrix}
\nee_t \\
\nms_t \\
\ns_t
\end{bmatrix}
\end{align}
\begin{align}
\tilde{Y}_T = \sum_{t=1}^T \tX_t \tX_t^{\prime} &= \sum_{t=1}^T\begin{bmatrix}
\xe_t (\xe_t)^{\prime} & \xe_t (\xms_t)^{\prime} & \xe_t (\xs_t)^{\prime}\\
\xms_t (\xe_t)^{\prime} & \xms_t (\xms_t)^{\prime} & \xms_t (\xs_t)^{\prime} \\
\xs_t (\xe_t)^{\prime} & \xs_t (\xms_t)^{\prime} & \xs_t (\xs_t)^{\prime}
\end{bmatrix}
\end{align}
We analyze the error of identification in the transformed system instead and show how it relates to the actual error. Note that $\tilde{P}$ is unknown, the transformation is done for ease of analysis. The invertibility of submatrix corresponding to stable and marginally stable components, \textit{i.e.},
\begin{align*}
X^{mss}_t = \begin{bmatrix}
X^{ms}_{t} \\
X^{s}_{t}
\end{bmatrix}
\end{align*}
follows from Theorem~\ref{main_result}. To see this let $A_e$ be a $d_e \times d_e$ matrix. Define 
\[
P_{mss} = \tilde{P}[d_e+1:d, :]
\]
\textit{i.e.}, $P_{mss}$ is the rectangular matrix formed by removing the rows of $\tilde{P}$ corresponding to the explosive part. Then, by definition, we have that 
\[
\begin{bmatrix}
\nms_t \\
\ns_t
\end{bmatrix} = P_{mss} \eta_t
\]
and 
\[
X_{t+1}^{mss}= \underbrace{\begin{bmatrix}
A_{ms} &0 \\
0 & A_s
\end{bmatrix}}_{A_{mss}}X_{t}^{mss} +  \begin{bmatrix}
\nms_{t+1} \\
\ns_{t+1}
\end{bmatrix}
\]
Further 
\[
\Ex[P_{mss}\eta_t \eta_t^{\prime} P_{mss}^{\prime}] = P_{mss}P_{mss}^{\prime} \succ 0
\]

Since all rows of $\tilde{P}$ are independent then $P_{mss}P_{mss}^{\prime}$ is invertible and $\{P_{mss}\eta_t\}_{t=1}^T$ are independent subGaussian vectors. Now this is the same set up as the general version of Theorem~\ref{main_result} discussed in Section~\ref{short_proof}. Since $A_{mss} \in \Sc_0 \cup \Sc_1$ only has stable and marginally stable components, it follows from the Eq.~\eqref{lower_bnd} that 
$$\sum_{t=1}^T X^{mss}_t  (X^{mss}_t)^{\prime} \succeq \frac{T}{4} \sigma_{\min}(P_{mss}P_{mss}^{\prime})I$$ 
with high probability. Then since $\sigma_{\min}(P_{mss}P_{mss}^{\prime}) \geq \sigma_{\min}(\tilde{P})^2 = R^2$, we have that $\sum_{t=1}^T X^{mss}_t  (X^{mss}_t)^{\prime} \succeq \frac{TR^2}{4}I$. Let $\sigma_{\max}(\tilde{P}) =1$. (this makes no difference to the results and $R$ can be interpreted as the inverse condition number)

Recall the definition of $\beta_0(\delta)$
$$\beta_0(\delta) = \inf{\Big\{\beta|\beta^2\sigma_{\min}(\Gamma_{\lfloor \frac{1}{\beta}\rfloor}(A))  \geq   \Big(\frac{ 8ec(A, \delta)}{ TR^2\sigma_{\min}(A A^{\prime})}\Big)\Big\}}$$
we refer to $\beta_0(\delta)$ as $\beta_0$. Following our discussion in Proposition~\ref{gramian_lb} we see that $\beta_0 > 0$ and since $\sigma_{\min}(\Gamma_t(A)) \geq \alpha(d)t$ we have that 
\[
\beta_0 \leq \frac{8ec(A, \delta)}{TR^2 \sigma_{\min}^2(A)C(d)} \implies \frac{1}{\beta_0} \geq \frac{T R^2\sigma_{\min}^2(A)C(d)}{8ec(A, \delta)}
\]
Define 
$$
V_e = (\sum_{t=1}^T \xe_t (\xe_t)^{\prime}), V_{s} = \frac{TR^2}{4}I, V_{ms} = \Big(\frac{TR^2}{8e} \Gamma_{\lfloor \frac{1}{\beta_0}\rfloor}(A_{ms}) \Big)
$$
where the invertibility in $V_e$ holds with high probability. Observe that $V_{ms} \preceq (\sum_{t=1}^T \xms_t (\xms_t)^{\prime}), V_{s} \preceq (\sum_{t=1}^T \xs_t (\xs_t)^{\prime})$ with high probability (follows from Eq.~\eqref{lower_bnd},\eqref{stable_yt}). This observation will be useful in proving the composite invertibility.

Although the technique to prove the invertibility of $\sum_{t=1}^T \tX_t \tX_t^{\prime}$ is similar in spirit to that of~\cite{faradonbeh2017finite}, it addresses additional difficulties arising due to the presence of a marginally stable block.
\begin{align}
B_{d \times d} &= \begin{bmatrix}
V_e^{-1/2} & 0 & 0   \\
0 & V_{ms}^{-1/2} & 0\\
0 & 0 & V_s^{-1/2}
\end{bmatrix}
\end{align}
We will show that $B \sum_{t=1}^T \tX_t \tX_t^{\prime} B^{\prime}$ is positive definite with high probability, \textit{i.e.},
\begin{align}
\sum_{t=1}^T B \tX_t \tX_t^{\prime}B^{\prime} &= \begin{bmatrix}
I  & \sum_{t=1}^T V_{e}^{-1/2}  \xe_t (\xms_t)^{\prime} V_{ms}^{-1/2 \prime} & \sum_{t=1}^T V_{e}^{-1/2}  \xe_t (\xs_t)^{\prime} V_{s}^{-1/2 \prime}\\
\sum_{t=1}^T V_{ms}^{-1/2} \xms_t (\xe_t)^{\prime} V_{e}^{-1/2 \prime}  & \sum_{t=1}^T V_{ms}^{-1/2} \xms_t (\xms_t)^{\prime} V_{ms}^{-1/2 \prime} & \sum_{t=1}^T V_{ms}^{-1/2} \xms_t (\xs_t)^{\prime} V_{s}^{-1/2 \prime} \\
\sum_{t=1}^T V_{s}^{-1/2} \xs_t (\xe_t)^{\prime} V_{e}^{-1/2 \prime}  & \sum_{t=1}^T V_{s}^{-1/2} \xs_t (\xms_t)^{\prime} V_{ms}^{-1/2 \prime} & \sum_{t=1}^T V_{s}^{-1/2} \xs_t (\xs_t)^{\prime} V_{ms}^{-1/2 \prime}
\end{bmatrix}
\end{align}
We already showed that lower submatrix is invertible. To show that the entire matrix is invertible we need to show  
\[
||V_{e}^{-1/2}  \sum_{t=1}^T \xe_t (\xms_t)^{\prime} V_{ms}^{-1/2 \prime}||, ||V_{e}^{-1/2}  \sum_{t=1}^T \xe_t (\xs_t)^{\prime} V_{s}^{-1/2 \prime}|| < \gamma/8
\]
with high probability for some appropriate $\gamma$ and 
\[
\sigma_{\min}\Bigg(\begin{bmatrix}
V_{ms}^{-1/2} & 0\\
 0 & V_s^{-1/2}
\end{bmatrix} \sum_{t=1}^T \xmss_t (\xmss_t)^{\prime} \begin{bmatrix}
V_{ms}^{-1/2} & 0\\
0 & V_s^{-1/2}
\end{bmatrix}\Bigg) \geq \gamma > 0
\]
\subsection{Cross Terms have low norm}
\label{cross_low}
Define the following quantities:
\begin{align}
\alpha(A_e, \delta) &= \frac{3\phi_{\max}(A_e)^2 \sigma_{\max}^2(A_e)}{\phi_{\min}(A_e)^2 \sigma_{\min}(A_e)^2}\frac{\Big(1 + \frac{1}{c} \log{\frac{1}{\delta}}\Big)\text{tr}(P_e(\Gamma_T(A_e^{-1} - I))P_e^{\prime})}{\psi(A_e)^2 \delta^2} \label{alpha_exp} \\ 
T_{mc}(\delta) &= {\Bigg\{T \bl \alpha(A_e, \delta)\text{tr}(A_{e}^{-T + k_{mc}(T)}(A_{e}^{-T + k_{mc}(T)})^{\prime}) \leq \frac{\gamma^2}{256} \Bigg\}} \label{ms_exp} \\
k_{mc} &= k_{mc}(T) = T \Bigg(1 - \frac{R^2\gamma^2}{2048 de \lambda_1\Big(\Gamma_T(A_{ms})\Gamma^{-1}_{\lfloor \frac{1}{\beta_0(\delta)} \rfloor}(A_{ms}){\Big(1 + \frac{1}{c}\log{\frac{1}{\delta}}\Big)}\Big)}\Bigg) \label{k_mc} \\
T_{sc}(\delta) &= {\Bigg\{T \bl \alpha(A_e, \delta)\text{tr}(A_{e}^{-T + k_{sc}(T)}(A_{e}^{-T + k_{sc}(T)})^{\prime}) \leq \frac{\gamma^2}{256} \Bigg\}} \label{s_exp} \\ 
k_{sc} &= k_{sc}(T) = T \Bigg(1 - \frac{R^2\gamma^2}{1024 d\lambda_1\Big(\Gamma_T(A_{s}){\Big(1 + \frac{1}{c}\log{\frac{1}{\delta}}\Big)}\Big)}\Bigg) \label{k_s}
\end{align}
\begin{remark}
Note that $T_{mc}(\delta)$ (and $T_{sc}(\delta)$) is a set where there exists a minimum $T_{*} < \infty$ such that $T \in T_{mc}(\delta)$ whenever $T \geq T_{*} $. However, there might be $T < T_{*}$ for which the inequality of $T_{mc}(\delta)$ holds. Whenever we write $T \in T_{mc}(\delta)$ we mean $T \geq T_{*}$. 
\end{remark}
Second note that for every $T$, since $R, \gamma < 1$ we have 
\[
k_{sc}(T), k_{mc}(T) \geq \frac{T}{2}
\]
These quantities will be useful in stating the error bounds. We have
\begin{align*}
|| V_{e}^{-1/2}  \sum_{t=1}^T \xe_t (\xms_t)^{\prime} V_{ms}^{-1/2 \prime}|| &\leq ||V_{e}^{-1/2}  \sum_{t=1}^k \xe_t (\xms_t)^{\prime} V_{ms}^{-1/2 \prime}|| + ||V_{e}^{-1/2}  \sum_{t= k+1}^T \xe_t (\xms_t)^{\prime} V_{ms}^{-1/2\prime}||
\end{align*} 
We will need a more nuanced argument to upper bound Eq.~\eqref{err0} than that provided in~\cite{faradonbeh2017finite} (although it will be similar in flavor).
\begin{align}
\Pb(||V_{e}^{-1/2}  \sum_{t= 1}^T \xe_t (\xms_t)^{\prime} V_{ms}^{-1/2} || ) \label{err0}
\end{align}
For any $v_1, v_2$ we break $|v_1^{\prime}V_{e}^{-1/2}  \sum_{t= 1}^T \xe_t (\xms_t)^{\prime} V_{ms}^{-1/2}v_2|$ into two parts $$|v_1^{\prime}V_{e}^{-1/2}  \sum_{t= 1}^k \xe_t (\xms_t)^{\prime} V_{ms}^{-1/2}v_2|$$ 
and 
$$|v_1^{\prime}V_{e}^{-1/2}  \sum_{t= k+1}^T \xe_t (\xms_t)^{\prime} V_{ms}^{-1/2}v_2|$$.
For $|v_1^{\prime}V_{e}^{-1/2}  \sum_{t= k+1}^T \xe_t (\xms_t)^{\prime} V_{ms}^{-1/2}v_2|$ we have

\begin{align}
|v_1^{\prime} V_{e}^{-1/2}  \sum_{t=k+1}^T \xe_t (\xms_t)^{\prime} V_{ms}^{-1/2}  v_2| &\leq \underbrace{\sqrt{v_1^{\prime} V_{e}^{-1/2}  \sum_{t=k+1}^T \xe_t (\xe_t)^{\prime} V_{e}^{-1/2}  v_1}}_{\leq 1} \sqrt{v_2^{\prime} V_{ms}^{-1/2} \sum_{t=k+1}^T \xms_t (\xms_t)^{\prime}V_{ms}^{-1/2} v_2} \nonumber\\
&\leq \sqrt{v_2^{\prime} V_{ms}^{-1/2}\sum_{t=k+1}^T \xms_t (\xms_t)^{\prime} V_{ms}^{-1/2} v_2} \leq \sqrt{\sigma_1( V_{ms}^{-1/2}\sum_{t=k+1}^T \xms_t (\xms_t)^{\prime} V_{ms}^{-1/2} )} \nonumber\\
&\leq \sqrt{\lambda_1( \sum_{t=k+1}^T \xms_t (\xms_t)^{\prime} V_{ms}^{-1} )}\label{err_2}
\end{align}
To upper bound Eq.~\eqref{err_2} we simply need to upper bound $V_{ms}^{-1/2}\sum_{t=k+1}^T \xms_t (\xms_t)^{\prime}V_{ms}^{1/2}$. We can use dependent Hanson--Wright inequality (Corollary~\ref{dep-hanson-wright}) and Corollary~\ref{sub_sum}. 
Then from Corollary~\ref{sub_sum} and since $V_{ms}$ is deterministic we can conclude that with probability at least $1-\delta$ we get 
\begin{equation}
V_{ms}^{-1/2}  \sum_{t=k+1}^T \xms_t (\xms_t)^{\prime} V_{ms}^{-1/2}   \preceq \sum_{t=k+1}^T\text{tr}( V_{ms}^{-1/2} \Gamma_t(A_{ms}) V_{ms}^{-1/2} )\Big(1 + \frac{1}{c}\log{\frac{1}{\delta}}\Big)I \label{errf_1}
\end{equation}
We can upper bound the deterministic quantity in Eq.~\eqref{errf_1} as
\begin{align}
\sum_{t=k+1}^T\text{tr}( V_{ms}^{-1/2}  \Gamma_t(A)  V_{ms}^{-1/2}) &\leq d\lambda_1(\sum_{t=k+1}^T\Gamma_t(A_{ms}) V_{ms}^{-1}) \nonumber\\
&= d\lambda_1\Big(\frac{8e}{TR^2}\sum_{t=k+1}^T\Gamma_t(A_{ms}) \Gamma_{\lfloor \frac{1}{\beta_0(\delta)} \rfloor}(A_{ms})^{-1}\Big) \nonumber \\
&\leq d\lambda_1\Big(\frac{8e(T-k)}{TR^2}\Gamma_T(A_{ms}) \Gamma_{\lfloor \frac{1}{\beta_0(\delta)} \rfloor}(A_{ms})^{-1}\Big)\label{errf_0}
\end{align} 
The last inequality holds because the eigenvalues of $P^{-1/2} Q P^{-1/2}$ are the same as $QP^{-1}$ and non--negative whenever $P, Q$ are psd matrices. The normalized gramian term, $\Gamma_t(A_{ms}) \Gamma_{\lfloor \frac{1}{\beta_0(\delta)} \rfloor}(A_{ms})^{-1}$, appears in Eq.~\eqref{errf_0} only because $V_{ms}$ is deterministic. This will help us in getting non--trivial upper bounds for the cross terms of explosive and marginally stable pair. The key is the choice of $k$. In Proposition~\ref{gramian_ratio} we showed that $\lambda_1(\Gamma_{t_1} \Gamma_{t_2}^{-1})$ only depends on the ratio of $t_1/t_2$ and $A_{ms}$ and not on the specific values of $t_1, t_2$. Note that due to Proposition~\ref{gramian_ratio} the normalized gramian term $\Gamma_T(A_{ms})\Gamma^{-1}_{\lfloor \frac{1}{\beta_0(\delta)} \rfloor}(A_{ms})$ has spectral radius that is at most polynomial in $T\beta_0(\delta)$. Since $\beta_0(\delta) \approx \frac{\log{T}}{T} \times \log{\frac{1}{\delta}}$, we get that $$\lambda_1(\Gamma_T(A_{ms})\Gamma^{-1}_{\lfloor \frac{1}{\beta_0(\delta)} \rfloor}(A_{ms})) = \text{poly}\Big(\log{T}, \log{\frac{1}{\delta}}\Big)$$
Our choices of $T_{mc}(\delta), k_{mc}(T)$ in Eq.~\eqref{ms_exp},\eqref{k_mc} are motivated by the preceding discussion. We set $k = k_{mc}(T)$ and we have that $d\lambda_1\Big(\frac{8e(T-k)}{TR^2}\Gamma_T(A_{ms}) \Gamma_{\lfloor \frac{1}{\beta_0(\delta)} \rfloor}(A_{ms})^{-1}\Big) \leq \frac{\gamma^2}{256}$  (check by directly substituting $k=k_{mc}(T)$ in Eq.~\eqref{errf_0}) and as a result from Eq.~\eqref{err_2} 
\[
|v_1^{\prime} V_{e}^{-1/2}  \sum_{t=k+1}^T \xe_t (\xms_t)^{\prime} V_{ms}^{-1/2}  v_2|  \leq \frac{\gamma}{16}
\]
for arbitrary $v_1, v_2$. Similarly for the second part
\begin{align}
|v_1^{\prime} V_{e}^{-1/2}  \sum_{t=1}^k \xe_t (\xms_t)^{\prime} V_{ms}^{-1/2}  v_2| &\leq \underbrace{\sqrt{v_1^{\prime} V_{e}^{-1/2}  \sum_{t=1}^k \xe_t (\xe_t)^{\prime} V_{e}^{-1/2}  v_1}}_{a_1} \underbrace{\sqrt{v_2^{\prime}V_{ms}^{-1/2}  \sum_{t=1}^k \xms_t (\xms_t)^{\prime} V_{ms}^{-1/2}  v_2}}_{\leq 1} \label{err_3}
\end{align}

For the choice of $k=k_{mc}$ the other term can be simplified as
\begin{align}
a_1 &= \sqrt{v_1^{\prime} V_{e}^{-1/2}  \sum_{t=1}^k \xe_t (\xe_t)^{\prime} V_{e}^{-1/2}  v_1} \leq \sqrt{\sigma_1(V_{e}^{-1/2}  \sum_{t=1}^k \xe_t (\xe_t)^{\prime} V_{e}^{-1/2} )} \leq \sqrt{\lambda_1(\sum_{t=1}^k \xe_t (\xe_t)^{\prime} V_{e}^{-1})} \nonumber\\
&\leq \sqrt{\text{tr}(\sum_{t=1}^k \xe_t (\xe_t)^{\prime} V_{e}^{-1})} \label{err_1}
\end{align}
 By ensuring that both $T, k = k_{mc} (\text{which is }\geq T/2) \in T_u(\delta)$ (from Table~\ref{notation}) we have from Eq.~\eqref{exp_bnds} that
 \begin{align*}
\sum_{l=1}^k  \xe_t (\xe_t)^{\prime} &\preceq\frac{3\phi_{\max}(A_e)^2}{2\sigma_{\min}(P_e)^2}(1+\frac{1}{c}\log{\frac{1}{\delta}})\text{tr}(P_e(\Gamma_T(A_e^{-1})-I)P_e^{\prime}) A_e^k A_e^{k \prime}  \nonumber \\
 V_e &\succeq \frac{\phi_{\min}(A_e)^2 \psi(A_e)^2 \delta^2}{2\sigma_{\max}(P_e)^2}A_e^T A_e^{T \prime}\nonumber\\  
 \end{align*}
 Define 
\[
\alpha(A_e, \delta) = \frac{3\phi_{\max}(A_e)^2 \sigma_{\max}^2(A_e)}{\phi_{\min}(A_e)^2 \sigma_{\min}(A_e)^2}\frac{\Big(1 + \frac{1}{c} \log{\frac{1}{\delta}}\Big)\text{tr}(P_e(\Gamma_T(A_e^{-1}) - I)P_e^{\prime})}{\psi(A_e)^2 \delta^2}
\]
and we can conclude 
\[
 \sqrt{\text{tr}(\sum_{t=1}^k \xe_t (\xe_t)^{\prime} V_{e}^{-1})} \leq  \sqrt{\alpha(A_e, \delta)\text{tr}(A_e^{-T + k}(A_e^{-T + k})^{\prime})}
\]
with probability at least $1-2\delta$. Since $T \in T_{mc}(\delta)$ we have 
\begin{equation}
a_1 \leq \sqrt{\alpha(A_e, \delta)\text{tr}(A_e^{-T + k}(A_e^{-T + k})^{\prime})} \leq  \frac{\gamma}{16} \label{errf_22}
\end{equation}
with probability at least $1-2\delta$. Then combining Eq.~\eqref{err_2},\eqref{errf_1},\eqref{err_3},\eqref{errf_22} we get with probability at least $1-4\delta$ that
\begin{align} 
|v_1^{\prime}V_{e}^{-1/2}  \sum_{t= 1}^T \xe_t (\xms_t)^{\prime} V_{ms}^{-1/2} v_2| &\leq \frac{\gamma}{8}
\end{align}
This implies with probability at $1-4 \delta$ we have 
\begin{align}
||V_{e}^{-1/2}  \sum_{t= 1}^T \xe_t (\xms_t)^{\prime} V_{ms}^{-1/2} || &\leq \frac{\gamma}{8}\label{fin_err}
\end{align}
We have a similar assertion for the stable--explosive block but with $T \in T_{sc}(\delta)$ and $k= k_{sc}(T)$. 
\begin{align}
||V_{e}^{-1/2}  \sum_{t= 1}^T \xe_t (\xs_t)^{\prime} V_{s}^{-1/2} || &\leq \frac{\gamma}{8}\label{fin_err2}
\end{align}
It should be noted that  $T \in T_{sc}(\delta), T_{mc}(\delta)$ are both poly logarithmic in $\delta$ because of $A^{-T + k_{mc}}$ (or $A^{-T + k_{sc}}$) term which is exponentially decaying.
\begin{remark}
Whenever $T \in T_{sc}(\delta) , T_{mc}(\delta)$, the other conditions on $T$ such as $T/2 \in T_{u}(\delta)$ or $T \geq T_{s}(\delta) \vee T_{ms}(\frac{\delta}{2T})$ for the invertibility of the individual stable, marginally stable blocks are satisfied simultaneously (or are trivial to satisfy) and we do not state them explicitly. 
\end{remark}
\subsection{Norm of scaled $\sum_{t=1}^T \xmss_t(\xmss_t)^{\prime}$ is high}
\label{highnorm_gen}
Now we need to check 
\[
\sigma_{\min}\Bigg(\begin{bmatrix}
V_{ms}^{-1/2} & 0\\
0 & V_s^{-1/2}
\end{bmatrix} \sum_{t=1}^T \xmss_t (\xmss_t)^{\prime} \begin{bmatrix}
V_{ms}^{-1/2} & 0\\
0 & V_s^{-1/2}
\end{bmatrix}\Bigg) \geq \gamma > 0
\]
Since from Theorem~\ref{main_result} and its extension in Section~\ref{short_proof} it is known that with probability at least $1-\delta$ we have $\sum_{t=1}^T \xmss_t (\xmss_t)^{\prime} \succeq R^{2} \frac{TI}{4}$ for some fixed $R =\sigma_{\min}(\tilde{P}) > 0$, then we know that the Schur complement of $\sum_{t=1}^T \xmss_t (\xmss_t)^{\prime}$ is invertible too. For shorthand let 
\[
M = \sum_{t=1}^T \xmss_t (\xmss_t)^{\prime} = \begin{bmatrix} 
M_{11} & Q^{\prime} \\
Q & M_{22}
\end{bmatrix}
\]
Then the Schur complement is 
\[
M/M_{11} = M_{22} - Q M_{11}^{-1} Q^{\prime}
\]
Since $\sigma_{\min}(M) \geq R^2 \frac{TI}{4}$ then from Corollary 2.3 in~\cite{Liu2005} we have that 
\[
\sigma_{\min}(M/M_{11}) \geq R^2 \frac{T}{4}
\]
Since $M_{22} \preceq \sum_{t=0}^{T-1}\text{tr}(\Gamma_t(A_s))\Big(1 + \frac{1}{c}\log{\frac{1}{\delta}}\Big)I$ with probability at least $1-\delta$. We see that with probability at least $1-\delta$
\begin{equation}
M_{22}^{-1/2}(M/M_{11})M_{22}^{-1/2} = I - M_{22}^{-1/2} Q M_{11}^{-1/2} M_{11}^{-1/2} Q^{\prime} M_{22}^{-1/2}\succeq \frac{R^2}{4 \text{tr}(\Gamma_T(A_s))(1 + \frac{1}{c}\log{\frac{1}{\delta}})}I \label{lb_schur}
\end{equation}
Since $A_s$ is stable $\text{tr}(\Gamma_T(A_s)) \leq \text{tr}(\Gamma_{\infty}(A_s)) < \infty$. Define 
\begin{equation}
\omega(\delta) = \frac{R^2}{4 \text{tr}(\Gamma_T(A_s))(1 + \frac{1}{c}\log{\frac{1}{\delta}})} > 0 \label{omega}
\end{equation}
Then this implies that
\[
\begin{bmatrix} 
M_{11}^{-1/2} & 0 \\
0 & M_{22}^{-1/2}
\end{bmatrix} M \begin{bmatrix} 
M_{11}^{-1/2} & 0 \\
0 & M_{22}^{-1/2}
\end{bmatrix}= \begin{bmatrix} 
I &M_{11}^{-1/2}Q^{\prime}M_{22}^{-1/2} \\
M_{22}^{-1/2}QM_{11}^{-1/2} & I
\end{bmatrix} \succeq \frac{\omega(\delta)}{4} I
\]
because for any $v = \begin{bmatrix}
v_1 \\
v_2
\end{bmatrix}$ we have 
\begin{align*}
v^{\prime}\begin{bmatrix} 
I &\underbrace{M_{11}^{-1/2}QM_{22}^{-1/2}}_{=D^{\prime}} \\
M_{22}^{-1/2}Q^{\prime}M_{11}^{-1/2} & I
\end{bmatrix}v &= v_1^{\prime} v_1 + v_1^{\prime} D v_2 + v_2^{\prime} D^{\prime} v_1 + v_2^{\prime}v_2 \\
&= v_1^{\prime}v_{1}-2 \sqrt{1 - \omega(\delta)} ||v_2|| ||v_1|| + v_2^{\prime} v_2 \\
&\geq v_1^{\prime}v_{1}-2 \Big(1 - \frac{\omega(\delta)}{2}\Big) ||v_2|| ||v_1|| + v_2^{\prime} v_2 
\end{align*}
Since from Eq.~\eqref{lb_schur} it follows that $||D||^2 \leq 1 - \omega(\delta)$ we obtain 
\begin{align*}
v_1^{\prime}v_{1}-2 \sqrt{1 - \omega(\delta)} ||v_2|| ||v_1|| + v_2^{\prime} v_2 &=v_1^{\prime}v_{1}-2 \Big({1 - \frac{\omega(\delta)}{2}}\Big) ||v_2|| ||v_1|| + v_2^{\prime} v_2 \\
&= \Big({1 - \frac{\omega(\delta)}{2}}\Big)(||v_1|| -||v_2||)^2 + \Big(1- \sqrt{{1 - \frac{\omega(\delta)}{2}}}\Big)(||v_1||^2 + ||v_2||^2)\\
&\geq \Big(\frac{\omega(\delta)}{4}\Big)(||v_1||^2 + ||v_2||^2)
\end{align*}
Combining these observations we get 
\begin{align*}
v^{\prime}\begin{bmatrix} 
	I &\underbrace{M_{11}^{-1/2}QM_{22}^{-1/2}}_{=D} \\
	M_{22}^{-1/2}Q^{\prime}M_{11}^{-1/2} & I
\end{bmatrix}v &\geq  \Big(\frac{\omega(\delta)}{4}\Big)
\end{align*}

We have that 
\[
\sigma_{\min}\Big(\begin{bmatrix} 
M_{11}^{-1/2} & 0 \\
0 & M_{22}^{-1/2}
\end{bmatrix} M \begin{bmatrix} 
M_{11}^{-1/2} & 0 \\
0 & M_{22}^{-1/2}
\end{bmatrix} \Big) \geq \Big(\frac{\omega(\delta)}{4}\Big)
\]

Since $M_{22} \succeq V_s, M_{11} \succeq V_{ms}$ we have with probability at least $1-\delta$
\begin{equation}
\sigma_{\min}\Bigg(\begin{bmatrix}
V_{ms}^{-1/2} & 0\\
0 & V_s^{-1/2}
\end{bmatrix} \sum_{t=1}^T \xmss_t (\xmss_t)^{\prime} \begin{bmatrix}
V_{ms}^{-1/2} & 0\\
0 & V_s^{-1/2}
\end{bmatrix}\Bigg) \geq \Big(\frac{\omega(\delta)}{4}\Big) > 0 \label{fin_term3}
\end{equation}

Now we replace in Eq.~\eqref{fin_err},\eqref{fin_err2} $\gamma \rightarrow \frac{\sqrt{{\omega(\delta)}}}{32}$. Then that implies 
\begin{align*}
||\ V_{e}^{-1/2}  \sum_{t= 1}^T \xe_t (\xs_t)^{\prime}  V_{s}^{-1/2}|| &\geq \frac{\sqrt{{\omega(\delta)}}}{64} \\
||\ V_{e}^{-1/2}  \sum_{t= 1}^T \xe_t (\xms_t)^{\prime}  V_{ms}^{-1/2}|| &\geq \frac{\sqrt{{\omega(\delta)}}}{64} 
\end{align*}
\subsection{Lower Bound on $\sum_{t=1}^T \tX_t \tX_t^{\prime}$}
\label{lower_gen}
Recalling that 
\begin{align*}
\sum_{t=1}^T B \tX_t \tX_t^{\prime}B^{\prime} &= \begin{bmatrix}
I  & \sum_{t=1}^T V_{e}^{-1/2}  \xe_t (\xms_t)^{\prime} V_{ms}^{-1/2 \prime} & \sum_{t=1}^T V_{e}^{-1/2}  \xe_t (\xs_t)^{\prime} V_{s}^{-1/2 \prime}\\
\sum_{t=1}^T V_{ms}^{-1/2} \xms_t (\xe_t)^{\prime} V_{e}^{-1/2 \prime}  & \sum_{t=1}^T V_{ms}^{-1/2} \xms_t (\xms_t)^{\prime} V_{ms}^{-1/2 \prime} & \sum_{t=1}^T V_{ms}^{-1/2} \xms_t (\xs_t)^{\prime} V_{s}^{-1/2 \prime} \\
\sum_{t=1}^T V_{s}^{-1/2} \xs_t (\xe_t)^{\prime} V_{e}^{-1/2 \prime}  & \sum_{t=1}^T V_{s}^{-1/2} \xs_t (\xms_t)^{\prime} V_{ms}^{-1/2 \prime} & \sum_{t=1}^T V_{s}^{-1/2} \xs_t (\xs_t)^{\prime} V_{ms}^{-1/2 \prime}
\end{bmatrix}
\end{align*}
then it follows from Eq.~\eqref{fin_term3} that
\begin{align*}
\sum_{t=1}^T B \tX_t \tX_t^{\prime}B^{\prime} &\succeq \begin{bmatrix}
I  & \sum_{t=1}^T V_{e}^{-1/2}  \xe_t (\xms_t)^{\prime} V_{ms}^{-1/2 \prime} & \sum_{t=1}^T V_{e}^{-1/2}  \xe_t (\xs_t)^{\prime} V_{s}^{-1/2 \prime}\\
\sum_{t=1}^T V_{ms}^{-1/2} \xms_t (\xe_t)^{\prime} V_{e}^{-1/2 \prime}  & \frac{\omega(\delta)}{4} I & 0 \\
\sum_{t=1}^T V_{s}^{-1/2} \xs_t (\xe_t)^{\prime} V_{e}^{-1/2 \prime}  & 0 & \frac{\omega(\delta)}{4} I
\end{bmatrix}
\end{align*}
Let $v = \begin{bmatrix}
v_1 \\
v_2 \\
v_3
\end{bmatrix}$
Then $v^{\prime} \sum_{t=1}^T B \tX_t \tX_t^{\prime}B^{\prime} v = ||v_1||^2 + \frac{\omega(\delta)}{4} (||v_2||_2^2 + ||v_3||_2^2) + 2 v_1^{\prime}\sum_{t=1}^T V_{e}^{-1/2}  \xe_t (\xms_t)^{\prime} V_{ms}^{-1/2 \prime} v_2 + 2 v_1^{\prime} \sum_{t=1}^T V_{e}^{-1/2}  \xe_t (\xs_t)^{\prime} V_{s}^{-1/2 \prime}v_3 \geq ||v_1||^2 + \frac{\omega(\delta)}{4} (||v_2||_2^2 + ||v_3||_2^2) - \frac{\sqrt{\omega(\delta)}}{32} ||v_1|| ||v_2|| - \frac{\sqrt{\omega(\delta)}}{32} ||v_1|| ||v_3|| $. Then we get 
\[
v^{\prime} \sum_{t=1}^T B \tX_t \tX_t^{\prime}B^{\prime} v  \geq ||v_1||^2 + \frac{\omega(\delta)}{4} (||v_2||_2^2 + ||v_3||_2^2) - \frac{\omega(\delta)}{64} (||v_1||^2 + ||v_2||^2) - \frac{\omega(\delta)}{64} (||v_1||^2 + ||v_3||^2)
\]
Thus $\sigma_{\min}(\sum_{t=1}^T B \tX_t \tX_t^{\prime}B^{\prime}) \geq  \frac{\omega(\delta)}{8}$. Summarizing we have with probability at least $1 - C\delta$. The $C \delta$ comes because we are considering the intersection of invertibility of $\sum_{t=1}^T X^{mss}_t (X^{mss}_t)^{\prime}$ and $\sum_{t=1}^T \xe_t (\xe_t)^{\prime}, \sum_{t=1}^T \xs_t (\xs_t)^{\prime}, \sum_{t=1}^T \xms_t (\xms_t)^{\prime}$. 
\[
\sigma_{\min}(\sum_{t=1}^T B \tX_t \tX_t^{\prime}B^{\prime}) \geq \frac{\omega(\delta)}{8}
\]
whenever 
\begin{align}
T \in T_{mc}(\delta) \cap T_{sc}(\delta) \label{t_mc} 
\end{align}
Replacing $\delta \rightarrow \frac{\delta}{C}$ we get with probability at least $1-\delta$ that
\[
\sigma_{\min}(\sum_{t=1}^T B \tX_t \tX_t^{\prime}B^{\prime}) \geq \frac{\omega( \frac{\delta}{C})}{8}
\]
Define 
\begin{align*}
V^e_{dn}(\delta) &= \frac{\phi_{\min}(A_e)^2 \psi(A_e)^2 \delta^2}{2\sigma_{\max}(P)^2}A_e^T A_e^{T \prime}, V^s_{dn}(\delta) = \frac{TR^2}{4}I, V^{ms}_{dn}(\delta) = \Big(\frac{TR^2}{8e} \Gamma_{\lfloor \frac{1}{\beta_0(\delta)}\rfloor}(A_{ms})\Big)
\end{align*}
This implies that with probability at least $1-2\delta$ we have that 
\begin{align}
\sum_{t=1}^T B \tX_t \tX_t^{\prime}B^{\prime} &\succeq \frac{\omega(  \frac{\delta}{C})}{8} I \implies \sum_{t=1}^T \tX_t \tX_t^{\prime}\succeq \frac{\omega(  \frac{\delta}{C})}{8}B^{-2} \nonumber \\ 
\sum_{t=1}^T \tX_t \tX_t^{\prime} &\succeq \underbrace{\frac{\omega(  \frac{\delta}{ C})}{8} \begin{bmatrix}
 V^e_{dn}(\delta) & 0 & 0   \\
 0 & V^{ms}_{dn}( \frac{\delta}{ C}) & 0\\
 0 & 0 & V^{s}_{dn}(  \frac{\delta}{ C})
 \end{bmatrix}}_{=V_{dn}}
\end{align}
 $V^e_{dn}$ depends differently than the rest because $V_e$ was chosen to be data dependent and we only apply the lower bound on $\sum_{t=1}^T \xe_t (\xe_t)^{\prime}$ at the very end. 

\subsection{Finding the Upper Bound $\sum_{t=1}^T \tX_t \tX_t^{\prime}$}
\label{ub_general}
For the upper bound on $\sum_{t=1}^T \tX_t \tX_t^{\prime}$. We use Lemma A.5 of~\cite{simchowitz2018learning}. Consider an arbitrary matrix $M = \begin{bmatrix} M_1 \\
M_2 \\
M_3 \end{bmatrix}$. Then $ \begin{bmatrix} 3M_1M_1^{\prime} & 0 & 0 \\
0 & 3M_2M_2^{\prime}& 0\\
0 & 0 & 3M_3M_3^{\prime} \end{bmatrix} \succeq MM^{\prime}$. This is because
\begin{align*}
\begin{bmatrix} 2M_1M_1^{\prime} & -M_1 M_2^{\prime} & -M_1 M_3^{\prime} \\
-M_2 M_1^{\prime} & 2M_2M_2^{\prime}& -M_2 M_3^{\prime}\\
-M_3 M_1^{\prime} & -M_3 M_2^{\prime} & 2M_3M_3^{\prime} \end{bmatrix} &= (\begin{bmatrix} M_1 \\
0 \\
0 \end{bmatrix}- \begin{bmatrix} 0 \\
M_2 \\
0 \end{bmatrix})(\begin{bmatrix} M_1 \\
0 \\
0 \end{bmatrix}- \begin{bmatrix} 0 \\
M_2 \\
0 \end{bmatrix})^{\prime} \\
&+ (\begin{bmatrix} M_1 \\
0 \\
0 \end{bmatrix}- \begin{bmatrix} 0 \\
0 \\
M_3 \end{bmatrix})(\begin{bmatrix} M_1 \\
0 \\
0 \end{bmatrix}- \begin{bmatrix} 0 \\
0 \\
M_3 \end{bmatrix})^{\prime} + (\begin{bmatrix} 0 \\
0 \\
M_3 \end{bmatrix}- \begin{bmatrix} 0 \\
M_2 \\
0 \end{bmatrix})(\begin{bmatrix} 0 \\
0 \\
M_3 \end{bmatrix}- \begin{bmatrix} 0 \\
M_2 \\
0 \end{bmatrix})^{\prime}
\end{align*}
Define  
\begin{align*}
V^e_{up}(\delta) &= \frac{3\phi_{\max}(A)^2\sigma_{\max}(\tilde{P})^4}{\sigma_{\min}(\tilde{P})^2}(1+\frac{1}{c}\log{\frac{1}{\delta}})\text{tr}(\Gamma_T(A_e^{-1})) A_e^T A_e^{T \prime} \\
V^s_{up}(\delta) &= 3\sigma_{\max}(\tilde{P})^2T\text{tr}(\Gamma_T(A_{s}))\Big(1 + \frac{1}{c}\log{\Big(\frac{1}{\delta}\Big)}\Big)I \\
V^{ms}_{up}(\delta) &= 3\sigma_{\max}(\tilde{P})^2T\text{tr}(\Gamma_T(A_{ms}))\Big(1 + \frac{1}{c}\log{\Big(\frac{1}{\delta}\Big)}\Big) I
\end{align*}
Then with probability at least $1-4\delta$ we have 
\begin{align*}
\begin{bmatrix}
\sum_{t=1}^T \xe(\xe_t)^{\prime} & 0 & 0   \\
0 & \sum_{t=1}^T \xms(\xms_t)^{\prime} & 0\\
0 & 0 & \sum_{t=1}^T \xs(\xs_t)^{\prime}
\end{bmatrix} \preceq \begin{bmatrix}
V^e_{up}(\delta) & 0 & 0   \\
0 & V^{ms}_{up}(\delta) & 0\\
0 & 0 & V^s_{up}(\delta)
\end{bmatrix}
\end{align*}
We get these upper bounds for stable and marginally stable matrices from Proposition~\eqref{energy_markov} and Eq.~\eqref{exp_bnds} for explosive matrices. Then with probability at least $1-4\delta$ we have 
\begin{align}
\sum_{t=1}^T \tX_t \tX_t^{\prime} \preceq \underbrace{\begin{bmatrix}
3V^e_{up}(\delta) & 0 & 0   \\
0 & 3V^{ms}_{up}(\delta) & 0\\
0 & 0 & 3V^s_{up}(\delta)
\end{bmatrix}}_{{=V_{up}}} \label{ub_comp}
\end{align}
Note that the time requirement in Eq.~\eqref{t_mc} is sufficient to ensure the upper bounds with high probability and we do not state them explicitly. 

\subsection{Getting Error Bounds}
\label{error}
We recall the discussion for Theorem~\ref{main_result}. We have $V_{up}, V_{dn}$, so we compute $V_{up}V_{dn}^{-1}$ which gives us 
\begin{align*}
V_{up}V_{dn}^{-1} &= \frac{8}{\omega(\frac{\delta}{C})}\begin{bmatrix}
3V^e_{up}(\delta)(V^e_{dn}(\delta))^{-1} & 0 & 0   \\
0 & 3V^{ms}_{up}(\delta)(V^{ms}_{dn}( \frac{\delta}{ C}))^{-1} & 0\\
0 & 0 & 3V^s_{up}(\delta)(V^{s}_{dn})^{-1}( \frac{\delta}{ C})
\end{bmatrix}\\
\text{det}(V_{up}V_{dn}^{-1}) &= \Big(\frac{24}{\omega(\frac{\delta}{C})}\Big)^d \text{det}(V^e_{up}(\delta)(V^e_{dn}(\delta))^{-1}) \text{det}(V^{ms}_{up}(\delta)(V^{ms}_{dn}( \frac{\delta}{ C}))^{-1})\text{det}(V^s_{up}(\delta)(V^{s}_{dn}( \frac{\delta}{ C}))^{-1})
\end{align*}
Further $V^{s}_{dn}( \frac{\delta}{ C}) = V^{s}_{dn}(\delta)$ (only the time required to be greater than this with high probability changes). Then 
\begin{align*}
\log{(\text{det}(V_{up}V_{dn}^{-1}))} &=d(\log{24} - \log{\omega(\frac{\delta}{C})}) + \log{\text{det}(V^e_{up}(\delta)(V^e_{dn}(\delta))^{-1})} \\
&+ \log{\text{det}(V^{ms}_{up}(\delta)(V^{ms}_{dn}( \frac{\delta}{ C}))^{-1})}  + \log{\text{det}(V^s_{up}(\delta)(V^{s}_{dn}( \frac{\delta}{ C}))^{-1})} 
\end{align*} 
Following this the bounds are straightforward and can be computed as shown in Eq.~\eqref{error_form}. It should be noted that Proposition~\ref{selfnorm_bnd} works for a general case of noise process which $\tilde{\eta}_t$ satisfies.

Now we only know the error of the transformed dynamics, \textit{i.e.},
\begin{align*}
 \sum_{t=1}^T(\sum_{t=1}^T \tX_t \tX_t)^{+}(\sum_{t=1}^T \tX_t \tilde{\eta}_{t+1})
\end{align*}
Since $(\sum_{t=1}^T \tX_t \tX_t)$ is invertible with high probability
\begin{align*}
 \sum_{t=1}^T(\sum_{t=1}^T \tX_t \tX_t)^{+}(\sum_{t=1}^T \tX_t \tilde{\eta}_{t+1})&= (\sum_{t=1}^T \tX_t \tX_t)^{-1}(\sum_{t=1}^T \tX_t \tilde{\eta}_{t+1}) \\
 &=  \sum_{t=1}^T\tilde{P}^{-1 \prime}(\sum_{t=1}^T X_t X_t)^{-1}\tilde{P}^{-1} \tilde{P}X_t \eta_{t+1}\tilde{P}^{\prime} \\
 &= \tilde{P}^{-1 \prime}\sum_{t=1}^T(\sum_{t=1}^T X_t X_t)^{-1}X_t \eta_{t+1}\tilde{P}^{\prime} 
\end{align*}
Then it is clear that 
\[
\bl\bl\sum_{t=1}^T(\sum_{t=1}^T \tX_t \tX_t)^{-1}(\sum_{t=1}^T \tX_t \tilde{\eta}_{t+1})\bl\bl \geq \sigma_{\min}(\tilde{P}^{-1}) \bl \bl \sum_{t=1}^T(\sum_{t=1}^T X_t X_t)^{-1}X_t \eta_{t+1}\bl \bl \sigma_{\min}(\tilde{P})
\]
and we have bounded the original error term in terms of the unknown $\sigma_{\min}(\tilde{P}), \sigma_{\min}(\tilde{P}^{-1})$. However this factor only depends on $d$ and not $T$.

%% file: content/misc.tex
\section{Extension to presence of control input}
\label{extensions}
	Here we sketch how to extend our results to the general case when we also have a control input, \textit{i.e.},
	\begin{equation}
	\label{control_eq}
	X_{t+1} = AX_t + BU_{t} + \eta_{t+1}
	\end{equation}
	Here $A, B$ are unknown but we can choose $U_t$. Pick independent vectors $\{U_t \sim \Nc(0, I)\}_{t=1}^T$. We can represent this as a variant of Eq.~\eqref{lti} as follows
	\begin{align*}
	\underbrace{\begin{bmatrix}
		X_{t+1} \\
		U_{t+1}
		\end{bmatrix}}_{\bar{X}_{t+1}} &= \underbrace{\begin{bmatrix}
		A & B \\
		0 & 0
		\end{bmatrix}}_{\bar{A}}\begin{bmatrix}
	X_{t} \\
	U_{t}
	\end{bmatrix} + \underbrace{\begin{bmatrix}
		\eta_{t+1} \\
		U_{t+1}
		\end{bmatrix}}_{\bar{\eta}_{t+1}}
	\end{align*}
	Since 
	\begin{align*}
	\text{det}\Bigg(\begin{bmatrix}
	A -\lambda I & B \\
	0 & -\lambda I
	\end{bmatrix}\Bigg) = 0
	\end{align*}
	holds when $\lambda$ equals an eigenvalue of $A$ or $0$. The eigenvalues of $\bar{A}$ are the same as $A$ with some additional eigenvalues that are zero. Now we can simply use Theorem~\ref{composite_result}.
	
\section{Extension to heavy tailed noise}
\label{noise_ind}
It is claimed in~\cite{faradonbeh2017finite} that techniques involving inequalities for subgaussian distributions cannot be used for the class of sub-Weibull distributions they consider. However, by bounding the noise process, as even \cite{faradonbeh2017finite} does, we can convert the heavy tailed process into a zero mean independent subgaussian one. In such a case our techniques can still be applied, and they incur only an extra logarithmic factor. We consider the class of distributions introduced in~\cite{faradonbeh2017finite} called sub--Weibull distribution. Let $\eta_{t, i}$ be the $i^{th}$ element of $\eta_t$ then $\eta_{t, i}$ has sub--Weibull distribution if
\begin{equation}
\label{sub_weibull}
\Pb(|\eta_{t, i} > y|) \leq b \exp{\Big(\frac{-y^{\alpha}}{m}\Big)}
\end{equation}
When $\alpha = 2$ it is subGaussian, $\alpha =1$ it is subExponential and $\alpha < 1$ it is subWeibull. Assume for now that $\eta_{t, i}$ has symmetric distribution. The extension to asymmetric case needs some computation in finding and is not discussed here. Consider the event 
$$\Wc(\delta) = \Bigg\{\max_{1 \leq t \leq T} ||\eta_{t}||_{\infty} \leq \nu_T(\delta)\Bigg\}$$
where $\nu_T(\delta) = \Big(m \log{\Big(\frac{bTd}{\delta}\Big)^{1/\alpha}}\Big)$. Then Proposition 3 in~\cite{faradonbeh2017finite} shows that $\Pb(\Wc(\delta)) \geq 1 - \delta$. Clearly because each $\{\eta_{t, i}\}_{t=1, i=1}^{t=T, i=d}$ are i.i.d and have symmetric distribution
\begin{equation}
\label{zero_mean}
\Ex[\eta_{t, i} | \Wc(\delta)] = \Ex[\eta_{t, i} | \{|\eta_{t, i}| \leq \nu_T(\delta)\}] = 0
\end{equation}
Then under $\Wc(\delta)$, $\eta_{t, i}$ has mean zero and $\{\eta_{t, i}\}_{t=1, i=1}^{t=T, i=d}$ are independent under the event $\Wc(\delta)$. Further since under $\Wc(\delta)$ these are bounded, they are also subGaussian. The subGaussian parameter or variance proxy $R^2 \leq \nu_T(\delta)^2$ which is logarithmic in $T$. This appears as simply a scaling factor in Theorem~\ref{selfnorm_main}, Proposition~\ref{selfnorm_bnd}. We can now use all our techniques from before.

\section{Optimality of Bound}
\label{optimal_bnd}

Let $A = a$ be 1-D system. Assume that $T \in T_{u}(\delta)$ (as in Table~\ref{notation}). Then $X_t, \eta_{t}$ are just numbers. Then let $E$ be the error, \textit{i.e.}, 
\begin{align*}
E &= (\sum_{t=1}^T x_t^2)^{-1}(\sum_{t=1}^T x_t \eta_{t+1}) \\
&= a^{-T}(\sum_{t=1}^T a^{-2T}x_t^2)^{-1}(\sum_{t=1}^T a^{-T}x_t \eta_{t+1})
\end{align*} 
In this section, we will show that the bound obtained for explosive systems is optimal in terms of $\delta$. Assume $\eta_t \sim \Nc(0, 1)$ i.i.d Gaussian. Let $S_T = \sum_{t=1}^T a^{-T}x_t \eta_{t+1}, U_T=\sum_{t=1}^T a^{-2T}x_t^2$. Now $E = a^{-T}U_T^{-1} S_T$ and $S_T$ has the following form
\begin{equation}
    2S_T = [\eta_{T+1}, \ldots, \eta_1] \underbrace{\begin{bmatrix}
    0 & a^{-T} & a^{-T+1}&\hdots & a^{-1} \\
    a^{-T} & 0 & a^{-T} & \hdots & a^{-2} \\
    \vdots & \ddots & \ddots & \ddots & \vdots \\
    \vdots & \vdots & \ddots & \ddots & \ddots \\
    a^{-1} & a^{-2} & a^{-3} & \hdots & 0
    \end{bmatrix}}_{=M} \underbrace{\begin{bmatrix}
    \eta_{T+1} \\
    \vdots \\
    \eta_1
    \end{bmatrix}}_{=\tilde{\eta}} \label{ST_form_dist}
\end{equation}

Define $F_T = \sum_{i=1}^T a^{-2i+2} (a^{-2T} x_T^2) = \frac{1-a^{-2T}}{1-a^{-2}}a^{-2T}x_T^2$. and $\sigma^2 = \text{Var}(a^{-2T}x_T^2)$. It is clear that $a^{-T}x_T$ is a Gaussian random variable. Note that $F_T, U_T$ are the same as Eq.~\eqref{ut_ft} and Section~\ref{explosive} when $A = a$. We can easily calculate $\sigma^2$
\[
a^{-2} \leq \sigma^2 \leq \frac{1}{a^2 - 1}
\]

Consider four events 
\begin{align*}
\Ec_1(\delta) &= \Bigg\{|U_T -F_T| \leq \frac{\delta^2 \sigma^2}{C} \vee \Big(\frac{C T^2 a^{-T}}{1-a^{-2}} + \Big(1 + \frac{1}{c} \log{\frac{1}{\delta}}\Big)\frac{Ta^{-2T}}{(1-a^{-2})}\Big) \Bigg\} ,\Ec_2(\delta) = \Bigg\{|S_T| \geq \frac{\delta}{-Ca^2\log{\delta}} \Bigg\} \\
\Ec_3(\delta) &= \Bigg\{0 \leq F_T \leq C_2\delta^2 \sigma^2\Bigg\} , \Ec_4(\delta) = \Bigg\{ 0 \leq U_T \leq  \Big((C_2 + 1/C)\delta^2 \sigma^2\Big) \vee \Big(\frac{C T^2 a^{-T}}{1-a^{-2}} + \Big(1 + \frac{1}{c} \log{\frac{1}{\delta}}\Big)\frac{Ta^{-2T}}{(1-a^{-2})}\Big) \Bigg\}
\end{align*}
From Eq.~\eqref{tight_error_cum} we have with probability at least $1-\frac{\delta}{2}$ that 
\begin{align*}
||U_T - F_T||_{2} &\underbrace{\leq}_{\text{Eq.}~\eqref{tight_error_cum}} \Bigg(4T^2 \sigma_1^2(A^{-\frac{(T+1)}{2}}) \text{tr}(\Gamma_T(A^{-1}))+ \Big(T + \frac{T}{c}\log{\frac{1}{\delta}}\Big)\sigma^2_1(A^{-T-1})\text{tr}(\Gamma_T(A^{-1}))\Bigg) \\
&\leq \frac{4T^2 a^{-T}}{1-a^{-2}} + \Big(1 + \frac{1}{c} \log{\frac{1}{\delta}}\Big)\frac{Ta^{-2T}}{(1-a^{-2})}
\end{align*}

Assume $\delta^2 \in (0, \frac{1}{128}]$ then
\begin{align*}
\Pb(\Ec_3(\delta)) &= \frac{2}{\sqrt{2 \pi} \sigma}\int_{2 \delta \sigma}^{16\delta \sigma} e^{-\frac{x^2}{2\sigma^2}}dx \\
&\geq  \frac{14 \delta}{\sqrt{2 \pi}} e^{-\frac{256 \delta^2}{2}} \\
&\geq \frac{14 \delta}{\sqrt{2 \pi}e} \geq 2 \delta
\end{align*}
Recall $T_u(\delta)$ is the set of  $T$ that satisfies Eq.~\eqref{t_exp_req} when $A = a$.
\subsection{$T \in T_u({\delta})$}
\label{t_tu}
For $T \in T_u(\delta)$ and from Eq.~\eqref{error_cum}, we have with probability at least $1 -\frac{\delta}{2}$ that
\begin{align*}
||U_T - F_T||_{2} &\leq \frac{4T^2 a^{-T}}{1-a^{-2}} + \frac{Ta^{-2T}}{\delta(1-a^{-2})} \underbrace{\leq}_{T \in T_u(\delta), \text{Eq.}~\eqref{t_exp_req}}  \frac{\phi_{\min}(a)^2 \psi(a)^2 \delta^2}{2\sigma_{\max}(P)^2} \leq \frac{C\delta^2}{(a^2-1)}
\end{align*}
The last inequality follows because for $1$-D systems $\phi_{\min}(A), \psi(A), \sigma_{\max}(P)$ are just constants, for example $P = 1, \phi_{\min}(a) = 1, \psi(a)^2 = C\sigma^2 \leq \frac{C}{a^2 - 1}$ which follows by definition. Note $T \in T_{u}(\delta)$ if and only if we have 
\[
{\delta^2 \sigma^2} >  \frac{C T^2 a^{-T}}{1-a^{-2}}
\]
Thus, $ \Pb(\Ec_1(\delta)) \geq 1 -\frac{\delta}{2}$. Clearly $\Ec_1(\delta) \cap \Ec_3(\delta) \implies \Ec_1(\delta) \cap \Ec_4(\delta)$ and  
$$\Ec_2(\delta) \cap \Ec_4(\delta) \implies \Big\{ |S_T| U_T^{-1} \geq \frac{C}{-\sigma^2 a^2 \delta \log{\delta}} \Big\}$$ 
We bound $\Pb(\Ec_2(\delta))$ in Section~\ref{char_fn} and Eq.~\eqref{ST_xT_lb}, which gives $\Pb(\Ec_2(\delta)) \geq 1 - \frac{\delta}{2}$ and then
\begin{align*}
\Pb(\Ec_1(\delta) \cap \Ec_2(\delta) \cap \Ec_4(\delta)) &\geq \Pb(\Ec_1(\delta) \cap \Ec_2(\delta) \cap \Ec_3(\delta)) \\
&\geq \Pb(\Ec_1(\delta)) + \Pb(\Ec_2(\delta) \cap \Ec_3(\delta)) - 1 \\
&\geq \Pb(\Ec_1(\delta)) + \Pb(\Ec_2(\delta)) +\Pb( \Ec_3(\delta)) - 2 \\
&\geq \frac{\delta}{2}
\end{align*}
Since $\Ec_2(\delta) \cap \Ec_4(\delta) \implies \{ |S_T| U_T^{-1} \geq \frac{C}{-\sigma^2 a^2 \delta \log{\delta} }\}$ when $T \in T_{u}(\delta)$ then 
$$  \Pb(\{ |S_T| U_T^{-1} \geq \frac{C}{-\sigma^2 a^2 \delta \log{\delta}}\}) \geq \frac{\delta}{2}$$
we have proved our claim that with probability at least $\delta$ we have that 
\begin{equation}
|E_T| \geq \Big(\frac{C}{-\sigma^2 a^2 \delta \log{\delta} }\Big)a^{-T} \geq \frac{C(1-a^{-2})}{-\delta \log{\delta}}a^{-T} \label{final_err_dist_lb}
\end{equation}
whenever $Ca^2T^2 a^{-T} \leq \delta^2$. 
\subsection{$T \not \in T_u(\delta)$}
\label{t_not_tu}
If $Ca^2T^2 a^{-T} > \delta^2$, then with probability at least $1-\frac{\delta}{2}$
$$|U_T-F_T| \leq \frac{CT^2 a^{-T}}{1-a^{-2}} + \Big(1 + \frac{1}{c} \log{\frac{1}{\delta}}\Big)\frac{Ta^{-2T}}{(1-a^{-2})}\Big)$$
and we have with probability at least $\delta$ that 
$$\Big\{ |S_T| U_T^{-1} \geq \frac{C(1-a^{-2}) \delta a^{T}}{-T^2 a^2 \log{\delta} + \Big(1 - \frac{\log{\delta}}{c} \Big)T a^{-T}}\Big\}$$
and we can conclude with probability at least $\delta$
\[
|E_T| \geq \frac{C(1-a^{-2}) \delta}{ -a^2 (\log{\delta})^3}
\]
where $Ca^2T^2 a^{-T} \geq \delta^2 \implies T \leq -\log{\delta}$.

\subsection{Comparison to existing bounds}
\label{comparison}
\begin{thm}[Theorem B.2~\cite{simchowitz2018learning}]
\label{b2}
Fix an $a_* \in \Rb$ and define $\Gamma_T = \sum_{t=1}a_{*}^{2t}$. Fix an alternative $a^{\prime} \in \{a_* - 2\epsilon, a_* + 2\epsilon\}$ and $\delta \in (0, 1/4)$. Then for any estimator $\hat{a}$
\[
\sup_{a \in \{a_*, a^{\prime}\}} \Pb(|\hat{a}(T) - a_*| \geq \epsilon) \geq \delta
\]
for any $T$ such that $T \Gamma_T \leq \frac{\log{(1/2\delta)}}{8 \epsilon^2}$.
\end{thm}
Note $\Gamma_T = \frac{a^{2T+2}-1}{a^2 -1 }$. Theorem~\ref{b2} suggests that for a given $T, \delta$ if $\epsilon \leq a^{-T}\sqrt{\frac{-C\log{\delta}}{T}}$ then $\Pb(|a_* - \hat{a}(T)| \geq \epsilon) \geq \delta$. However we show that whenever $Ca^2 T^2 a^{-T} \leq \delta^2$, we have that 
\[
\Pb\Big(|a_* - \hat{a}(T)| \geq a^{-T} \frac{C(1-a^{-2})}{-\delta \log{\delta}} \Big) \geq \delta
\]
Since $a^{-T}\sqrt{\frac{-C\log{\delta}}{T}} \leq a^{-T} \frac{C(1-a^{-2})}{-\delta \log{\delta}}$ our lower bound is tighter.
\begin{thm}[Theorem B.1~\cite{simchowitz2018learning}]
\label{b1}
Let $\epsilon \in (0, 1)$ and $\delta \in (0, 1/2)$. Then $\Pb(|\hat{a}(T) - a_*| \leq \epsilon) \geq 1 -\delta$ as long as 
\[
T \geq \max \Big\{\frac{8}{(|a_* - \epsilon|)^2 - 1} \log{\frac{2}{\delta}}, \frac{4 \log{\frac{1}{\epsilon}}}{\log{(|a_*| - \epsilon)}} + 8 \log{\frac{2}{\delta}}\Big\}
\]
\end{thm}
We now compare Eq.~\eqref{final_err_dist_lb} to the upper bound in Theorem~\ref{b1}. Eq.~\eqref{final_err_dist_lb} gives us that if 
\[
\epsilon \leq  \frac{C(1-a^{-2})}{-\delta \log{\delta}}a^{-T}
\]
we have with probability at least $\delta$ that $|E_T| \geq \epsilon$. This reduces to whenever
\begin{equation}
    \label{t_lb_req}
    T_{-} \leq \frac{\log{\frac{1}{\epsilon}}}{\log{a}} + \frac{\log{\frac{C(1-a^{-2})}{\delta}}}{\log{a}}
\end{equation}
we have with probability at least $\delta$ that $|E_T| \geq \epsilon$. We focus on the case $a_{*} > 1+\epsilon$ of Theorem~\ref{b1}. Let $a_{*} = 1 + \epsilon + \gamma$, then the bounds in Theorem~\ref{b1} indicate that whenever 
\[
T_{+} \geq \frac{8}{2\gamma + \gamma^2} \log{\frac{2}{\delta}} + \frac{4 \log{\frac{1}{\epsilon}}}{\log{(\gamma + 1)}} + \log{\frac{2}{\delta}}
\]
we have with probability at least $1 -\delta$ $|E_T| \leq \epsilon$. If $\gamma = o(\epsilon)$, then the requirement on $T$ reduces to 
\[
T_{+} \geq \frac{8}{o(\epsilon)} \log{\frac{2}{\delta}} + \frac{4 \log{\frac{1}{\epsilon}}}{o(\epsilon)} + \text{ smaller terms}
\]
By substituting $\log{a} \approx \epsilon$ in $T_{-}$ we note that $ T_{-} \leq T_{+}$. For the case when $\gamma = \Omega(\epsilon)$ for $T_{+}$ we get 
\[
T_{+} \geq \Big(\frac{8}{\Omega(\epsilon)} \vee 1\Big) \log{\frac{2}{\delta}} + \frac{4 \log{\frac{1}{\epsilon}}}{\log{(1+\Omega(\epsilon))}}  \approx  \underbrace{\Big(\frac{8}{\Omega(\epsilon)} \vee 1\Big)}_{\geq (\log{a})^{-1}} \log{\frac{2}{\delta}} + \frac{2 \log{\frac{1}{\epsilon}}}{\log{a}} 
\]
In either cases $T_- \leq T_+$.

\section{Distribution of $S_T$}
\label{char_fn}
Recall $S_T$ from Eq.~\eqref{ST_form_dist}. Since $\sum_{i, j}|M|_{i, j} \geq ||M||_{*}$ (the nuclear norm), we have that $||M||_{*} \leq \frac{2a^{-1}}{1-a^{-1}}$ and it is obvious that $||M||_2 \geq a^{-1}$. Since $M = U^{\top} \Lambda U$ (because it is symmetric) and $\eta_t$ are i.i.d Gaussian then $U \tilde{\eta}$ is also Gaussian with each of its entries being i.i.d Gaussian. This implies that $2S_T = \sum_{j=1}^{T+1}\lambda_j g_j^2$ where $\lambda_j$ are eigenvalues of $M$ and $g_j$ are i.i.d Gaussian with $\sum_j \lambda_j = 0, \sum_j |\lambda_j| \leq \frac{2a^{-1}}{1-a^{-1}}$. The characteristic function of $S_T$ is 
\[
\phi_{S_T}(t) = \prod_{j=1}^{T+1} \Big(\frac{1}{1-2it \lambda_j}\Big)^{1/2} =  \Big(\frac{1}{1 - 4t^2 (\sum_{l \neq j} \lambda_l \lambda_j) - i 8t^3 (\sum_{l \neq j \neq k} \lambda_l \lambda_j \lambda_k) + 16t^4 (\sum_{l \neq j \neq k \neq p} \lambda_l \lambda_j \lambda_k \lambda_p) \hdots}\Big)^{1/2}
\]
where the coefficient of $t$ vanishes because $\sum_{j=1}^{T+1} \lambda_j = 0$. Further since $\sum_{l \neq j} 2\lambda_l \lambda_j = - \sum_{j} \lambda_j^2$ we have
and 
\begin{align*}
    (\sum_{l \neq j \neq k \neq m} \lambda_l \lambda_j \lambda_k \lambda_m) &= \sum_{l} \lambda_l (\sum_{l \neq j \neq k \neq m} \lambda_j \lambda_k \lambda_m) = \sum_{l} \lambda_l (\sum_{l \neq j \neq k \neq m} \lambda_j \lambda_k \lambda_m + \sum_{l \neq p \neq m} \lambda_l \lambda_p \lambda_m - \sum_{l \neq p \neq m} \lambda_l \lambda_p \lambda_m) \\
    &= \sum_{l} \lambda_l (\sum_{j \neq k \neq m} \lambda_j \lambda_k \lambda_m - \sum_{l \neq p \neq m} \lambda_l \lambda_p \lambda_m - \sum_{l \neq m } \lambda_l^2  \lambda_m +  \sum_{l \neq m } \lambda_l^2  \lambda_m) \\
    &= \sum_{l} \lambda_l (- \lambda_l \sum_{ p \neq m} \lambda_p \lambda_m +  \sum_{l \neq m } \lambda_l^2  \lambda_m) = \frac{(\sum_l \lambda_l^2)^2}{2} - \sum_l \lambda_l^4 = {\frac{\text{tr}(M^2)^2}{2}} - \text{tr}(M^4)
\end{align*}
The coefficients of even powers of $t$ can be obtained in a similar fashion. Then recall by Levy's theorem that 
\[
f_{S_T}(x) = \int_{-\infty}^{\infty} e^{-itx} \phi_{S_T}(t) dt  \implies \sup_x f_{S_T}(x) \leq \int_{-\infty}^{\infty} |\phi_{S_T}(t)| dt \leq \int_{-\infty}^{\infty} \frac{1}{\sqrt{1 + c_1 t^2 + c_2 t^4 + \hdots}} dt
\]
Now whenever $c_k > 0$ (and not decaying asymptotically to zero) for some $k \geq 2$, we get $\sup_x f_{S_T}(x) \leq C $ for some universal constant $C$ and we can use Proposition~\ref{cont_rand} to get $\Pb(|S_T| \leq \delta) \leq C \delta$. But since that may not be always be true we can explicitly calculate the integral 
\begin{align*}
    f_{S_T}(x) &= \int_{-\infty}^{\infty} e^{-itx} \phi_{S_T}(t) dt \approx \underbrace{\int_{-\infty}^{\infty} \frac{e^{itx}}{\sqrt{1 + 2a^{-2}t^2}} dt}_{\text{Modified Bessel Function of the Second Kind} } \\
    \int_{-\delta}^{\delta} f_{S_T}(x) dx &=  \int_{-\delta}^{\delta}\int_{-\infty}^{\infty} \frac{e^{itx}}{\sqrt{1 + 2a^{-2}t^2}} dt dx = 2\int_{-\infty}^{\infty} \int_{-\delta}^{\delta} \frac{\cos(tx)}{\sqrt{1 + 2a^{-2}t^2}} dx dt \\
    &= C\delta \int_{0}^{\infty} \frac{\sin(t\delta)}{\delta t\sqrt{1 + 2a^{-2}t^2}} dt = C\delta \int_{0}^{\delta} \frac{\sin(t\delta)}{\delta t\sqrt{1 + 2a^{-2}t^2}} dt + C\delta \int_{\delta}^{\infty} \frac{\sin(t\delta)}{\delta t\sqrt{1 + 2a^{-2}t^2}} dt \\
    &\leq C \delta^2 - C a\delta \log(\delta)
\end{align*}
Thus
\[
\Pb(|S_T| \leq \delta) \leq -Ca  \delta \log{\delta}
\]
and replacing $\delta \rightarrow \frac{-C\delta}{2a\log{\delta}}$ we get 
\begin{equation}
  \Pb\Big(|S_T| \leq \frac{-C\delta}{a\log{\delta}}\Big) \leq \frac{\delta}{2}  \label{ST_xT_lb}
\end{equation}
\section{Lemma B}
\label{lemmab}
Let the characteristic and minimal polynomial be $\chi(t), \mu(t)$ respectively. 
\[
\chi(t) = \prod_{i=1}^k (t-\lambda_i)^{a_i}, \mu(t) = \prod_{i=1}^k (t-\lambda_i)^{b_i}
\]
where $b_i \leq a_i$. $b_i$ is the size of the largest Jordan block corresponding to $\lambda_i$ in the Jordan normal form. $a_i$ sum of size of all Jordan blocks corresponding to $\lambda_i$. Now, if $\chi(t) = \mu(t)$ then $a_i=b_i$, \textit{i.e.}, there is only Jordan block corresponding to each $\lambda_i$. On the other if there is only one Jordan block (geometric multiplicity $=1$) corresponding to each eigenvalue $\implies a_i=b_i$ and $\chi(t) = \mu(t)$. 
\section{Inconsistency of explosive systems}
\label{inconsistent}
Recall that $A = a I$ where $a \geq 1.1$ and 
\[
\begin{bmatrix}
\x_{t+1} \\
\y_{t+1}
\end{bmatrix} = A \begin{bmatrix}
\x_{t} \\
\y_{t}
\end{bmatrix} + \begin{bmatrix}
\n_{t+1} \\
\w_{t+1}
\end{bmatrix} 
\]
Since $A$ is scaled identity we have that $\x_{t} = \sum_{t=1}^{T}a^{T-t} \n_t, \y_t = \sum_{t=1}^{T}a^{T-t} \w_t$. The scaled sample covariance matrix $a^{-2T}Y_T = a^{-2T}\sum_{t=1}^T X_t X_t^{\top}$ is of the following form
\begin{align}
   a^{-2T}Y_T &= \begin{bmatrix}
   a^{-2T}\sum_{t=1}^T (\x_t)^2 & a^{-2T} \sum_{t=1}^T \x_t \y_t \\
   a^{-2T} \sum_{t=1}^T \x_t \y_t & a^{-2T}\sum_{t=1}^T (\y_t)^2    \end{bmatrix} \label{sample_cov} 
\end{align}
Define $a^{-T}X_T = Z_T$ with $Z_T^{(i)}$ corresponding to appropriate coordinates, and recall that $Z^{(i)}_T$ is a Gaussian random variable with variance  in $(a^{-2}, \frac{a^{-2}}{1-a^{-2}})$ and each $a^{-T} X_t = \la a^{-T} X_t, Z_T\ra Z_T + \la a^{-T} X_t, \ztp \ra \ztp$. This implies 
\begin{align*}
    a^{-2T}\sum_{t=1}^T X_t X_t^{\top} &=  \sum_{t=1}^T  (\underbrace{a^{-T}\la X_t, Z_T\ra}_{=\alpha_t})^2 Z_T Z_T^{\top} + \sum_{t=1}^T a^{-2T} \la  X_t, Z_T \ra  \la X_t, \ztp \ra  Z_T (\ztp)^{\top} \\
    &+ \sum_{t=1}^T \underbrace{\la a^{-T}X_t, Z_T\ra}_{=\alpha_t} \underbrace{\la a^{-T}X_t, \ztp \ra}_{=\beta_t} \ztp  Z_T^{\top} + \sum_{t=1}^T  (\underbrace{a^{-T}\la X_t, \ztp \ra}_{=\beta_t})^2  \ztp (\ztp)^{\top} \\
    &= \underbrace{||\alpha||^2 Z_T Z_T^{\top} + ||\beta||^2  \ztp (\ztp)^{\top}}_{=M} + \la \alpha, \beta \ra (\ztp  Z_T^{\top} + Z_T (\ztp)^{\top}) \\
    &= M + \underbrace{\la \alpha, \beta \ra [Z_T \ztp]}_{=U} \underbrace{\begin{bmatrix}
    0 & 1 \\
    1 & 0
    \end{bmatrix}}_{=C}\underbrace{\begin{bmatrix}
    Z_T^{\top} \\
    (\ztp)^{\top}
    \end{bmatrix}}_{=V}
\end{align*}
By using Woodbury's matrix identity and since $M^{-1} = ||\alpha||^{-2} Z_T Z_T^{\top} + ||\beta||^{-2} \ztp (\ztp)^{\top}, C=C^{-1}$ we get 
\begin{align*}
   (a^{-2T}\sum_{t=1}^T X_t X_t^{\top})^{-1} &= M^{-1} - \la \alpha, \beta \ra M^{-1} U(C + \la \alpha, \beta \ra U^{\top} M^{-1} U)^{-1} U^{\top} M^{-1} \\
   &= M^{-1} - \la \alpha, \beta \ra [||\alpha||^{-2}Z_T \hspace{2mm} ||\beta||^{-2}\ztp] \Big(\begin{bmatrix}
   \la \alpha, \beta \ra ||\alpha||^{-2} & 1 \\
   1 & ||\beta||^{-2} \la \alpha, \beta \ra
   \end{bmatrix} \Big)^{-1}\begin{bmatrix}
    ||\alpha||^{-2}Z_T^{\top} \\
    ||\beta||^{-2}(\ztp)^{\top}
    \end{bmatrix}
\end{align*}
Then the error term is
\begin{align*}
    \hat{A}_o - A_o &= \Big(\sum_{t=1}^{T}a^{-2T} \eta_{t+1} X_t^{\prime}\Big) (a^{-2T}\sum_{t=1}^T X_t X_t^{\top})^{-1} \\
    &= \Big(\sum_{t=1}^T \la a^{-T} X_t, Z_T \ra a^{-T}\eta_{t+1} Z_T^{\prime} + \sum_{t=1}^T \la a^{-T} X_t, \ztp \ra a^{-T}\eta_{t+1} (\ztp)^{\prime}\Big) (a^{-2T}\sum_{t=1}^T X_t X_t^{\top})^{-1}
\end{align*}
We now check the projection of $Z_T, \ztp$ on $(a^{-2T}\sum_{t=1}^T X_t X_t^{\top})^{-1}$
\begin{align}
    Z_T^{\top} (a^{-2T}\sum_{t=1}^T X_t X_t^{\top})^{-1} &= ||\alpha||^{-2} Z_T^{\top} - \la \alpha, \beta \ra [||\alpha||^{-2} \hspace{2mm} 0] \Big(\begin{bmatrix}
   \la \alpha, \beta \ra ||\alpha||^{-2} & 1 \\
   1 &  \la \alpha, \beta \ra ||\beta||^{-2}
   \end{bmatrix} \Big)^{-1}\begin{bmatrix}
    ||\alpha||^{-2}Z_T^{\top} \\
    ||\beta||^{-2}(\ztp)^{\top}
    \end{bmatrix} \nonumber \\
    &= \frac{-||\alpha||^{-2}Z_T^{\top} + \la \alpha, \beta \ra ||\alpha||^{-2}||\beta||^{-2}(\ztp)^{\top}}{\la \alpha, \beta \ra^2 ||\alpha||^{-2} ||\beta||^{-2} - 1} \label{zt_proj}\\
    (\ztp)^{\top} (a^{-2T}\sum_{t=1}^T X_t X_t^{\top})^{-1} &= ||\beta||^{-2} (\ztp)^{\top} - \la \alpha, \beta \ra[0 \hspace{2mm} ||\beta||^{-2}] \Big(\begin{bmatrix}
   \la \alpha, \beta \ra ||\alpha||^{-2} & 1 \\
   1 & \la \alpha, \beta \ra ||\beta||^{-2}
   \end{bmatrix} \Big)^{-1}\begin{bmatrix}
    ||\alpha||^{-2}Z_T^{\top} \\
    ||\beta||^{-2}(\ztp)^{\top}
    \end{bmatrix} \nonumber\\
    &= \frac{-||\beta||^{-2}(\ztp)^{\top} +  \la \alpha, \beta \ra ||\alpha||^{-2}||\beta||^{-2}Z_T^{\top}}{\la \alpha, \beta \ra^2 ||\alpha||^{-2} ||\beta||^{-2} - 1} \label{ztp_proj}
\end{align}
We will show that with high probability $||\alpha||^{-2} = \Theta(1), ||\beta||^{-2} = \Omega(a^{2T}), \la \alpha, \beta \ra = O(a^{-T})$ as a result Eq.~\eqref{zt_proj} is $\Omega(a^{T})$ and Eq.~\eqref{ztp_proj} is $\Omega(a^{2T})$. Note that $\ztp = \begin{bmatrix}
\zy_T \\
-\zx_T
\end{bmatrix}$ where we have ignored the scaling (as these will be of constant order with high probability). First taking a closer look at $\alpha_t = a^{-2T} \x_t \zx_T + a^{-2T} \y_t \zy_T$ reveals the following behaviour
\begin{align*}
    a^{-2T} \x_{T-1} \zx_T &= a^{-1} (\zx_T)^2 - a^{-T-1} \zx_T \n_T \\
    \alpha_{T-1} &= a^{-1}( (\zx_T)^2 +(\zy_T)^2) - a^{-T-1} (\zx_T \n_T + \zy_T \w_T) \\
    a^{-2T} \x_{T-2} \zx_T &= a^{-2} (\zx_T)^2 - a^{-T-1} \zx_{T-1} \n_T - a^{-T-2} \zx_T \n_T \\
    \alpha_{T-2} &= a^{-2}( (\zx_T)^2 +(\zy_T)^2) - a^{-T-1} (\zx_{T-1} \n_T + \zy_{T-1} \w_T) - a^{-T-2} (\zx_T \n_T + \zy_T \w_T)
\end{align*}
Since $\zx_T$ is a Gaussian random variable with bounded variance, we see that $\alpha_t$ decays exponentially as $t$ decreases (up to some $a^{-T}$ additive terms). In a similar fashion one can show that $\sum_{t=1}^{T} \alpha_t^2 =  \frac{1-a^{-2T}}{1-a^{-2}}((\zx_T)^2 +(\zy_T)^2)^2 + O(T^2a^{-T})$ with high probability. Clearly $||\alpha||^{-2} = \Theta(1)$ with high probability. For $\beta$, note that $\zy_T$ is independent of $\x_t$ and observe that $\{a^T \beta_t \}_{t=1}^{T-1}$ are non--decaying and non--trivial random variables. Specifically these are subexponential random variables with $||\cdot||_{\psi_1}$ norm as $||a^{T} \beta_t ||_{\psi_1} = Ca^{-1}$. Here $||\cdot||_{\psi_1}$ norm is the same Definition 2.7.5 in~\cite{vershynin2018high}. To see this consider for example $t=T-1, T-2$, then 
\begin{align}
a^{T}\beta_{T-1} &= \la X_{T-1}, \ztp \ra = \x_{T-1} \zy_{T} - \y_{T-1} \zx_{T} = a^{-1}( \w_T \zx_T -\n_T \zy_T) \nonumber\\ 
a^{T} \beta_{T-2} &= \la X_{T-1}, \ztp \ra = \x_{T-1} \zy_{T} - \y_{T-1} \zx_{T} = a^{-1}(( \w_{T-1} + a^{-1} \w_{T})\zx_T - ( \n_{T-1} + a^{-1} \n_{T})\zy_T ) \label{scale_beta}
\end{align} 
Clearly, $a^{2T}|| \beta||_2^2 = \Omega(1)$ and $a^{2T}|| \beta||_2^2 = O(T)$ with high probability. Recall the error term 
\begin{align}
    \hat{A}_o - A_o &= \Big(\sum_{t=1}^{T}a^{-2T} \eta_{t+1} X_t^{\prime}\Big) (a^{-2T}\sum_{t=1}^T X_t X_t^{\top})^{-1} \nonumber\\
    &= \Big(\sum_{t=1}^T \la a^{-T} X_t, Z_T \ra a^{-T}\eta_{t+1} Z_T^{\prime} + \sum_{t=1}^T \la a^{-T} X_t, \ztp \ra a^{-T}\eta_{t+1} (\ztp)^{\prime}\Big) (a^{-2T}\sum_{t=1}^T X_t X_t^{\top})^{-1} \nonumber\\
    (\hat{A}_o - A_o)\ztp &= (\sum_{t=1}^T \la a^{-T} X_t, Z_T \ra a^{-T}\eta_{t+1} Z_T^{\prime}) (a^{-2T}\sum_{t=1}^T X_t X_t^{\top})^{-1} \ztp \nonumber \\
    &+ (\sum_{t=1}^T \la a^{-T} X_t, \ztp \ra a^{-T}\eta_{t+1} (\ztp)^{\prime}) (a^{-2T}\sum_{t=1}^T X_t X_t^{\top})^{-1} \ztp \nonumber \\
    &=  \frac{\la \alpha, \beta \ra ||\alpha||^{-2}||\beta||^{-2}}{\la \alpha, \beta \ra^2 ||\alpha||^{-2} ||\beta||^{-2} - 1}\sum_{t=1}^T \la a^{-T} X_t, Z_T \ra a^{-T}\eta_{t+1} -  \frac{-||\beta||^{-2}}{\la \alpha, \beta \ra^2 ||\alpha||^{-2} ||\beta||^{-2} - 1}\sum_{t=1}^T \la a^{-T} X_t, \ztp \ra a^{-T}\eta_{t+1}  \nonumber \\
    &= \frac{||\alpha||^{-2}||a^{T}\beta||^{-2}}{\la \alpha, a^T \beta \ra^2 ||\alpha||^{-2} ||a^T \beta||^{-2} - 1} \Big(\sum_{t=1}^T (\la \alpha, a^T\beta \ra \alpha_t - a^T \beta_t ||\alpha||^2) \eta_{t+1} \Big) = \gamma_T\label{error_dist}
\end{align}
Observe the term $a^T\beta_t||\alpha||^2 \eta_{t+1}$
\begin{align*}
 &a^T\beta_t||\alpha||^2 \eta_{t+1}  = ||\alpha||^2\begin{bmatrix}
 (a^{-1} (\w_{t+1} \zx_T - \n_{t+1} \zy_T) + a^{-2}(\w_{t+2} \zx_T - \n_{t+2} \zy_T) + \hdots) \n_{t+1} \\
(a^{-1} (\w_{t+1} \zx_T - \n_{t+1} \zy_T) + a^{-2}(\w_{t+2} \zx_T - \n_{t+2} \zy_T) + \hdots) \w_{t+1}
 \end{bmatrix} \\
 &= ||\alpha||^2\begin{bmatrix}
 a^{-1} (\w_{t+1}\n_{t+1} \zx_T - (\n_{t+1})^2 \zy_T) + (a^{-2}(\w_{t+2} \zx_T - \n_{t+2} \zy_T) + \hdots) \n_{t+1} \\
a^{-1} ((\w_{t+1})^2 \zx_T - \w_{t+1}\n_{t+1}\zy_T) + (a^{-2}(\w_{t+2} \zx_T - \n_{t+2} \zy_T) + \hdots) \w_{t+1}
 \end{bmatrix}\\
&\sum_{t=1}^T a^T\beta_t||\alpha||^2 \eta_{t+1} = a^{-1}||\alpha||^2 \Big(\underbrace{\begin{bmatrix}
-\sum_{t=1}^T (\n_{t+1})^2 \zy_T \\
\sum_{t=1}^T (\w_{t+1})^2 \zx_T
\end{bmatrix}}_{=\Theta(T)} + \sum_{t=1}^T \begin{bmatrix}
 \w_{t+1}\n_{t+1} \zx_T  + (a^{-1}(\w_{t+2} \zx_T - \n_{t+2} \zy_T) + \hdots) \n_{t+1} \\
 \w_{t+1}\n_{t+1}\zy_T + (a^{-1}(\w_{t+2} \zx_T - \n_{t+2} \zy_T) + \hdots) \w_{t+1}
 \end{bmatrix}\Big) \\
 &= a^{-1}||\alpha||^2 \Big(\Theta(T) \\
 &+\underbrace{\sum_{t=1}^T \begin{bmatrix}
 \w_t \n_t \zx_T + a^{-1} \w_{t+1} \n_t \zx_T + a^{-2} \w_{t+2} \n_t \zx_T + \ldots -  a^{-1} \n_{t+1} \n_t \zy_T - a^{-2} \n_{t+2} \n_t \zy_T - \ldots\\
 \w_t \n_t \zy_T + a^{-1} \w_{t+1} \w_t \zx_T + a^{-2} \w_{t+2} \w_t \zx_T + \ldots -  a^{-1} \n_{t+1} \w_t \zy_T - a^{-2} \n_{t+2} \w_t \zx_T - \ldots
 \end{bmatrix}}_{=O(\sqrt{T} \log{\frac{T}{\delta}})}\Big)
\end{align*}
The $O(\sqrt{T} \log{T})$ follows by applying Hanson-Wright inequality to each of $a^{-j}\sum_{t=1}^T \w_{t+j} \n_t$ terms where we get with probability at least $1-\delta/T$ that $a^{-j}\sum_{t=1}^T \w_{t+j} \n_t \leq c a^{-j} O(\sqrt{T} \log{\frac{T}{\delta}})$. Therefore simultaneously for all $j \leq T$ we have with probability at least $1 -\delta$ (using union bound) that $a^{-j}\sum_{t=1}^T \w_{t+j} \n_t \leq c a^{-j} O(\sqrt{T}\log{\frac{T}{\delta}}) \implies \sum_{j=1}^T a^{-j}\sum_{t=1}^T \w_{t+j} \n_t \leq O(\sqrt{T}\log{\frac{T}{\delta}})$. Plugging this in Eq.~\eqref{error_dist} we get that 
\begin{align*}
  \gamma_T &= \frac{||\alpha||^{-2}||a^{T}\beta||^{-2}}{\la \alpha, a^T \beta \ra^2 ||\alpha||^{-2} ||a^T \beta||^{-2} - 1} \Big(\underbrace{\sum_{t=1}^T (\la \alpha, a^T\beta \ra \alpha_t}_{=O(\sqrt{T})} - \underbrace{ a^T \beta_t ||\alpha||^2) \eta_{t+1}}_{=\Theta(T)} \Big)  
\end{align*}
Clearly then $\gamma_T$ in Eq.~\eqref{error_dist} satisfies a non--trivial pdf, \textit{i.e.}, error does not decay to zero.

Another interesting observation is that $\sum_{t=1}^T a^{-2T} \eta_{t+1} X_t^{\top}$ decays $O(a^{-T})$ with high probability, however the error is a non--decaying random variable. This immediately gives us that 
\begin{prop}
\label{condition_number}
The sample covariance matrix $\sum_{t=1}^T  X_t X_t^{\top}$ has the following singular values 
\[
\sigma_1(\sum_{t=1}^T  X_t X_t^{\top}) = \Theta(a^{2T}), \sigma_2(\sum_{t=1}^T  X_t X_t^{\top}) = O(\sqrt{T}a^{T})
\]
\end{prop}
\begin{proof}
The largest singular values of $\sum_{t=1}^T X_t X_t^{\top} = \Theta(a^{2T})$ this follows because $$||\sum_{t=1}^T a^{-2T} X_t X_t^{\top} - \frac{1-a^{-2T}}{1-a^{-2}} Z_T Z_T^{\top}||_2 \leq O(a^{-T})$$ with high probability, which follows from the claims of Eq.~\eqref{zt_form}, \eqref{ut_ft} in Theorem~\ref{main_result} and discussion in Section~\ref{explosive}. The second claim follows because $\sum_{t=1}^T a^{-2T} \eta_{t+1} X_t^{\top}$ decays $\Omega(a^{-T})$ with high probability. To see this $$\sum_{t=1}^T a^{-2T} \eta_{t+1} X_t^{\top} \leq a^{-T} \sqrt{\sum_{t=1}^T \eta_t^{\prime} \eta_t}\sqrt{\sum_{t=1}^T a^{-2T}X_{t}^{\prime} X_t } \approx \sqrt{T} a^{-T} $$
The $\sqrt{T}$ factor can be removed by similar arguments as above. However the identification error is a random variable which implies that $\sigma_2(\sum_{t=1}^T a^{-2T} X_t X_t^{\top}) = O(\sqrt{T} a^{-T})$.
\end{proof}

%% file: ms.bbl
\newcommand{\etalchar}[1]{$^{#1}$}
\begin{thebibliography}{SMT{\etalchar{+}}18}

\bibitem[ACB13]{accikmecse2013lossless}
Beh{\c{c}}et A{\c{c}}{\i}kme{\c{s}}e, John~M Carson, and Lars Blackmore.
\newblock Lossless convexification of nonconvex control bound and pointing
  constraints of the soft landing optimal control problem.
\newblock {\em IEEE Transactions on Control Systems Technology},
  21(6):2104--2113, 2013.

\bibitem[AYPS11]{abbasi2011improved}
Yasin Abbasi-Yadkori, D{\'a}vid P{\'a}l, and Csaba Szepesv{\'a}ri.
\newblock Improved algorithms for linear stochastic bandits.
\newblock In {\em Advances in Neural Information Processing Systems}, pages
  2312--2320, 2011.

\bibitem[CW02]{campi2002finite}
Marco~C Campi and Erik Weyer.
\newblock Finite sample properties of system identification methods.
\newblock {\em IEEE Transactions on Automatic Control}, 47(8):1329--1334, 2002.

\bibitem[Erx94]{erxiong1994691}
Jiang Erxiong.
\newblock Bounds for the smallest singular value of a jordan block with an
  application to eigenvalue perturbation.
\newblock {\em Linear Algebra and its Applications}, 197-198:691 -- 707, 1994.

\bibitem[FTM17]{faradonbeh2017finite}
Mohamad Kazem~Shirani Faradonbeh, Ambuj Tewari, and George Michailidis.
\newblock Finite time identification in unstable linear systems.
\newblock {\em arXiv preprint arXiv:1710.01852}, 2017.

\bibitem[HMR16]{hardt2016gradient}
Moritz Hardt, Tengyu Ma, and Benjamin Recht.
\newblock Gradient descent learns linear dynamical systems.
\newblock {\em arXiv preprint arXiv:1609.05191}, 2016.

\bibitem[IL11]{ipsen2011determinant}
Ilse~CF Ipsen and Dean~J Lee.
\newblock Determinant approximations.
\newblock {\em arXiv preprint arXiv:1105.0437}, 2011.

\bibitem[KM18]{forecasting_mohri}
Vitaly Kuznetsov and Mehryar Mohri.
\newblock Theory and algorithms for forecasting time series.
\newblock {\em CoRR}, abs/1803.05814, 2018.

\bibitem[Liu05]{Liu2005}
Jianzhou Liu.
\newblock {\em Eigenvalue and Singular Value Inequalities of Schur
  Complements}, pages 47--82.
\newblock Springer US, Boston, MA, 2005.

\bibitem[LW83]{lai1983asymptotic}
TL~Lai and CZ~Wei.
\newblock Asymptotic properties of general autoregressive models and strong
  consistency of least-squares estimates of their parameters.
\newblock {\em Journal of multivariate analysis}, 13(1):1--23, 1983.

\bibitem[Nie08]{nielsen2008singular}
Bent Nielsen.
\newblock Singular vector autoregressions with deterministic terms: Strong
  consistency and lag order determination.
\newblock 2008.

\bibitem[OO18]{oymak2018non}
Samet Oymak and Necmiye Ozay.
\newblock Non-asymptotic identification of lti systems from a single
  trajectory.
\newblock {\em arXiv preprint arXiv:1806.05722}, 2018.

\bibitem[PLS08]{pena2008self}
Victor~H Pe{\~n}a, Tze~Leung Lai, and Qi-Man Shao.
\newblock {\em Self-normalized processes: Limit theory and Statistical
  Applications}.
\newblock Springer Science \& Business Media, 2008.

\bibitem[PM13]{phillips2013inconsistent}
Peter~CB Phillips and Tassos Magdalinos.
\newblock Inconsistent var regression with common explosive roots.
\newblock {\em Econometric Theory}, 29(4):808--837, 2013.

\bibitem[SMT{\etalchar{+}}18]{simchowitz2018learning}
Max Simchowitz, Horia Mania, Stephen Tu, Michael~I Jordan, and Benjamin Recht.
\newblock Learning without mixing: Towards a sharp analysis of linear system
  identification.
\newblock {\em arXiv preprint arXiv:1802.08334}, 2018.

\bibitem[Ver10]{vershynin2010introduction}
Roman Vershynin.
\newblock Introduction to the non-asymptotic analysis of random matrices.
\newblock {\em arXiv preprint arXiv:1011.3027}, 2010.

\bibitem[Ver18]{vershynin2018high}
Roman Vershynin.
\newblock High-dimensional probability: An introduction with applications in
  data science.
\newblock 47, 2018.

\bibitem[VK06]{vidyasagar2006learning}
Mathukumalli Vidyasagar and Rajeeva~L Karandikar.
\newblock A learning theory approach to system identification and stochastic
  adaptive control.
\newblock In {\em Probabilistic and randomized methods for design under
  uncertainty}, pages 265--302. Springer, 2006.

\bibitem[Yu94]{yu1994rates}
Bin Yu.
\newblock Rates of convergence for empirical processes of stationary mixing
  sequences.
\newblock {\em The Annals of Probability}, pages 94--116, 1994.

\end{thebibliography}
